\documentclass{theoretics}

\title{\texorpdfstring{Elementary first-order model checking for\newline sparse graphs}{Elementary first-order model checking for sparse graphs}}

\ThCSauthor[1,2]{Jakub Gajarský}{gajarsky@fi.muni.cz}[0000-0002-4761-3432]
\ThCSauthor[1]{Michał Pilipczuk}{mp.pilipczuk@uw.edu.pl}[0000-0001-7891-1988]
\ThCSauthor[3]{Marek Sokołowski}{marek.sokolowski@inf.ethz.ch}[0000-0001-8309-0141]
\ThCSauthor[4]{Giannos Stamoulis}{stamoulis@irif.fr}[0000-0002-4175-7793]
\ThCSauthor[1]{Szymon Toruńczyk}{szymtor@uw.edu.pl}[0000-0002-1130-9033]
\ThCSaffil[1]{University of Warsaw, Institute of Informatics, Poland}
\ThCSaffil[2]{Masaryk University, Faculty of Informatics, Czech Republic}
\ThCSaffil[3]{ETH Zürich, Switzerland}
\ThCSaffil[4]{Université Paris Cité, CNRS, IRIF, F-75013, Paris, France}

\ThCSthanks{The work of Jakub Gajarský, Michał Pilipczuk, Marek Sokołowski, Giannos Stamoulis, and Szymon Toruńczyk on this manuscript is a part of a project that has received funding from the European Research Council (ERC), grant agreement No 948057 --- BOBR. Further work of Jakub Gajarský, after the end of his employment in BOBR, was also supported by the Polish National Science Centre SONATA-18 grant number  2022/47/D/ST6/03421. Further work of Szymon Toruńczyk, after the end of his employment in BOBR, was also supported by ERC grant BUKA (No. 101126229). A majority of work on this manuscript was done while Giannos Stamoulis  was employed in BOBR and affiliated to the University of Warsaw. Giannos Stamoulis also acknowledges support by the French-German Collaboration ANR/DFG Project UTMA (ANR-20-CE92-0027).\\
A preliminary version of this paper appeared in the proceedings of LICS 2024~\cite{DBLP:conf/lics/GajarskyP0ST24}.}
\ThCSshortnames{J.~Gajarský, M.~Pilipczuk, M.~Sokołowski, G.~Stamoulis and S.~Toruńczyk}
\ThCSshorttitle{Elementary first-order model checking for sparse graphs}
\ThCSyear{2026}
\ThCSarticlenum{19}
\ThCSreceived{May 28, 2025}
\ThCSrevised{Jun 8, 2026}
\ThCSaccepted{Jul 5, 2026}
\ThCSpublished{Nov 25, 2026}
\ThCSdoicreatedtrue
\ThCSkeywords{First-order model checking, sparsity, bounded expansion, elementary model-checking, transductions}

\usepackage[normalem]{ulem}
\usepackage{mathtools}
\usepackage{xspace}
\usepackage{textpos}

\newcommand{\wins}{\mathit{Rank}}
\newcommand{\winsp}{\mathit{RankDown}}

\newcommand{\Oh}{\mathcal{O}}

\newcommand{\C}{\mathcal{C}}
\newcommand{\D}{\mathcal{D}}
\renewcommand{\phi}{\varphi}
\newcommand{\CC}{\C}
\renewcommand{\cal}{\mathcal}
\renewcommand{\subset}{\subseteq}
\newcommand{\from}{\colon}
\newcommand{\set}[1]{\{#1\}}
\renewcommand{\le}{\leqslant}
\renewcommand{\ge}{\geqslant}

\newcommand{\wh}{\widehat}
\newcommand{\DD}{\cal D}
\newcommand{\dist}{\mathrm{dist}}

\newcommand{\CCC}{\C} 
\newcommand{\tup}[1]{{\bar{#1}}} 
\newcommand{\str}{\mathbb}
\newcommand{\Oof}{\mathcal O}
\newcommand{\setof}[2]{\set{#1 \mid #2}}
\newcommand{\tower}{\mathsf{tower}}

\newcommand{\AW}{\mathsf{AW}}
\newcommand{\FPT}{\mathsf{FPT}}
\newcommand{\DKT}{{Dvo\v r\'ak, Kr\'al, and Thomas}\xspace}
\newcommand{\topminors}{\textup{TopMinors}}
\DeclareMathOperator{\poly}{poly}

\renewcommand{\leq}{\leqslant}
\renewcommand{\le}{\leqslant}
\renewcommand{\geq}{\geqslant}
\renewcommand{\ge}{\geqslant}

\renewcommand{\setminus}{-}

\newcommand{\N}{\mathbb{N}}

\begin{document}
\maketitle
\begin{abstract}
 It is known that for subgraph-closed graph classes the first-order model checking problem is fixed-parameter tractable if and only if the class is nowhere dense [Grohe, Kreutzer, Siebertz, 2017]. However, the dependency on the formula size is non-elementary, and in fact, this is unavoidable even for the class of all trees [Frick and Grohe, 2004]. On the other hand, it is known that the dependency is elementary for classes of bounded degree [Frick and Grohe, 2004] as well as for classes of bounded pathwidth [Lampis, 2026]. In this paper we generalise these results and~almost completely characterise subgraph-closed graph classes for which the model checking problem is fixed-parameter tractable with an elementary dependency on the formula size. Those are the graph classes for which there exists a number $d$ such that for every~$r$, some tree of depth $d$ and size bounded by an elementary function of $r$ is avoided as an $({\leq} r)$-subdivision in all graphs in the class. In particular, this implies that if the class in question excludes a fixed tree as a topological minor, then first-order model checking for graphs in the class is fixed-parameter tractable with an elementary dependency on the formula size.
\end{abstract}

\section{Introduction}
A major theme in algorithmic graph theory and in finite model theory explores 
 the tractability frontier of the model checking problem on restricted graph classes.
In the model checking problem for first-order logic,
the goal is to decide if a given graph $G$ satisfies a given first-order formula $\phi$. In the setting of parameterised complexity, the length of the formula $\phi$ is treated as a parameter,
and then the aim is to devise an algorithm that runs in time 
\begin{align}\label{eq:fpt}
    f(|\phi|)\cdot |G|^c,    
\end{align}
for some constant $c$ and computable function $f$.
Since this problem is believed to be intractable in general (it is $\AW[\ast]$-complete), the objective is to characterise those graph classes $\CC$ for which the problem becomes \emph{fixed parameter tractable} (fpt) when the inputs are restricted to graphs from~$\CC$,
that is, admits an algorithm as in \eqref{eq:fpt} for $G\in\CC$.
Initiated\footnote{While model checking results for first-order logic (e.g.,~\cite{see96}) and monadic second-order logic (e.g.,~\cite{Courcelle90}) already existed in the literature, these were not presented under this systematic view, in the framework of parameterised complexity.} by Flum and Grohe~\cite{flum-grohe},
this direction of research reveals a deep interplay between structural graph theory,
logic, and algorithms, and has enjoyed major successes.

Results expressing that the model checking problem for 
formulas of some logic is fpt on some class of graphs are sometimes called \emph{algorithmic meta-theorems}.
The prime example of such a result is the theorem of Courcelle~\cite{Courcelle90}, 
which states that model checking  \emph{monadic second-order} (MSO) formulas is fpt on all classes of bounded treewidth. In particular, every problem expressible in MSO can be solved in linear time on all graphs of treewidth bounded by a fixed constant.
However, the hidden constants in such an algorithm
are often astronomical. Indeed, in Courcelle's result, 
  the function $f$ in \eqref{eq:fpt} is not even bounded by any elementary function. Equivalently, it is not bounded by any function of the form
  $$n\mapsto\underbrace{2^{2^{\cdot^{\cdot^{2^{n}}}}}}_{\text{height $h$}},$$
  for a fixed height $h$.
As observed by Frick and Grohe~\cite{frick2004complexity},
this non-elementary growth of $f$ is unavoidable in Courcelle's theorem. We discuss this in greater detail below.

In this paper, we  focus on first-order (FO) formulas,
which are much more restrictive than MSO formulas, but still powerful enough to express problems such as the existence of an independent set of a given size, or the existence of a dominating set of a given size, which are intractable in general (W[1]-complete and W[2]-complete, respectively).

In a landmark paper, Grohe, Kreutzer, and Siebertz \cite{gks} proved that for a \emph{monotone} (that is, subgraph-closed) graph class $\CC$, 
the model checking problem for first-order logic is fpt  on $\CC$ if and only if $\CC$ is \emph{nowhere dense}, a central notion of sparsity theory developed by Ne\v set\v ril and Ossona de Mendez~\cite{sparsity}.
A monotone graph class $\CC$ is called nowhere dense if 
for every $r\in\N$ there is some clique $K_n$ that is not an $r$-shallow topological minor of any graph  $G\in \CC$.
A graph $H$ is an \emph{$r$-shallow topological minor} of $G$ if by replacing edges of $H$ by paths of length at most $r+1$ we can obtain some subgraph of $G$\footnote{In the literature, see e.g.~\cite{sparsity}, one often allows paths of length at most $2r+1$ in the definition of $r$-shallow topological minors, instead of $r+1$. This difference is immaterial for all the results presented in this work.}.
Nowhere dense classes include e.g.\ the class of trees, 
the class of planar graphs, every class of bounded degree,
every class of bounded treewidth, every class that excludes a fixed graph as a (topological) minor, and every class of bounded expansion.

As proved already by Frick and Grohe \cite{frick2004complexity},
assuming $\AW[\ast]{\neq}\FPT$,
even the class of trees
does not admit an algorithm
solving the model checking problem for first-order logic
running in time as in~\eqref{eq:fpt}, if we require $f\from\N\to\N$ to be bounded by an elementary function.
That is, when the input graphs $G$ are trees, 
there is no algorithm 
with running time 
\begin{align}\label{eq:elfpt}
    \underbrace{2^{2^{\cdot^{\cdot^{2^{|\phi|}}}}}}_{\text{height $h$}}\cdot |G|^c,    
\end{align}
for any constants $c$ and $h$, unless $\AW[\ast]{=}\FPT$.
On the other hand, Frick and Grohe observed that every class of bounded degree does admit an algorithm with running time as in~\eqref{eq:elfpt} (with $c=1$ and triply exponential dependence on $|\phi|$).

The aim of this paper is to characterise those monotone graph classes for which the model checking problem for first-order logic is \emph{elementarily-fpt}, that is, can be solved in time as in \eqref{eq:elfpt} for graphs $G\in\CC$, where $c$ and $h$ are constants.
For instance, by the results of Frick and Grohe, the model checking problem is not elementarily-fpt on the class of trees (assuming $\AW[\ast]{\neq}\FPT$),
while it is such on every class of bounded degree \cite{frick2004complexity}. 
By a result of Gajarsk{\'y} and Hlin{\v{e}}n{\'y}~\cite{gajarsky-hlineny},
model checking is elementarily-fpt on classes of trees of bounded depth, and more generally, on classes of bounded treedepth (roughly, a graph has bounded treedepth if it has a tree decomposition of bounded depth and bag size).
This result was recently generalised by Lampis~\cite{lampis} to classes of bounded pathwidth (a graph has bounded pathwidth if it admits a tree decomposition in the form of a path, with bags of bounded size).
Note that the results of Lampis and of Frick and Grohe are incomparable in scope, since  bounded degree and  bounded pathwidth are incomparable properties.

\subsection*{Contribution}
Motivated by the apparent connection between the considered problem
and the existence of deep trees in the class, we define the  \emph{tree rank} of a graph class $\CC$.
Roughly, the tree rank of a class $\CC$ is the maximum  $d\in\N\cup\set\infty$ such that $\CC$ contains 
all trees of depth $d$ as $r$-shallow topological minors, for some fixed $r$.
By a \emph{tree} in this paper we always mean a rooted tree, that is, a connected acyclic graph with a distinguished root. The \emph{depth} of a tree is the maximal number of vertices on a root-to-leaf path (so stars rooted naturally have depth $2$).
Let $\cal T_d$ denote the class of trees of depth $d$.
For a graph class $\CC$, let $\topminors_r(\CC)$ denote the class of all $r$-shallow topological minors of graphs $G\in\CC$.

\begin{definition}\label{def:rank}
    The \emph{tree rank} of a graph class $\CC$ is defined as:
    $$\textup{tree\,rank}(\CC):=\max\setof{d\in \N}{\exists r\in\N:
    \cal T_d\subset \topminors_r(\CC)}.$$
    If this maximum exists, we say that $\CC$ has \emph{bounded tree rank} and otherwise, $\CC$ has \emph{unbounded tree rank}.    
\end{definition}

Clearly, the class of all trees has unbounded tree rank.
 On the other hand, every class $\CC$ that excludes some tree $T$ as a topological minor (that is, excludes $T$ as an $r$-shallow topological minor, for all~$r$),
 has tree rank smaller than the depth of $T$.
Below we list some classes that exclude a tree as a topological minor.

\begin{example}\label{ex:examples}
    \begin{enumerate}
    \item The class of trees of depth $d$ has tree rank $d$.
    \item A graph class has bounded degree if and only if it has tree rank~$1$.
    \item The class $\CC$ of graphs of pathwidth $d$ has tree rank exactly $d+1$.
    The upper bound uses the fact that the complete ternary tree $T$ of depth $d+2$ has pathwidth $d+1$ \cite[Prop. 3.2]{KIROUSIS1986205} (see also \cite[Theorem 68]{BODLAENDER19981}).
    It follows that $\CC$ excludes $T$ as a topological minor, so $\CC$ has tree rank at most $d+1$.
    For the lower bound, $\CC$ includes all trees of depth $d+1$.
    \qedhere
    \end{enumerate}
\end{example}

    We now give examples of classes that contain every tree as a topological minor, and have tree rank $2$.
\begin{example}
    Consider the class $\CC$ consisting of, for every $n\in\N$ and $n$-vertex graph $G$, the graph $G^{(n)}$ obtained from $G$ by replacing every edge by a path of length $n+1$. Then $\CC$ has tree rank $2$.

    As another example, 
    consider the following class $\mathcal{C}'$ of (rooted) trees: it contains the one-vertex (rooted) tree and for every tree $T$ in $\mathcal{C}'$ rooted at a vertex $r$, $\mathcal{C}'$ also contains the tree $T'$ (rooted at the same vertex $r$) that is obtained from $T$ by appending to every leaf $v$ of $T$, $d+1$ paths of length $d+1$, where $d$ is the distance of $v$ from $r$. Note that $\mathcal{C}'$ contains all trees as topological minors. But, since in each tree in $\mathcal{C}'$ the distance between any two vertices of degree $z> 2$ depends on $z$, $\mathcal{C}'$ has tree rank $2$.
 \end{example}

Clearly, every class of bounded tree rank is nowhere dense.
In fact, as we prove, those classes have \emph{bounded expansion}, a more restrictive notion which is also central in sparsity theory \cite{nevsetvril2008grad} (see Proposition~\ref{prop:be}).

For our main algorithmic result, we need the following, effective variant of tree rank.
\begin{definition}\label{def:el-sparse-rank}
    The \emph{elementary tree rank} of a class $\CC$ 
    is the least number $d\in\N$ with the following property.
 There is an elementary function $f\from \N\to\N$ 
    such that for all $r\in\N$ there is a tree $T$ of depth $d+1$ and of size $f(r)$, such that $T\not\in\topminors_r(\CC)$.
    If such $d$ exists, we say that $\CC$ has \emph{bounded elementary tree rank}.
\end{definition}
Note that if we drop the requirement that $f$ is elementary, the above definition becomes equivalent 
to the definition of tree rank. Hence, if $\CC$ has elementary tree rank at most $d$, then also $\CC$ has tree rank at most $d$, but the converse implication is false.
However, if $\CC$ excludes a tree $T$ as a topological minor, then the condition in Def. \ref{def:el-sparse-rank} holds for  $f\from\N\to\N$ being a constant function (the converse implication also holds).
In particular, all the examples listed in Example~\ref{ex:examples} have bounded elementary tree rank.

We can now state our main result.
\begin{theorem}\label{thm:main}
    Let $\CC$ be a graph class of bounded elementary tree rank.
Then model checking first-order logic is elementarily-fpt on $\CC$.
\end{theorem}

\begin{corollary}\label{cor:main}
    Let $\CC$ be a graph class that avoids some fixed tree $T$ as a topological minor.
    Then model checking first-order logic is elementarily-fpt on $\CC$.
\end{corollary}

Theorem~\ref{thm:main} (and even Corollary~\ref{cor:main}) vastly extends the result of Frick and Grohe \cite{frick2004complexity}
(concerning classes of bounded degree), as well as the result of Lampis \cite{lampis} 
(concerning classes of bounded pathwidth); see Example~\ref{ex:examples}.

Conversely, we prove the following lower bound, generalising (and building upon) the result of Frick and Grohe~\cite{frick2004complexity}.

\begin{theorem}\label{thm:lowerbound-intro}Assume $\AW[*]{\neq}\FPT$.
    Let $\CC$ be a monotone graph class.
    If $\CC$ admits elementarily-fpt model checking first-order logic,
    then $\CC$ has bounded tree rank.
\end{theorem}

Thus, Theorem~\ref{thm:main} and Theorem~\ref{thm:lowerbound-intro} almost completely characterise those monotone graph classes 
for which the model checking problem is elementarily-fpt.
There is a slight gap left in between the two statements, 
as there exist classes that have 
bounded tree rank but have unbounded elementary tree rank.

\paragraph{Alternation hierarchy}
As another main contribution, we study the expressive power of first-order logic on classes of bounded tree rank.
Namely, we prove that every first-order formula is equivalent on such a class to a formula with a bounded number of quantifier alternations (roughly, the number 
of alternations between existential and universal quantifiers in the formula, assuming the negations are only applied to atomic formulas). In all statements to follow, by formulas we mean first-order formulas.

\begin{theorem}\label{thm:alt}
    Let $\CC$ be a graph class of tree rank $d$.
    Then every formula $\phi$ is equivalent on $\CC$
    to a formula $\psi$ of alternation rank $3d$.
    Also, if $\CC$ has elementary tree rank $d$,
    then the size of $\psi$ is elementary in the size of~$\phi$.
\end{theorem}

Apart from being of independent interest from the logical perspective, this result is a key part of our proof of Theorem~\ref{thm:main}.

It is well known that 
the alternation hierarchy is strict:
for every number $k$ there is a formula $\phi$ of alternation rank $k+1$ that is not equivalent 
to any formula of alternation rank~$k$.
We prove that among monotone graph classes,
classes of bounded tree rank are precisely those classes, for which 
the alternation hierarchy collapses.

\begin{theorem}\label{thm:collapse}
    Let $\CC$ be a monotone graph class. Then $\CC$ has bounded tree rank if and only if there is $k\in\N$ such that every formula is equivalent to a formula of alternation rank $k$ on~$\CC$.
\end{theorem}
The forward implication follows immediately from Theorem~\ref{thm:alt}, whereas the backward implication
extends the known proof
of the strictness of the alternation hierarchy, due to Chandra and Harel~\cite{chandra1982structure}.

The overview of the proofs of the main results is relegated to Section~\ref{sec:overview}.

\subsection*{Discussion}
The following discussion puts our results in a broader context, of sparsity theory and of model theory.
We also discuss related~work.

\paragraph{New notion in sparsity theory}
Sparsity theory defines combinatorial notions
of tameness of graph classes, such as bounded expansion and nowhere denseness.
These notions are often defined by forbidding certain 
graphs as shallow topological minors, 
and the theory then develops decomposition results,
or duality results, which state that such notions can be alternatively characterised in terms of the existence of certain decompositions. Such decompositions often turn out to be useful for algorithmic and logical~purposes.

Bounded tree rank is another notion in this vein.
In this paper, we develop basic decomposition results 
for such classes, and demonstrate that they can be employed to obtain results in algorithms and logic.
We believe that this new notion will grow to become an important part of sparsity theory.

\paragraph{Towards dense graph classes}
A recent trend in algorithmic graph theory is to lift 
structural and algorithmic results concerning sparse graphs,
to the setting of dense (or not necessarily sparse) graphs.
Monotone graph classes with fpt model checking of first-order logic are completely understood, and are exactly the nowhere dense classes. In a nowhere dense graph class,
every $n$-vertex graph has $\Oof(n^{1+\varepsilon})$ edges, for every fixed $\varepsilon>0$~\cite{dvorak2007asymptotical},
and therefore such classes are considered sparse.

On the other hand, there are hereditary classes (that is, closed under vertex removal) with fpt model checking that are not sparse, such as the class of all cliques (trivially), but also 
e.g. classes of bounded cliquewidth.

Lifting sparsity theory to the setting of dense classes is a vast project, with many promising results \cite{bd_deg_interp,tww1,tww4,blcw,dms,DreierEMMPT24,flipwidth}. It reveals a deep connection between sparsity theory and stability theory, which is now the main branch of model theory.
It is conjectured that a hereditary graph class $\CC$ 
has fpt model checking of first-order logic if and only if $\CC$ is \emph{monadically NIP}, that is, $\CC$ does not transduce the class of all graphs. 
Roughly, a class $\CC$ \emph{transduces} a class $\DD$, if there is a first-order formula $\phi(x,y)$ 
such that every graph $H\in\DD$ can be obtained from some (colored) graph $G\in \CC$, by defining new edges via the formula $\phi(x,y)$, and then removing some vertices.
Several important notions from sparsity theory and model theory are (conjecturally) characterised and related in terms of forbidding transductions to some classes \cite{gajarsky2022stable,BraunfeldNMS25, flipwidth}.
For instance, it is known that model checking is fpt on \emph{structurally nowhere dense} classes \cite{dms} --- classes that can be transduced from some nowhere dense class --- and this has been very recently extended to \emph{monadically stable} classes --- classes that do not transduce the class of all half-graphs. Such classes are typically dense.

Similarly, there are hereditary classes with elementarily-fpt model checking that are not sparse,
as listed below. Here, for a property $\Pi$, we say that a class is {\em{structurally $\Pi$}} if it can be transduced from a class satisfying $\Pi$.

\begin{example}\label{ex:dense-mc}
    \begin{enumerate}
\item 
A class has bounded shrubdepth \cite{shrubdepth-journal} if and only if it can be transduced from a class of trees of bounded depth.
Every class of bounded shrubdepth  admits elementarily-fpt model checking of first-order logic, which can be easily derived from~\cite{gajarsky-hlineny}.
\item  A class has structurally bounded degree if it can be transduced from a class of bounded degree. 
It follows easily from \cite{bd_deg_interp} that
every class of structurally bounded degree
has elementarily-fpt model checking of first-order logic.
\item A class has structurally bounded pathwidth if it can be transduced from a class of bounded pathwidth. 
It follows easily from \cite{lampis,NesetrilRMS20} that 
every class of structurally bounded pathwidth
has elementarily-fpt model checking of first-order logic.
\qedhere
    \end{enumerate}
\end{example}

The research in this paper is motivated by the following conjecture,
posed by 
Pilipczuk and Toruńczyk  \cite[Conjecture XI.24]{flipwidth}.

\begin{conjecture}[\cite{flipwidth}]\label{conj:dense}
    A hereditary graph class $\CC$ has elementarily-fpt model checking if and only if $\CC$ does not transduce the class $\cal T_d$ of all trees of depth $d$, for some $d\in\N$.
\end{conjecture}

Say that a class $\CC$ of graphs has 
\emph{rank $d$} if $\CC$ does not transduce 
the class $\cal T_{d+1}$ of trees of depth $d+1$.
All the examples in Ex.~\ref{ex:dense-mc} have bounded rank, in accordance with the conjecture.
Conjecture \ref{conj:dense} predicts 
that hereditary classes of bounded rank are precisely those with elementarily-fpt model checking.
(In retrospect, it seems that the conjecture should rather involve an elementary variant of rank, analogous to elementary tree rank, although currently we have no counterexample to Conjecture~\ref{conj:dense}.)

As a first step towards Conjecture~\ref{conj:dense},
in this paper, we focus on the case of monotone graph classes, and more generally (in this context), on weakly sparse graph classes. A graph class $\CC$ is \emph{weakly sparse} if it avoids some biclique as a subgraph.
Relating the main results of the present paper with Conjecture~\ref{conj:dense}, in Section~\ref{sec:dense} we prove the following.

\begin{theorem}\label{thm:dense}
    Let $\CC$ be a weakly sparse graph class. 
    Then $\CC$ has bounded tree rank if and only 
if $\CC$ has bounded rank.
\end{theorem}

Our work opens a way towards  Conjecture~\ref{conj:dense}.
For instance, it follows from the forward implication in Theorem~\ref{thm:dense} that classes that can be transduced from a class of bounded elementary tree rank have bounded rank,
so Conjecture~\ref{conj:dense} predicts that such classes are elementarily-fpt.
This is left for future work.
Theorem~\ref{thm:dense} also suggests the following conjecture, generalising Theorem~\ref{thm:collapse}.

\begin{conjecture}\label{conj:dense-collapse}
    Let $\CC$ be a hereditary graph class.
    Then $\CC$ has bounded rank if and only if 
    there is $k\in\N$ such that every first-order formula is equivalent on $\CC$ to a formula of alternation rank $k$.
\end{conjecture}

\paragraph{Extensions}
Our results can be easily extended in various ways which are by now standard in the area of sparsity theory.
For instance, we may consider classes $\CC$ of logical structures, not only graphs, and define the tree rank of $\CC$ as the tree rank of the class of Gaifman graphs underlying the structures in $\CC$.
Our results easily extend to this more general setting.

Similarly, apart from the model checking problem,
we may consider problems specified by first-order formulas $\phi(x_1,\ldots,x_k)$ with free variables, such as query enumeration
(compute an enumerator for all $k$-tuples satisfying $\phi$), query answering (compute a data structure that determines whether a given $k$-tuple satisfies $\phi$), and counting answers to queries (compute the number of $k$-tuples satisfying $\phi$).
All those problems are known to be solvable efficiently on classes of bounded expansion, with preprocessing time $f(|\phi|)\cdot |G|$ 
and query time $f(|\phi|)$, for some function $f\from\N\to\N$
(see \cite{dvovrak2013testing,DBLP:journals/lmcs/KazanaS19,aggregate-queries}).
From our results it can be easily derived  that for classes of bounded elementary tree rank,  $f$ can be made an elementary function in all those extensions. We omit the details in this paper.

\paragraph{Related work}
The work most related to ours has already been mentioned: the work of Frick and Grohe~\cite{frick2004complexity}, of Gajarsk{\'y} and Hlin{\v{e}}n{\'y} \cite{gajarsky-hlineny}, and of Lampis \cite{lampis}.
See \cite{lampis} for a broad overview of other relevant works, in particular in the context of parameterised complexity.

In~\cite{model-theory-makes-formulas-large}, Dawar, Grohe, Kreutzer, and Schweikardt
study upper and lower bounds on the sizes of various normal forms of formulas, on restricted graph classes.
In particular, a non-elementary lower bound on the size of the Gaifman normal form is obtained,
whereas for classes of bounded degree, an elementary upper bound is obtained. It would be interesting to see how the upper and lower bounds studied in that paper relate to our notion of tree~rank.

\section{Proof overview}\label{sec:overview}
We now describe the main ideas behind the proofs of 
 Theorem~\ref{thm:main} and Theorem~\ref{thm:alt}.
 Our starting point is a game characterisation of tree rank. The game is a modification of the splitter game characterizing nowhere dense graphs, which lies at the core of the model checking result of Grohe, Kreutzer, and Siebertz \cite{gks}. The difference is that, in our variant, the splitter is allowed to delete multiple vertices instead of only one, as we explain below. 

 Fix integers $r,m\in\N$. The \emph{$m$-batched Splitter game} of radius $r$
 on a graph $G$ is played between two players,
 called Splitter and Localiser. In each round of the game,
 first Localiser replaces $G$ with a subgraph $G'$ of $G$ of radius at most $r$,
 and then Splitter deletes at most $m$ vertices from $G'$, obtaining a new graph $G''$ on which the game continues, unless $G''$ has no vertices; in this case Splitter wins the game.

\begin{theorem}\label{thm:splitter}
    Let $d\in\N$. The following conditions are equivalent for a graph class $\CC$:
    \begin{enumerate}
        \item $\CC$ has (elementary) tree rank $d$,
        \item There exists an (elementary) function $f\from\N\to\N$ such that for every $r\in\N$ Splitter wins the $f(r)$-batched splitter game of radius $r$ in at most $d$ rounds, on any $G\in\CC$.
    \end{enumerate}
\end{theorem}
Theorem~\ref{thm:splitter} is proved in Section~\ref{sec:splitter} through elementary combinatorial methods.
 This key result allows us to apply induction on the number of rounds, when considering classes of bounded tree rank.
Using this in combination with the locality property of first-order logic,
we prove Theorem~\ref{thm:collapse} in Section~\ref{sec:alt}.
Very roughly, locality says that for every first-order formula $\phi$ there is a radius $r$ such that the satisfaction of $\phi$ in $G$ only depends on the first-order definable properties of radius-$r$ balls that can be found in $G$. 
This number $r$ gives rise to the radius of the Splitter game that is played on $G$. By Theorem~\ref{thm:splitter}, Splitter wins the $f(r)$-batched splitter game of radius $r$ in at most $d$ rounds.
To obtain the desired formula $\psi$ of alternation rank bounded in terms of $d$, essentially, for every radius-$r$ ball in $G$ we use the fact that there exists an $f(r)$-tuple of vertices whose removal results in a graph in which the inductive hypothesis can be applied. Those $f(r)$ many vertices are quantified using existential quantifiers, and this is followed by a formula obtained by inductive hypothesis.

Thus, if $\CC$ is a fixed class of tree rank $d$, 
then every formula $\phi$ is equivalent on $\CC$ 
to a formula $\psi$ of alternation rank bounded in terms of $d$; in fact, as explained in Section~\ref{sec:alt}, we get alternation rank at most $3d$.
Moreover, the size of $\psi$ is elementary in terms of the size of $\phi$, provided $\CC$ has elementary tree rank $d$.
However, we do not know how to compute $\psi$ from $\phi$ in elementary time.
To simplify the reasoning, let us pretend below that we \emph{can} compute $\psi$, given $\phi$, in elementary time. The actual argument in the proof of Theorem~\ref{thm:main} is a bit more involved, and only relies on the existence of $\psi$ of size elementary in the size of $\phi$.

Next we observe (in Proposition~\ref{prop:be}) that every class of bounded tree rank has 
bounded expansion. 
A monotone class $\CC$ has \emph{bounded expansion} if and only if for every $r\in\N$ there is a constant $c=c(r)$ such that all $r$-shallow topological minors $H$ of graphs from $\CC$ satisfy $|E(H)|\le c(r)\cdot |V(H)|$.
Proposition~\ref{prop:be} follows easily from the known fact that every graph of a sufficiently large average degree contains every tree of a given size as a subgraph.
Moreover, if $\CC$ has bounded elementary tree rank, then 
the function $r\mapsto c(r)$ witnessing bounded expansion is elementary, that is, $\CC$ has \emph{elementarily bounded expansion}.

For classes $\CC$ of bounded expansion,
there is a known algorithm, due to \DKT~\cite{dvovrak2013testing}, that solves the model checking problem in time $f(\psi)\cdot |G|$ for $G\in\CC$.
Other proofs of this result are known \cite{grokre11,DBLP:journals/lmcs/KazanaS19,lsd-journal}.
By inspecting and refining the proof from~\cite{lsd-journal}, we show that the running 
time is elementary for formulas $\psi$ whose alternation rank is fixed by a constant, whenever $\CC$ has elementarily bounded expansion.
Since $\psi$ has alternation rank $3d$, this yields an elementarily-fpt model checking algorithm for $\CC$,
assuming that the formula $\psi$ can be computed from $\phi$ 
in Theorem~\ref{thm:alt} in elementary time.

In the actual proof, in Section~\ref{sec:mc}, we rely on the existence of a formula $\psi$ whose size is elementary in terms of the size of $\phi$.
Using the elementarily-fpt algorithm for classes of elementarily bounded expansion and formulas of a fixed alternation rank,
we compute a partition of the vertex set 
of a given graph into an elementary number of parts, 
where each part consists of vertices of the same first-order type of quantifier rank $q$, where $q$ is the quantifier rank of the original formula $\phi$.
 This is sufficient 
to perform model checking, by known arguments.
This  completes the proof sketch for Theorem~\ref{thm:main}.

As mentioned, the lower bounds -- Theorem~\ref{thm:lowerbound-intro} and the backward implication in Theorem~\ref{thm:collapse} -- follow by refining the 
known lower bounds, due to Frick and Grohe~\cite{frick2004complexity}, and Chandra and Harel~\cite{chandra1982structure}. This is done in Section~\ref{sec:lowerbound} and Subsection~\ref{sec:lower-alt}, respectively.

In Section~\ref{sec:dense}, 
we prove Theorem~\ref{thm:dense}, relating tree rank and rank for weakly sparse classes.
We upper-bound the rank of a class $\CC$ by its tree rank plus one,  again using an inductive argument based on the batched splitter game.
In the converse direction, we bound the tree rank of a weakly sparse class $\CC$ by its rank,
using known results concerning the existence of large trees as induced subgraphs in weakly sparse classes of unbounded degeneracy.

\section{Preliminaries}
\label{sec:prelims}

\paragraph{Graph theory} 
For a graph $G$, by $|G|$ we denote the number of vertices of $G$.
The \emph{length} of a path $P$ is the number of edges in $P$.
The \emph{depth} of a rooted tree $T$ is the maximal number of vertices on a path from the root to any leaf of $T$. For $k,d \in \N$ we denote by $T^d_k$ the rooted tree of depth $d$ in which every non-leaf node has exactly $k$ children.

The \emph{radius} of a~connected graph $G$ is defined as the minimum possible value $r \in \N$ such that for some vertex $v \in V(G)$, all vertices of $G$ are at distance at most $r$ from $v$, where the \emph{distance} of two vertices is defined as the minimum length of a path connecting them.
If $G$ is disconnected, we put $+\infty$ as its radius.
Also, given a graph $G$, an integer $r\in\N$, and a vertex $v\in G$,
we denote by $B^G_r(v)$ the set of all vertices $u$ of $G$ that are at distance at most $r$ from $v$ in~$G$. To unclutter the notation, when it is clear from the context, we use $B_r^G(v)$ instead of $G[B_r^G(v)]$.

Let $H$ be a graph. We say that a graph $H'$ is an \emph{${\le} r$-subdivision} of $H$ if it can be obtained from $H$ by replacing each edge $uv$ of $H$ by a path with endpoints $u,v$ and with at most $r$ internal vertices. 
We say that $H$ is an \emph{$r$-shallow topological minor} of $G$ if $G$ contains an ${\le}r$-subdivision of $H$ as a subgraph.

\paragraph{Elementary functions}
For every $a,b\in\N$, define $\tower(a,b)$ by induction on $a$: $\tower(0,b)=b$ and $\tower(a,b)=2^{\tower(a-1,b)}$ for~$a\ge 1$.

We do not define elementary functions,
as we only need to rely on the fact that a function $f\from\N\to\N$ is bounded by an elementary function if and only if there is some $d\in\N$ such that 
$f(k)\le \tower(d,k)$, for all $k\in\N$.
Also importantly,
elementary functions are preserved by addition, multiplication, maximum, and composition, and include the functions $n\mapsto n^n$ and $n\mapsto c^n$, for any constant $c$.

\paragraph{Tree rank}
Recall the notion of tree rank (see Def.~\ref{def:rank}).
It is easily seen that the tree rank of a graph class $\CC$  is the least number $d$ such that for every $r$ there exists $k$ such that $\C$ does not contain $T^{d+1}_k$ as an $r$-shallow topological minor. 

Similarly, it is easily seen (see Def.~\ref{def:el-sparse-rank}) that the elementary tree rank of a graph class $\CC$ is the least number $d$ such that there is an elementary function $f\from\N\to\N$ such that for every $r$, the class $\C$ does not contain $T^{d+1}_{f(r)}$ as an $r$-shallow topological minor.

\section{Game characterisation}\label{sec:splitter}

In this section we prove a characterisation of classes of bounded tree rank via games.

\paragraph{Batched splitter game}
Let $r,m \in \mathbb{N}$. The $m$-\emph{batched splitter game} with radius $r$ is played by two players, Splitter and Localiser, who play on a graph $G$.
At the beginning we set $G_0\coloneqq G$. 
For $i \ge 1$, in the $i$-th round the players make a move as follows:
\begin{itemize}
\item Localiser picks $v \in G_{i-1}$ and sets $G_{i-1}'\coloneqq  B_r^{G_{i-1}}(v)$.
\item Splitter removes at most $m$ vertices from $G_{i-1}'$ to obtain $G_i$.
\end{itemize}
Splitter wins when $V(G_i)$ is empty.

 Note that on any finite graph Splitter can always win in a finite number of rounds, as his move decreases the number of vertices in $G$ in each round. We will be interested in the smallest number of rounds such that Splitter can win the $m$-batched splitter game with radius $r$.
In particular, we will be interested in graph classes for which there exists $d\in\N$ such that the following holds: For every $r\in\N$ there is $m\in\N$ so that Splitter wins the $m$-batched splitter game with radius $r$ in at most $d$ rounds on any $G \in \C$.

As an example, consider any class $\C$ of graphs of maximum degree at most $t$. 
Then, for every $G\in\C$ and every vertex $v$ of $G$ it holds that $B_r^G(v)$ contains at most $t^r+1$ vertices. Hence, for every $r$, we can set $m\coloneqq t^{r}+1$ and observe that Splitter wins the $m$-batched splitter game in one round on any $G\in \C$.

The statement below describes our characterisation of tree rank. Its proof spans the rest of this~section.

\begin{theorem}
\label{thm:game_char}
Let $d\in\N$ with $d\ge 1$.
The following are equivalent for any graph class $\C$:
\begin{enumerate}
    \item $\C$ has tree rank $d$,
    \item There exists a function $g\from \N\to \N$ such that, for every $r\in\N$, Splitter wins the $g(r)$-batched splitter game of radius $r$ in  $d$ rounds, on any $G \in \C$.
\end{enumerate}
Moreover, $\C$ has elementary tree rank $d$ if and only if there is an elementary function $g$ as in condition 2 above.
\end{theorem}

\subsection{Proof of implication \texorpdfstring{(1) $\to$ (2)}{(1) → (2)}}

This is the more involved of the two implications. The technical core is Lemma~\ref{lem:bd_num_first_moves}, which tells us that in a graph without an ${\le}r$-subdivision of $T^{d+1}_k$, in the $r$-neighborhood of any vertex we can find a bounded number of vertices such that removing them makes this $r$-neighborhood simpler. These will be the moves of Splitter. We start with a simple lemma which will be used in the proof of Lemma~\ref{lem:bd_num_first_moves}.

\begin{lemma}
\label{lem:many_reachable}
Let $r,t \in \N$ and let $G$ be a graph of radius at most $r$. Let $S\subseteq V(G)$.
If $|S| \ge t^r+1$, then there exists $u \in V(G)$ such that there are $t$ vertices in $S$ that are distinct from $u$ and can be reached from $u$ by internally vertex-disjoint paths of length at most $r$.
\end{lemma}
\begin{proof}
    By assumption, $G$ has a rooted spanning tree $T$ of depth at most $r+1$. There is a node $u$ of $T$ such that $t$ of the subtrees of $T$ rooted at $u$ contain a node in $S$, as otherwise we would have that $|S|\le t^r$. Then $u$ satisfies the required condition.
\end{proof}

For a graph $G$ and $d,r,k \in \N$, we denote by $S^{d}_{k,r}(G)$ the set of all $v \in V(G)$ such that $v$ is the root of an ${\le}r$-subdivision of $T^d_k$ in $G$.

\begin{lemma}
\label{lem:bd_num_first_moves}
For every $d,r,k \in \N$ with $d \ge 1$ and $r\ge 1$ there exist $c\in\N$ and $k'\in\N$ such that the following holds.
Let $G$ be a graph that does not contain $T^d_k$ as an $(r-1)$-shallow topological minor. Then, for every $v \in V(G)$, we have that $|B^G_{r}(v) \cap S^{d-1}_{k',r-1}(G)| \le c$.

Moreover, for any fixed $d$, the dependence of $c$ and $k'$ on $r$ and $k$ is elementary.
\end{lemma}

We remark that the difference between values $r$ and $r-1$ in the lemma is caused by the fact that in an $\le (r-1)$-subdivision of a graph the subdivision paths can have length up to $r$.

\begin{proof}
Let $z$ be the maximum number of vertices any ${\le}(r-1)$-subdivision
of $T^{d-1}_k$ can have. We set $k'\coloneqq k(z+r)+r$,  $t\coloneqq k(z + r)$ and $c\coloneqq t^{r}$.

Suppose for contradiction that for some  $v \in V(G)$ we have that $|B^G_r(v) \cap S^{d-1}_{k',r-1}(G)| > c$. 
By Lemma~\ref{lem:many_reachable}, there is a vertex $u \in B^G_{r}(v)$ and  $t$ vertices $v_1, \ldots, v_t$, distinct from $u$ and reachable from $u$ by internally vertex-disjoint paths $P_1, \ldots, P_t$ of length at most $r$ such that each $v_i$ belongs to $S^{d-1}_{k',r-1}(G)$, or, in other words,  every $v_i$ is a root of an ${\le }{(r-1)}$-subdivision $T_i$ of $T^{d-1}_{k'}$ in $G$. 
Our goal is to use $u$ together with pairs $(P_1, T_1), \ldots, (P_t,T_t)$ to find an ${\le }{(r-1)}$-subdivision of $T^d_{k}$ in $G$, which will yield a contradiction. 

We will first adjust each $T_i$ by sacrificing some branches to ensure that $P_i$ intersects the remaining tree only in $v_i$ as follows.
In any $T_i$, let $R_i^1, \ldots, R_i^{k'}$ be the subtrees of $T_i$ rooted at children of $v_i$. Then $P_i$ intersects at most $r$ trees in $\{R_i^1, \ldots, R_i^{k'}\}$, and by removing all such subtrees from $T_i$ there will be at least $k_0\coloneqq k' - r = k(z+r)$ subtrees left.  This way we obtain $T_i'$, which contains an ${\le }{(r-1)}$-subdivision of $T^{d-1}_{k_0}$ as a subgraph, and $P_i$ is disjoint from $T_i'$, except for their common vertex $v_i$. 

We thus have $t$ pairs $(P_1,T_1'), \ldots, (P_t,T_t')$ such that $P_i$ intersects $T_i'$ only in $v_i$ and each $T_i'$ is a ${\le}{(r-1)}$ subdivision of $T^{d-1}_{k_0}$.
We now proceed by extracting from $(P_1,T_1'), \ldots, (P_t,T_t')$ a set of $k$ pairs $(Q_1, W_1),\ldots, (Q_k,W_k)$ such that
\begin{itemize}[nosep]
    \item each $W_i$ is an ${\le}{(r-1)}$-subdivision of $T^{d-1}_k$,
    \item $V(Q_i) \cap V(W_i) = \{v_i\}$ and $Q_i$ connects the root of $W_i$ to $u$, and
    \item $(V(W_i) \cup V(Q_i)) \cap (V(W_j) \cup V(Q_j)) = \{u\}$ for any $i,j$ with~$i \not= j$.
\end{itemize}
The vertex $u$ together with trees $W_1,\ldots, W_k$ and paths joining them to $u$ will then form an ${\le}{(r-1)}$-subdivision of $T^{d}_k$ in $G$, as desired.

We start with $U = \{1,\ldots, t\}$ and proceed in $k$ rounds. 
We keep the invariant that after the $l$-th round we have that $|U| \ge (k-l)(z + r )$ and for each $i \in U$ we have that $T_i'$ contains a subdivision of $T^{d-1}_{k_l}$, where $k_l\coloneqq (k-l)(z+r)$. Clearly this holds for $l=0$.

For each $l \ge 1$, in the $l$-th round we do the following:
\begin{itemize}[nosep]
    \item We pick any $(P_i,T_i')$ such that $i \in U$ and consider a subtree $T_i''$ of $T_i'$ rooted at $v_i$ that is an ${\le }{(r-1)}$-subdivision of $T^{d-1}_{k}$. 
    \item We remove from $U$ both $i$ and every $j$ such that $v_j \in  V(P_i) \cup V(T_i'')$; since all $v_1,\ldots, v_t$ are pairwise different and  $|V(P_i) \cup V(T_i'')| \le z + r$, there are at most $r + z$ indices removed, and so after the $l$-th round there are at least $(k-l)(z+r)$ indices left in $U$.
    \item From every $T_j'$ with $j \in U$ we remove all subtrees rooted at children of $v_j$ that intersect $ V(P_i) \cup V(T_i'')$. Since $|V(P_i) \cup V(T_i'')| \le z + r$ and before the $l$-th round the tree $T_j'$ contained an ${\le}{(r-1)}$-subdivision  of $T^{d-1}_{k_{l-1}}$, we get that after this trimming  the tree $T_j'$ contains an ${\le}{(r-1)}$-subdivision  of $T^{d-1}_{k_l}$.
    \item Finally, we set $W_l\coloneqq T_i''$ and $Q_l\coloneqq P_i$.
\end{itemize}
By our choices of $t$ and $k'$, this process will last for at least $k$ rounds and will produce $k$ pairs $(Q_1,W_1), \ldots, (Q_k, W_k)$. By the construction, these pairs have the required properties.

It remains to argue that for fixed $d$, the dependence of $k'$ and $c$ on $k$ and $r$ is elementary. It is easily seen that both values depend elementarily on $z$. Thus, to conclude the proof, it suffices to show that the dependence of $z$ on $k$ and $r$ is elementary.
To bound $z$, first note that $T^{d-1}_k$ can have at most $1 + k + k^2 + \ldots + k^{d-1} \le 1+ dk^{d-1}$ vertices (the $+1$ is added to cover the case $k=0$ and $d\ge 2$). Since in any ${\le }{(r-1)}$ subdivision of  $T^{d-1}_k$ each edge of $T^{d-1}_k$  is replaced by at most $r-1$ vertices, and the number $dk^{d-1}+1$ that bounds the number of vertices also upper-bounds the number of edges,  we have that  $z \le 1+ dk^{d-1} + (dk^{d-1}+1)(r-1)$, which is elementary in $k$ and $r$.
\end{proof}

With the previous lemma at hand we can now show the existence of a winning strategy for Splitter. The idea is that at every round, Splitter removes a suitable set $S^d_{k,r}(G)$ intersected with the ball played by the Localiser, and Lemma~\ref{lem:bd_num_first_moves} ensures that this way he removes only a bounded number of vertices.

\begin{lemma}
\label{lem:win} 
For any $d, r,k \in \N$ with $d \ge 1$ and $r \ge 1$ there exists $m=m(d,r,k)$ such that if $G$ is a graph that does not contain $T^{{d+1}}_k$ as an $(r-1)$-shallow topological minor, then Splitter wins the $m$-batched splitter game of radius $r$ in $d$ rounds on $G$.

Moreover, for any fixed $d$, the dependence of $m$ on $r$ and $k$ is elementary.
\end{lemma}
\begin{proof}
We proceed by induction on $d$. For $d=1$, we set $m\coloneqq 1+ (k-1) + \ldots + (k-1)^{r}$.  Since $G$ does not contain $T^2_{k}$ as an $(r-1)$-shallow topological minor, the maximum degree of any vertex of $G$ is $k-1$. Therefore, the $r$-ball around any $v \in V(G)$ contains at most $m$ vertices, and so no matter which vertex $v$ Localiser plays in the first move, Splitter can make $B^G_r(v)$ empty by removing all vertices.

For $d> 1$, we proceed as follows.
Let us first note that the choice of $m$ will be specified later; it will be clear that it does not depend on any particular graph $G$ but only on $d$, $k$, and $r$. Let $v$ be the vertex played by Localiser in his first move, and let $G_0'\coloneqq B^G_{r}(v)$.
Then, by Lemma~\ref{lem:bd_num_first_moves}, the set $S$ of vertices of $G_0'$ which can serve as the root of an ${\le }{(r-1)}$-subdivision of $T^{d}_{k'}$ has size at most $c$, where $c$ and $k'$ are the constants provided by Lemma~\ref{lem:bd_num_first_moves} for $d$, $k$, and $r$.
Splitter removes these vertices and so $G_1\coloneqq G_0'\setminus S$ does not contain an ${\le }{(r-1)}$-subdivision of $T^{{d}}_{k'}$. 
By the induction hypothesis applied to $d-1,r,$ and $k'$, there exists $m'\in\N$ such that Splitter wins the $m'$-batched game in $d-1$ rounds on $G_1$. Therefore we can set $m\coloneqq \max\{c, m'\}$.

We now argue the elementary dependence of $m$ on $r$ and $k$, for a fixed $d$. Both $c$ and $k'$ from Lemma~\ref{lem:bd_num_first_moves} depend elementarily on $r$ and $k$. Since from the induction hypothesis we know that  $m'=m(d-1,r,k')$ also depends  elementarily on its parameters, we get that $m=\max\{c, m'\}$ is elementary in $r$ and~$k$.
\end{proof}

\begin{proof}[Proof of the implication (1) $\to$ (2) of Theorem~\ref{thm:game_char}]
Let $\C$ be a graph class of tree rank at most $d$. Then there exists a function $f\from\N \to \N$ such that for every $r\in\N$,
no $G \in C$ contains $T^{{d+1}}_{f(r)}$ as $r$-shallow topological minor. For $r=0$, set $g(0)\coloneqq 1$.
For any $r \ge 1 $, set $g(r)\coloneqq m(d,r,f(r-1))$, where $m$ is the function from Lemma~\ref{lem:win}.
Then, by Lemma~\ref{lem:win},  for any $G \in\C$ Splitter wins the $g(r)$-batched splitter game of radius $r$ in at most $d$ rounds on $G$.

Also, assuming that $f$ is an elementary function of $r$, the function $g \colon r \mapsto m(d,r,f(r-1))$ is also an elementary function of $r$, because the function $m$ from Lemma~\ref{lem:win}  is elementary~in~$r$.
\end{proof}

\subsection{Proof of implication \texorpdfstring{(2) $\to$ (1)}{(2) → (1)}}
This implication follows easily from the following lemma.

\begin{lemma}
\label{lem:connector_wins}
For every $d, r, m \in \N$ with $d \ge 1$ there is a strategy for Localiser in the $m$-batched splitter game with radius $d\cdot (r+1)$ to make the game last at least $d+1$ rounds on any ${\le}r$-subdivision of $T^{{d+1}}_{m+1}$.
\end{lemma}
\begin{proof}
By induction on $d$. For $d=1$ we note that $T^2_{m+1}$ is a star with $m+1$ leaves.  The root $v$ of an ${\le}r$-subdivision of a star with $m+1$ leaves has degree at least $m+1$ in $G$. Localiser  can therefore play $v$ as his first move, and then $|B^G_r(v)| \ge m+2$, which means that Splitter cannot make $G_0' = B^G_r(v)$ empty in one~move.

For $d > 1$, if $T$ is any ${\le}r$-subdivision of $T^{{d+1}}_{m+1}$, Localiser chooses the root $v$ of $T$ as his first move. Then the whole tree $T$ is contained in $G_0'\coloneqq B^T_{d\cdot (r+1)}(v)$.
Let $v_1,\ldots, v_{m+1}$ be the branching vertices closest to $v$ in $T$; each $v_i$ is the root of an ${\le}r$-subdivision of $T^{{d}}_{m+1}$, call it $T_i$.
Since Splitter can remove at most $m$ vertices, at least one
subtree $T_i$
will not be affected by this, and so $G_1$ will contain an ${\le}r$-subdivision of $T^{{d}}_{m+1}$ as a subgraph. 
By the induction hypothesis, Localiser can survive at least $d$ rounds on $G_1$, and so together with the first move played this gives him a strategy to survive $d+1$ rounds.
\end{proof}

\begin{proof}[Proof of the implication (2) $\to$ (1) of Theorem~\ref{thm:game_char}]
Let $\C$ be a graph class with function $g$ such that for every $r\in\N$,
Splitter wins the $g(r)$-batched game of radius $r$ in at most $d$ rounds on any $G \in \C$.
We define the function $f\from \N \to \N$ by setting $f(r)\coloneqq g(d\cdot (r+1))+1$. We claim that no $G \in \C$ can contain $T^{d+1}_{f(r)}$ as an $r$-shallow topological minor. Assume for contradiction that there exists $r\in\N$ such that some $G \in \C$ contains an ${\le}r$-subdivision of $T^{d+1}_{f(r)}$ as a subgraph. 
Set $m\coloneqq f(r)-1 = g(d\cdot (r+1))$ and consider the $m$-batched splitter game of radius $d\cdot (r+1)$ played on $G$. By our assumption on $\C$ and $g$, Splitter can win in at most $d$ rounds. But by Lemma~\ref{lem:connector_wins}, Localiser can survive at least $d+1$ rounds in the $m$-batched game of radius $d\cdot (r+1)$ on any ${\le}r$-subdivision of $T^{{d+1}}_{m+1} = T^{{d+1}}_{f(r)}$ and consequently on $G$, a contradiction.

Moreover, if $g$ is an elementary function  of $r$, then clearly so is $f$, as desired.
\end{proof}

\section{Auxiliary definitions and logical preliminaries}
\label{sec:aux}
In this section we collect the necessary definitions and tools concerning logic which will be used later on.

\subsection{\texorpdfstring{$(r,m)$}{(r,m)}-game rank}
We will use the following auxiliary definitions related to classes of bounded tree rank. These are based on the $m$-batched splitter game instead of excluding trees as $r$-shallow topological minors. This is motivated by Theorem~\ref{thm:game_char}.
\begin{definition}\label{def:rm_rank}
Let $r,m$ be integers.
We say that graph $G$ has $(r,m)$-game rank
at most $d$ if Splitter wins the $m$-batched splitter game with radius $r$ in at most $d$ rounds.
\end{definition}

It follows easily from Definition~\ref{def:rm_rank} and Theorem~\ref{thm:game_char} that a graph class $\C$ has tree rank at most $d$ if for every $r\in\N$ there exists $m\in\N$ such that every $G \in \C$ has $(r,m)$-game rank at most $d$. One also easily observes the following.

\begin{lemma}
Let $G$ be a graph of $(r,m)$-game rank at most $d$. If $r' < r$, then $G$ is of $(r',m)$-game rank at most $d$.
\end{lemma}

For a graph $G$ and a subset $S$ of its vertices we define $G \ast S$ as the operation that isolates the vertices from $S$ in $G$.
\begin{definition}
\label{def:isolate}
Let $G$ be a graph and $S \subseteq V(G)$.
We define $G \ast S$ to be the graph with vertex set $V(G\ast S)\coloneqq V(G)$ and edge set 
$E(G\ast S)\coloneqq E(G)\setminus \{ uv~|~ u \in S \text{ or } v \in S \}$.
\end{definition}

Using the above definition, we prove the following~result.

\begin{lemma}
\label{lem:game_rank}
Let $G$ be a graph of $(r,m)$-game rank $d >1$, let $v \in V(G)$ and let $u_1,\ldots, u_m$ be vertices such that  $B^G_r(v) \setminus \{u_1,\ldots, u_m\}$ is of $(r,m)$-game rank $d-1$. Then also $B^G_r(v)\ast \{u_1,\ldots,u_m\}$ is of $(r,m)$-game rank at most~$d-1$.
\end{lemma}
\begin{proof}
By the assumption, Splitter wins in at most $d-1$ rounds on $B^G_r(v) \setminus \{u_1,\ldots, u_m\}$.
The graph $G\ast \{u_1,\ldots,u_m\}$ differs from $G'$ in that it has $m$ extra isolated vertices, and this does not affect the duration of the batched splitter game.
\end{proof}

\subsection{Tools from logic}

We assume familiarity with basic notions of first-order logic: signatures, formulas, quantifier rank.
We use $\tup x,\tup y,$ etc to denote sets of variables.

Graphs are modeled as structures over a single binary predicate symbol $E$. 
Often, we will speak of \emph{colored} graphs -- these are graphs equipped with finitely many colors modeled as unary predicates. To each class of colored graphs, we associate a signature $\Sigma$ which consists of the binary predicate symbol $E$ and finitely many unary predicate symbols; each such unary predicate symbol is associated to a color. Unless explicitly specified otherwise, all classes of structures considered in this paper will be (colored) graphs, and we will therefore use letters $G,H, \ldots$ for structures and $V(G), V(H), \ldots$ for their underlying universes.

\paragraph{Gaifman's theorem}
Several of our results will rely on Gaifman's locality theorem.
We say that a first-order formula $\gamma(x_1,\ldots, x_k)$ is $r$-\emph{local} if for every $G$ and every $v_1,\ldots,v_k \in V(G)$ we have that $G \models \gamma(v_1,\ldots,v_k) \Leftrightarrow B^G_r(v_1,\ldots,v_k) \models \gamma(v_1,\ldots,v_k)$ where $ B^G_r(v_1,\ldots,v_k) = \bigcup_{1 \le i \le k} B^G_r(v_i) $.
A \emph{basic local sentence} is a sentence of the form
$$ \tau \>\equiv\>
	\exists x_1 \ldots \exists x_k	\left(\bigwedge_{1 \le i < j \le k} 
		\dist(x_i,x_j) > 2r 
	\land \bigwedge_{1 \le i \le k} \alpha(x_i) \right)
,$$
where $\alpha(x)$ is $r$-local and $\dist(x,y)>2r$ is a formula expressing that the distance between $x$ and $y$ is larger than $2r$. The number $r$ is called the \emph{locality radius} of the basic local sentence $\tau$.

\begin{theorem}[Gaifman's locality theorem \cite{gaifman1982local}]\label{thm:Gaifman}
Every first-order formula $\phi(x_1,\ldots,x_k)$ is equivalent to a Boolean combination of $7^q$-local formulas $\rho_1(x_1,\ldots,x_k),\ldots,\rho_m(x_1,\ldots,x_k)$ and basic
local sentences $\tau_1,\ldots,\tau_n$ with locality radius $7^q$, where $q$ is the quantifier rank of $\phi$ and $m,n\in\mathbb{N}$.  Furthermore,
the quantifier rank of the resulting sentences is at most $f(q)$, where $f$ is some elementary non-decreasing function.
\end{theorem}

The assertion in the last statement of the theorem is not claimed explicitly in \cite{gaifman1982local}, but follows easily from the proof.
Note that if $\phi$ is a sentence, then Gaifman's theorem says that $\phi$ is equivalent to a Boolean combination of basic local sentences (there are no formulas~$\rho_i$).

\paragraph{Interpretations}
In some of our proofs we will  rely on the notion of \emph{interpretation},
which is a mechanism using logical formulas to produce new graphs from old ones. We present it here specialized to our setting of (colored) graphs. 
Let $\Sigma$, $\Gamma$ be two signatures of colored graphs and let $p$ be the number of unary predicate symbols in $\Gamma$. An \emph{interpretation} from $\Sigma$ to $\Gamma$ is a tuple $I=(\psi_E, \vartheta, \theta_1, \ldots \theta_p)$ of $\Sigma$-formulas such that $\psi_E$ has two free variables and all other formulas in $I$ have one free variable. To any graph $G$ over the signature $\Sigma$ an interpretation $I$ assigns the new graph $H\coloneqq I(G)$ such that $V(H)\coloneqq \{v~|~G \models \vartheta(v)\}$, $E(G)\coloneqq \{ uv~|~u,v \in V(H), G \models \psi_E(u,v)\}$, and the $i$-th unary predicate of $\Gamma$ is realized by the set $S_i\coloneqq \{v \in V(H)~|~ G \models \theta_i(v)\}$.

\paragraph{Interpretations with parameters}
In our proof of Theorem~\ref{thm:alternations} we consider a slightly less well-known form of interpretations which use parameters.
We first present the notion for uncolored graphs.
Fix two formulas $\vartheta(x,y_1,\ldots,y_k)$ and $\psi_E(x_1,x_2,y_1,\ldots, y_k)$ and set $I\coloneqq(\vartheta,\psi_E)$; the variables $y_1,\ldots,y_k$ are called \emph{parameters}.
For any graph $G$ and vertices  $u_1,\ldots,u_k \in V(G)$ these formulas define a new graph $H=I(G,u_1,\ldots,u_k)$ with
$$V(H) \coloneqq \{ w \in V(G)~|~ G \models \vartheta(w,u_1,\ldots,u_k) \} $$
$$ E(H) \coloneqq  \{ w_1w_2 \in V(H)~|~ G \models \psi_E(w_1,w_2,u_1,\ldots,u_k) \}$$

We include two examples, which will be used in our proofs. First, there is an interpretation that from $G$ and $v,u_1,\ldots, u_k$ produces the graph $B^G_r(v) - \{u_1,\ldots,u_k\}$. This uses $k+1$ parameters and consists of formulas $\vartheta(x,y_0,y_1,\ldots,y_k)\coloneqq(\dist(x,y_0) \le r) \land (x\not=y_1) \land \ldots \land (x\not=y_k)$ and $\psi_E(x_1,x_2,y_0,\ldots,y_k)\coloneqq E(x_1,x_2)$.
The second example is an interpretation that from $G$ and $v,u_1,\ldots,u_k$ produces the graph $B^G_r(v) \ast \left(\{u_1,\ldots, u_k\}\cap B^G_r(v)\right)$. This is given by formulas $\vartheta(x,y_0,y_1,\ldots,y_k)\coloneqq\dist(x,y_0) \le r$ and  $$\psi_E(x_1,x_2,y_0,y_1,\ldots,y_k)\coloneqq E(x_1,x_2) \land \bigwedge_{1\le i \le k} (x_1 \not= y_i)  \land \bigwedge_{1\le i \le k} (x_2 \not= y_i).$$

In the setting of colored graphs, say with $t$ colors, interpretations use $t$ additional formulas $\theta_1(x,y_1,\ldots,y_k),\ldots,$ $\theta_t(x,y_1,\ldots,y_k)$, we set $I\coloneqq\{\vartheta,\psi_E,\theta_1,\ldots,\theta_t\}$, and for each $i$ we mark in $H = I(G, u_1,\ldots,u_k)$ the vertices from the set
$$ S_i\coloneqq \{w \in V(G)~|~ G\models \theta_i(w,u_1,\ldots,u_k) \land \vartheta(w,u_1,\ldots,u_k )\} $$
by the unary predicate $C_i$ corresponding to the $i$-th color.

\section{Bounded alternation rank}\label{sec:alt}

This section is devoted to the proof of the following theorem, which is a restatement of one direction of Theorem~\ref{thm:alt} suited for formulas with free variables. See Section~\ref{sec:prelim'} for the definition of the alternation rank of a formula.
Roughly, it is the number of alternations 
between quantifiers $\forall$ and $\exists$ in a formula which applies negations only to atomic formulas.

\begin{theorem}
\label{thm:batched_rank}
\label{thm:alternations}
Let $\C$ be a graph class of tree rank $d$. Every formula $\phi(\tup x)$ is equivalent on $\CC$ to some  formula
$\psi(\tup x)$  of alternation rank at most $3d-1$.
Moreover, if $\C$ is a class of elementary tree rank $d$ then the size of $\psi$ is  elementary in the size of $\phi$.
\end{theorem}

We now briefly discuss the high-level idea behind the proof of Theorem~\ref{thm:alternations}.
From Gaifman's theorem (Theorem~\ref{thm:Gaifman}) we know that deciding whether $G \models \phi$ can be reduced to evaluating formulas on $r$-neighborhoods of vertices of $G$.
Moreover, from the Splitter game characterisation (Theorem~\ref{thm:game_char}) we know that for any class $\C$ of graphs of tree rank $d$ we have the property that for any $G \in \C$ and $v\in V(G)$ we can make the $r$-neighborhood of $v$ simpler (of rank $d-1$) by removing a bounded number of vertices.
Combining these two facts leads to a proof of Theorem~\ref{thm:alternations} by induction on the tree rank of $\C$. First, we express $\phi$ in terms of local formulas. Since these formulas need to be only evaluated on small neighborhoods around vertices of $G$, and these neighborhoods can easily be made of rank $d-1$, we can replace these local formulas by equivalent formulas of desired alternation rank obtained from induction hypothesis. We then combine these formulas to obtain $\psi$.

While this plan is conceptually simple, it comes with certain technical difficulties.
First, one has to be careful about what it means to make the $r$-neighborhood of some vertex $v$ simpler by removing a bounded number of vertices. Consider a graph $G$ such that Splitter wins the $m$-batched splitter game of radius $r$ on $G$ in $d$ rounds. Then we know that if we pick any $v$ in $G$, we can remove a small set $S$ of vertices from $B^G_r(v)$ to get a graph on which Splitter can win the game in $d-1$ rounds. However, this makes the graph simpler only with respect to this particular $r$, and not necessarily with respect to larger radii. In short, the graph class 
$$
\C'\coloneqq\{B^G_r(v)- S~|~G \in \C,~v \in V(G),~S \text{ is $v$-good}\} 
$$
 where being $v$-good means that Splitter wins in $d-1$ rounds on $B^G_r(v)-S$, is not necessarily of tree rank $d-1$. This means that we cannot apply the induction hypothesis to $\C'$. To resolve this, we strengthen the induction hypothesis -- we will prove for every formula $\phi$ there exists an $r\in\N$ such that if $\C$ is any class for which Splitter has a winning strategy with small batch size to win the game of radius $r$ in at most $d$ rounds on any $G \in \C$, then there exists a formula
 $\psi$ equivalent to $\phi$ on $\C$.

Second, the formulas we obtain from the induction hypothesis only work as desired on $r$-neighborhoods in $G$ (and even that only after removing some vertices), while the formula $\psi$ has to work on the whole graph $G$.
To get around this, the formula $\psi$  essentially does the following: it considers $B^G_r(v)$ for any $v \in G$ it needs, identifies vertices which lead to Splitter's win in $d-1$ rounds, removes/isolates them, evaluates the formula of small alternation rank obtained from induction hypothesis, and then translates the result to the original graph. This is achieved by using interpretations with parameters, introduced in Section~\ref{sec:aux} and further discussed in Section \ref{sec:prelim'2}.

\subsection{Alternation rank and batched formulas}\label{sec:prelim'}

Let $\Sigma_i$ denote the set of formulas 
in prenex normal form 
of the form 
\[\underbrace{\exists^*\forall^*\exists^*\cdots }_{i \text{ blocks}}\phi,\]
where $\phi$ is quantifier-free.
Dually, let $\Pi_i$ denote the set of formulas of the form 
$$\underbrace{\forall^*\exists^*\forall^*\cdots }_{i \text{ blocks}}\phi.$$
Formally, $\Sigma_0$ and $\Pi_0$ are quantifier-free formulas,
while for $i>0$, the set $\Sigma_i$ consists of formulas of the form 
$\exists \bar x. \phi(\bar x,\bar y),$
for $\phi\in\Pi_{i-1}$,
and dually, $\Pi_i$ consists of formulas of the form $\forall \bar x. \phi(\bar x,\bar y),$
for $\phi\in\Sigma_{i-1}$.

Note that every $\Sigma_i$-formula is equivalent to a negation of a $\Pi_i$-formula, and vice-versa.
Let $B\Sigma_i$ denote the set of Boolean combinations of formulas in $\Sigma_i$.
Those are equivalent to positive Boolean combinations of formulas in $\Sigma_i$ and $\Pi_i$.

A formula $\phi$ is in \emph{negation normal form}
if it contains no negations, apart from 
the literals (atomic formulas or their negations).
Every formula $\phi$ can be converted into negation
normal form by using de Morgan's rules, in polynomial time.

For a formula $\phi$, we define its \emph{alternation rank} as follows.
First, convert $\phi$ into negation normal form.
Next, consider the maximum, over all branches in the syntax tree of $\phi$, of the number of alternations between the symbols $\forall$ and $\exists$.
Namely, a branch in the syntax tree of $\phi$
determines a word over the alphabet $\set{\lor,\land,\exists,\forall}$.
We ignore the symbols $\lor$ and $\land$ in this word, obtaining a word 
over the alphabet $\set{\exists,\forall}$.
The alternation number of this branch is the number of subwords of the form $\exists\forall$ or $\forall\exists$.
The alternation number of a formula in negation normal form is the maximal alternation number of all branches in the syntax tree.

For $m\in\N$, we define the set of  \emph{$m$-batched formulas} inductively, as follows. Simultaneously, for each such formula we define its \emph{batched quantifier rank}.
\begin{itemize}
\item Any atomic formula is an $m$-batched formula of batched quantifier rank $0$.
\item Any Boolean combination of $m$-batched formulas $\phi_1,\ldots,\phi_k$ is an $m$-batched formula,
whose batched quantifier rank is equal to the maximal batched quantifier rank of the formulas $\phi_1,\ldots,\phi_k$.
\item If $\phi$ is an $m$-batched formula, then also $(\exists x_1 \ldots \exists x_k)  \phi$ and $(\forall x_1 \ldots \forall x_k)  \phi$ are $m$-batched formulas, where $k\le m$ and $x_1,\ldots,x_k$ are any variables.
The batched quantifier rank of those formulas is one plus the batched quantifier rank of $\phi$.
\end{itemize}
A \emph{batched formula} is an $m$-batched formula, for some $m$.

\begin{lemma}\label{lem:alternation-rank}
  Fix $q\ge 1$. The following conditions are equivalent for a formula $\phi$:
	\begin{enumerate}
		\item $\phi$ is equivalent to a  $B\Sigma_q$-formula,
		\item $\phi$ is equivalent to a batched formula of batched quantifier rank $q$, 
		\item $\phi$ is equivalent to a 
		formula of alternation rank $q -1$.
	\end{enumerate}
	Moreover, all translations are computable in polynomial time.
\end{lemma}

\begin{proof}
	Implications 1$\rightarrow$2 and 2$\rightarrow$3 are immediate.
	The implication 3$\rightarrow$1
	follows from the standard conversion into prenex normal form, which does not increase the alternation rank.
\end{proof}

\begin{lemma}\label{lem:batched-to-alt}
  Fix $q\in\N$. 
  There is an elementary function $f_{q}\from \N^3 \to\N$ 
  such that for all $m\in\N$, $k \in \N$ and signature $\Sigma$, every $m$-batched $\Sigma$-formula $\phi(x_1,\ldots,x_k)$ of batched quantifier rank~$q$ 
  is equivalent to an $m$-batched $\Sigma$-formula $\psi(x_1,\ldots,x_k)$ 
  of batched quantifier rank~$q$, of size at most $f_{q}(m,k,|\Sigma|)$.
\end{lemma}
\begin{proof}
  We proceed 
  by induction on $q$. For $q=0$ the statement is trivial, since the number of distinct atomic $\Sigma$-formulas on variables among $x_1,\ldots,x_k$ is an elementary function of $k$ and the number of symbols in $\Sigma$.
  In the inductive step,
  every $m$-batched formula $\phi$ is a Boolean combination of  formulas of the form $(\exists x_1\ldots \exists x_\ell)\phi'$ 
  for $\ell\le m$ and an $m$-batched formula $\phi'$ of batched quantifier rank $q-1$ with at most $k+m$ free variables.
By inductive assumption, we may convert each $\phi'$ 
to a batched $\Sigma$-formula $\psi'$ of size at most $f_{q-1}(m,k+m,|\Sigma|)$ and batched quantifier rank $q-1$.
Thus, $\phi$ is equivalent to a Boolean combination $\phi_1$ of $m$-batched formulas of batched quantifier rank $k$ and of size at most $f_{q-1}(m,k+m,|\Sigma|)+m$ each. Let $f(|\Sigma|,s)$ denote the number of distinct $\Sigma$-formulas of size at most $s$; then $f$ is an elementary function.  
Thus, $\phi_1$ is a Boolean combination of at most $f(|\Sigma|,f_{q-1}(m,k+m,|\Sigma|)+m)$ different formulas. 
Let $g(k)$ denote the 
smallest number $t$ such that every Boolean combination of $k$ atoms is equivalent to a Boolean combination $\psi$ of size at most $t$. Then $g$ is a singly-exponential function.
Hence, $\phi_1$, and therefore $\phi$, is equivalent to a Boolean combination $\xi$ of $m$-batched formulas of size at most $f_{q-1}(m,k+m,|\Sigma|)+m$ and $\xi$ has size at most $g(f(|\Sigma|,f_{q-1}(m,k+m,|\Sigma|)+m))$.

Define $f_q(m,k,|\Sigma|)=g(f(|\Sigma|,f_{q-1}(m,k+m,|\Sigma|)+m))\cdot (f_{q-1}(m,k+m,|\Sigma|)+m)$.
Then $f_q$ is an elementary function, and satisfies the condition of the lemma.
\end{proof}

\subsection{Further preliminaries}
\label{sec:prelim'2}

In the proof of Theorem~\ref{thm:alternations} we will repeatedly use the following fact about interpretations with parameters.
Let $G$ be a graph, $I$ an interpretation with parameters and $\phi$ a sentence.
If we want to determine whether $I(G,v_1,\ldots,v_k) \models \phi$, it is enough to evaluate a suitable formula $\hat{\phi}$ on $G,v_1,\ldots, v_k$. This is formalized as follows.

\begin{lemma}
\label{lem:interp_lemma}
Fix an interpretation $I$ with parameters $y_1,\ldots, y_k$.
 Let $\phi$ be a sentence. Then one can rewrite $\phi$ into a formula $\hat{\phi}(y_1,\ldots,y_k)$ such that for every $G$ and $u_1,\ldots,u_k \in V(G)$ we have  
 $$  I(G,u_1,\ldots,u_k) \models \phi \Longleftrightarrow G \models \hat{\phi}(u_1,\ldots,u_k).$$   

Moreover, if $\phi$ is $a$-batched of batched quantifier rank $\alpha$ and each formula in the interpretation is $b$-batched of batched quantifier rank $\beta$, then we can make $\hat{\phi}$ to be $c$-batched with batched quantifier rank at most $\alpha + \beta$, where $c\coloneqq\max\{a,b\}$.
\end{lemma}

We only sketch the proof of the lemma, as the first part of it is standard, and the second part follows by a simple modification of the proof of the first part.

\begin{proof}[Proof sketch]
To obtain the formula $\hat{\phi}$, we first replace each occurrence of $E(z_1,z_2)$  
by $\psi_E(z_1,z_2,y_1,\ldots, y_k)$ (here $z_1,z_2$ are any variables), and each occurrence of unary predicate symbol $C_1(z)$ by $\theta_i(z,y_1,\ldots,y_k)$. Moreover, the quantifications occurring in $\phi$ are relativized in $\hat{\phi}$ to vertices which satisfy $\vartheta$ --  every subformula $\exists z \rho(z,\bar{x})$ of $\phi$ is replaced by $\exists z (\vartheta(z,y_1,\ldots,y_k) \land \rho(z,\bar{x}))$  and every  $\forall z \rho(z,\bar{x})$ of $\phi$ is replaced by $\forall z (\vartheta(z,y_1,\ldots,y_k) \rightarrow \rho(z,\bar{x}))$. One can then prove the result by induction on the structure of $\phi$.

If we work with batched formulas, we first note that all formulas in $ \{\phi\} \cup I$ are $c$-batched formulas; for every $0\le i < j$, every $i$-batched formula is also $j$-batched. Then we proceed as above, but we handle the relativization of quantifiers in batches, meaning that $\exists z_1 \ldots \exists z_c ~\rho(z_1,\ldots,z_c,\bar{x})$ is replaced with
$$ \exists z_1 \ldots \exists z_c \bigl(\bigwedge_{1 \le i \le c} \vartheta(z_i,y_1,\ldots, y_k) \land \rho(z_1,\ldots,z_c,\bar{x}) \bigr) $$
and analogously for the batched universal quantifier. 
Again, induction on the structure of $\phi$ yields the result.
The bound on the batched quantifier rank of $\hat{\phi}$ follows form considering the syntactic tree of $\phi$ and how it is changes when constructing  $\hat{\phi}$.
\end{proof}

\paragraph{Formulas for batched Splitter game}
Finally, we will need batched formulas of small batched quantifier rank to express that Splitter wins the batched Splitter game.

\begin{lemma}
\label{lem:swins_formula}
Let $d,r,m\in\mathbb{N}$  with $d \ge 1$. Set $b:=\max\{m+1,r\}$. There exists a $b$-batched sentence $\wins_d^{r,m}$ of batched quantifier rank $3d-2$ such that for every $G$ we have $G \models \wins_d^{r,m}$ if and only if the $(r,m)$-game rank of $G$ is at most $d$.

Moreover, for every $d>1$ there exists a  $b$-batched formula $\winsp_{d-1}^{r,m}(x,y_1,\ldots,y_m)$ of \mbox{$b$-batched} quantifier rank $3d-4$ such that for any graph $G$ and vertices $v,u_1,\ldots,u_m \in V(G)$ we have that 
$ G \models \winsp_{d-1}^{r,m}(x,y_1,\ldots,y_m)$ if and only if 
the $(r,m)$-game rank of $B^G_r(v) - \{u_1,\ldots,u_m\}$ is at most $d-1$.
\end{lemma}

\begin{proof}
For $d=1$, Splitter loses if and only if there exists a vertex such that there are at least $m+1$ vertices in its $r$-neighborhood. This is easily expressed by an $(m+1)$-batched existential sentence, which is of batched quantifier rank $1$. By taking its negation, we get the result. 

Let $d> 1$.  
Let $\wins_{d-1}^{r,m}$ be the $b$-batched sentence of batched quantifier rank $3(d-1)-2$ obtained from the induction hypothesis. 
The operation which maps $G$ and $v,u_1,\ldots, u_m$ to  $G':=B^G_r(v) \setminus \{u_1,\ldots, u_m\}$ can be modeled by an $r$-batched interpretation with parameters (that is independent of $G$) of $r$-batched quantifier rank $1$, and therefore, by Lemma~\ref{lem:interp_lemma},  there exists a formula $\winsp_{d-1}^{r,m}(y_0,y_1,\ldots, y_m)$ of $b$-batched quantifier rank $3(d-1)-1$ such that 
 $$B^G_r(v) \setminus \{u_1,\ldots, u_m\} \models \wins_{d-1}^{r,m} \Longleftrightarrow  G \models \winsp_{d-1}^{r,m}(v, u_1,\ldots, u_m).$$ 
Then we can set $$\wins_{d}^{r,m}:=\left(\forall x\right)\left(\exists x_1 \ldots \exists x_m\right) \winsp_{d-1}(x,y_1,\ldots,y_m),$$
which expresses that for every move of Localiser there exist $m$ vertices such that Splitter wins the remainder of the game in $d-1$ rounds. Since we added $2$ quantifier batches, the result follows.
\end{proof}

\subsection{Proof of Theorem~\ref{thm:alternations}}

Theorem~\ref{thm:alternations} will follow from the following lemma.

\begin{lemma}
\label{lem:alternation}
For every $d\in \N$ with $d\ge 1$ and 
sentence $\phi$ of quantifier rank $q$  
there exists $r=r(d,q)$ such that for every $m\in \N$ there exists $p=p(d,q,m)$ and a $p$-batched sentence $\phi'$ of batched quantifier rank at most $3d$ such that for any colored graph $G$ of  $(r,m)$-game rank $d$ we have that 
\[ G \models \phi \quad\Longleftrightarrow\quad  G \models \phi'. \]
Moreover, for every fixed $d$, $r(d,q)$  is elementary in $q$ and  $p(d,q,m)$ is elementary in $q$ and $m$.
\end{lemma}

\begin{proof}
Before proceeding with the proof, we will discuss the general strategy for producing $\phi'$ from $\phi$. Let $\phi$ be an arbitrary sentence. After setting $r$ suitably and fixing arbitrary $m$, we proceed as follows. Using Gaifman's theorem, we turn $\phi$ into a Boolean combination of basic local sentences of locality radius $r':=7^q$, where $q$ is the quantifier rank of $\phi$. Note that we use two values related to locality in the proof -- the locality radius $r'$ and the value $r$ for the $(r,m)$-game rank. One should think of $r$ as being larger than $r'$.  Each basic local sentence $\tau$ is of the form
\begin{equation}
\label{eq:bls}
 \tau:= \exists x_1 \ldots \exists x_k \bigwedge_{1\le i<j \le k} (\dist(x_i,x_j)>2r') \land \bigwedge_{i=1}^k \alpha(x_i),
\end{equation}
where $\alpha(x)$ is $r'$-local.

We will argue that in each $\tau$ we can replace  $\alpha(x)$ by formula $\beta(x)$ of $p'$-batched quantifier rank $3d-1$ to obtain sentence $\tau'$ such that for any graph $G$ of $(r,m)$-game rank at most $d$ we have $G \models \tau \Leftrightarrow G \models \tau'$. Then, we set $p:=\max\{p', k, 2r'\}$. Since $\tau$ is a $p'$-batched sentence and $p\ge p'$, $\tau'$ is also $p$-batched. 
From the structure of (\ref{eq:bls}) it follows that $\tau'$ has $p$-batched quantifier rank $3d$, and so the sentence $\phi'$ obtained from (the Gaifman form) of $\phi$ by replacing each $\tau$ with $\tau'$ also has $p$-batched quantifier rank $3d$, and we have $G \models \phi' \Longleftrightarrow G \models \phi$ on any $G$ of  $(r,m)$-game rank at most $d$, as desired. Moreover, since  $p=\max\{p', k, 2r'\}$, if we can show that $p'$ depends elementarily on $q$ and $m$, then so does $p$, because $k$ and $r'$ depend elementarily on $q$. 
This finishes the discussion of the general strategy.

We set $r(1,q) := 7^q$ and $r(d,q):= r(d-1,f(q))$, where $f$ is the (elementary and nondecreasing) function from Theorem~\ref{thm:Gaifman}. It is easily seen that this function is elementary in $q$ for any fixed $d$.

We now proceed with the proof by induction on $d$. Let $d=1$. Let $\phi$ be a sentence of quantifier rank $q$. Set $r:=r(1,q) = 7^q$. Let $r'=7^q$ be the number obtained from Gaifman's theorem applied to $\phi$. Let $m$ be arbitrary. As was argued in the proof sketch above, we only need to find formula $\beta(x)$ equivalent to $\alpha(x)$ on all graphs of  $(r,m)$-game rank one. 
In such graphs, since $r=r'$, the  $r'$-neighborhood of any vertex has at most $m$ vertices. 
Since $\alpha(x)$ is $r'$-local, for any $G$ and $v \in V(G)$ we have that 
$G \models \alpha(v) \Leftrightarrow B^G_{r'}(v) \models \alpha(v)$.  Whether $B^G_{r'}(v) \models \alpha(v)$ holds depends only on the isomorphism type
of $B^G_{r'}(v)$, and since $|B_{r'}(v)| \le m$ for the graphs we consider, the formula $\beta(x)$ can be written as a Boolean combination of formulas $\delta_1(x),\ldots,\delta_s(x)$, one for each possible isomorphism type. Each such formula guesses the vertices in the $r'$-neighborhood of $x$ and checks their adjacencies. One easily verifies that this guessing and checking can be done by an $m$-batched formula of rank $2$ (even rank $1$ would suffice). Clearly $p':=m$ depends elementarily on the parameters.
Consequently, as discussed at the beginning of the proof, we can set $p(1, q, m):= \max\{p',k,2r'\} = \max\{m,f(q),2\cdot 7^q\}$, which is clearly elementary in $q$ and $m$.
 
Assume $d>1$. Let $\phi$ be a sentence of quantifier rank $q$. Set $r:=r(d,q)$.
Let $m\in\N$ be given.
Apply Gaifman's theorem to $\phi$ in order to obtain an equivalent sentence $\bar{\phi}$ which is a Boolean combination of basic local sentences and let $r':=7^q$ be the locality radius of $\bar{\phi}$. One easily checks that $r' \le r$.
As before, we only need to consider how to convert a subformula $\alpha(x)$ in any basic local sentence $\tau$ into a formula $\beta(x)$ such that for any $G$ of $(r,m)$-game rank $d$ and $v \in V(G)$ we have that 
\begin{equation}
 G \models \alpha(v) \Longleftrightarrow G \models \beta(v).
\end{equation}
 We note that the quantifier rank of $\alpha$ is at most $f(q)$, since $\alpha$ is a subformula of $\bar{\phi}$. Also, since $\alpha$ is $r'$-local, we have for any $G$ and $v \in V(G)$ that 
\begin{equation}
\label{eq:a}
	G \models \alpha(v) \Longleftrightarrow B^G_{r'}(v) \models \alpha(v).
\end{equation}
\paragraph{Construction of $\beta$}
We will describe the construction of  $\beta(x)$ from $\alpha(x)$ in several steps. During the construction we will mostly ignore the batch sizes and batched ranks of formulas involved; the detailed analysis of these will be provided later.

\textsf{Step I:} For any $G$ and $v,u_1,\ldots,u_m$, let $H(G,v,u_1,\ldots,u_m)$ be the graph $B^G_{r'}(v) \ast \{u_1,\ldots, u_m\}$ in which vertices are marked by unary predicates as follows: Each $u_i$ (if it is in $B^G_{r'}(v)$) is marked by $L_i$, and all neighbors of $u_i$ which are in $B^G_{r'}(v)$ are marked with $N_i$. Moreover, $v$ is marked by $L_x$ and all its neighbors by $N_x$. Then there exists a sentence $\alpha'$ such that for any $G$ and $v,u_1,\ldots,u_m$
we have that 
\begin{equation}
\label{eq:b}
 B^G_{r'}(v) \models \alpha(v) \Longleftrightarrow H(G,v,u_1,\ldots,u_m) \models \alpha'.
\end{equation}
Such $\alpha'$ is easily obtained from  $\alpha(x)$ by syntactic manipulation -- we replace each occurrence of $E(z_1,z_2)$ in $\alpha$ (here $z_1,z_2$ are any variables) by $$E(z_1,z_2) \lor \bigvee_{1 \le i \le m} \big((L_i(z_1) \land N_i(z_2))\lor (N_i(z_1)\land L_i(z_2))\big),$$ each occurrence of $x=z$ by $L_x(z)$ and each occurrence of $E(x,z)$ by $N_x(z)$.
Note that $\alpha'$ 
and $\alpha$ have the same quantifier rank, which is upper-bounded by $f(q)$.

\textsf{Step II:}
Since $r =r(d,q)=r(d-1,f(q))$ and $\alpha'$ has quantifier rank at most $f(q)$, by the induction hypothesis there exists a batched sentence $\alpha''$ of batched rank $3(d-1)$ such that for any $H$ of  $(r,m)$-game rank $d-1$ we have that 
$ H \models \alpha' \Longleftrightarrow H \models \alpha''$.
In particular, if $G$ and $v,u_1,\ldots,u_m$ are such that $H(G,v,u_1,\ldots,u_m)$ is of  $(r,m)$-game rank $d-1$, then
\begin{equation}
\label{eq:c}
 H(G,v,u_1,\ldots,u_m) \models \alpha' \Longleftrightarrow H(G,v,u_1,\ldots,u_m) \models \alpha''.
\end{equation}

\textsf{Step III:}
Since the operation which from $G$ and $v,u_1,\ldots,u_m$ produces  $H(G,v,u_1,\ldots,u_m)$ can be modeled by an interpretation with parameters, by Lemma~\ref{lem:interp_lemma} the sentence $\alpha''$ can be transformed into a formula $\hat{\alpha} (y_0,y_1,\ldots,y_m)$ such that for any $G$ and any $v, u_1,\ldots, u_m$ we have 
\begin{equation}
\label{eq:d}
  H(G,v,u_1,\ldots,u_m) \models \alpha''  \Longleftrightarrow  G \models \hat{\alpha}(v,u_1,\ldots, u_m).
\end{equation}

\textsf{Step IV:}
Let $\winsp(x,y_1,\ldots,y_m)$ be a formula 
expressing that 
the $(r,m)$-game rank of the graph $B_{r'}(x) - \{y_1,\ldots, y_m\}$ is at most $d-1$. 
We postpone the definition of this formula and the discussion of its batch size and quantifier rank for later, but it will be roughly derived from Lemma~\ref{lem:swins_formula}.
Finally, we~set
\begin{equation}
\label{eq:beta}
  \beta(x):= \exists y_1 \ldots \exists y_m (\winsp(x,y_1,\ldots,y_m)  \land \hat{\alpha}(x,y_1,\ldots,y_m)).
\end{equation}
This finishes the construction of $\beta(x)$.

We claim that for any $G$ of  $(r,m)$-game rank $d$ and $v \in V(G)$ we have that 
$$
 G \models \alpha(v) \Longleftrightarrow G \models \exists y_1 \ldots \exists y_m (\winsp(v,y_1,\ldots,y_m)  \land \hat{\alpha}(v,y_1,\ldots,y_m)).
$$
We first prove the forward direction. Assume that $G \models \alpha(v)$. Since $G$ is of $(r,m)$-game rank $d$, there exist $u_1,\ldots, u_m \in B^G_{r'}(v)$ such that  Splitter wins on $B^G_{r'}(v) - \{u_1,\ldots, u_m\}$ in $d - 1$ rounds, and so we have that $G \models \winsp(v,u_1,\ldots,u_m)$. We need to show that for this choice of $u_1,\ldots,u_m$ we also have $G \models \hat{\alpha}(v,u_1,\ldots,u_m)$.
By Lemma~\ref{lem:game_rank} we know that $H(G,v,u_1,\ldots,u_m)$ is of $(r,m)$-game rank $d-1$. Therefore, from  $G \models \alpha(v)$, we derive that $G \models \hat{\alpha}(v,u_1,\ldots,u_m)$
from (\ref{eq:a}), (\ref{eq:b}), (\ref{eq:c}), and (\ref{eq:d}).

Conversely, assume that there exist $u_1,\ldots, u_m$ such that
\[G \models \winsp(v,u_1,\ldots,u_m)  \land \hat{\alpha}(v,u_1,\ldots,u_m).\]
Then we know that $B^G_{r'}(v) - \{u_1,\ldots, u_m\}$ is of $(r,m)$-game rank $d-1$, and so by Lemma~\ref{lem:game_rank} we have that 
$H(G,v,u_1,\ldots,u_m)$ is also of $(r,m)$-game rank $d-1$. We conclude that $G \models \alpha(v)$ by (\ref{eq:d}), (\ref{eq:c}), (\ref{eq:b}), and (\ref{eq:a}).

\paragraph{The batch size and batched rank of $\beta$} Recall from the discussion at the beginning of the proof that we need to show that $\beta$ is $p'$-batched, of batched quantifier rank at most $3d-1$, and that the dependence of $p'$ on $m$ and $q$ is elementary, for a fixed $d$.

We first analyze the batch size and rank of formula $\hat{\alpha}$.
For convenience we set $$b:= p(d-1,f(q),m),$$ where $p(d-1,\cdot, \cdot)$ is the function obtained from the induction hypothesis and which is elementary in $q$ and $m$. From the induction hypothesis, we know that $\alpha''$ is $b$-batched and has batched quantifier rank $3d-3$. The interpretation which from $G,v,u_1,\ldots,u_m$ creates $H(G,v,u_1,\ldots,u_m)$ consists of an existential formula with $r'= 7^q$ quantifiers and quantifier-free formulas, and so the formula $\hat{\alpha}(y_0,\ldots,y_m)$ has $b'$-batched quantifier rank $3d-2$, where $b' := \max\{b,r'\}$. 

We now describe and analyze the formula $\winsp$.
First, let  $\winsp_{d-1}^{r,m}$ be the  formula from Lemma~\ref{lem:swins_formula}. This formula is $b''$-batched, where $b''=\max\{r,m+1\}$, and has batched quantifier rank $3d-4$.  We can take as formula $\winsp$ the formula $\winsp_{d-1}^{r,m}$ relativized to the $r'$-neighborhood of $x$. Since this relativization can be done using $r'$-batched quantification and $r' \le r$, we can keep $b''$ as the batch size and the resulting formula will be of batched quantifier rank at most $3d-3$.  

To conclude, from the structure of (\ref{eq:beta}) we see that formula $\beta$ can be made $p'$-batched, where $p':=\max\{b',b'',r'\}$, and its batched quantifier rank is $1 + \max\{3d-2, 3d-3\} = 3d-1$, as required. Finally, since $b$ depends elementarily on the parameters $q$ and $m$ (by the induction hypothesis), so do $b'$ and $b''$. Since $r':=7^q$ is also elementary, we conclude that $p'$  depends elementarily on the parameters.
\end{proof}

\begin{proof}[Proof of Theorem~\ref{thm:alternations}]
Let $\C$ be a class of tree rank at most $d$. 
We first consider the case when $\phi$ is a sentence.
The general case will be argued later.

Let  $q$ be the quantifier rank of $\phi$ and let $r:=r(d, q)$ be obtained by applying Lemma~\ref{lem:alternation} to $q$. Since $\C$ has tree rank $d$, it also has  $(r,m)$-game rank at most $d$ for some $m\in\N$, and so by Lemma~\ref{lem:alternation} there exists a $p$-batched sentence $\phi'$ of batched quantifier rank $3d$
such that for each $G \in \C$ we have that 
$G \models \phi \Longleftrightarrow G \models \phi',$
as desired.

We now argue that if $\C$ has elementary tree rank $d$, then $p$ depends on $q$ elementarily. 
If $\C$ has elementary tree rank $d$, then we have  $m:=g(r)$ for some elementary function $g$. Since from Lemma~\ref{lem:alternation} we know that $r(d, q)$ is elementary in $q$, we also know that $m = g(r(d, q))$ depends elementarily on $q$. Finally, since for a fixed $d$, $p=p(d,q,m)$ is an elementary function of $q$ and $m$  by Lemma~\ref{lem:alternation}, the claim follows.

By Lemma~\ref{lem:batched-to-alt}, sentence $\phi'$ is equivalent to 
a $p$-batched sentence $\psi$ 
of batched quantifier rank $3d$ and 
of size $f(p)$, for some elementary function $f$.
Thus, the size of $\psi$ depends elementarily on $q$.
Moreover, $\psi$ can be viewed as a sentence of alternation rank $3d-1$. This proves Theorem~\ref{thm:alternations} in the case when $\phi$ is a sentence.

Now suppose $\phi(x_1,\ldots,x_k)$ is a formula with free variables. We reduce this case to the case of sentences, in a standard way, by interpreting 
the variables $x_1,\ldots,x_k$ as constants, which in turn are represented by unary predicates.

Let $\Sigma$ be the signature of $\CC$,
and let $\Sigma'$ extend $\Sigma$ by unary predicates $U_1,\ldots,U_k$.
Let $\CC'$ be the class of all graphs $G'$ (over the signature $\Sigma'$) that can be  obtained from a graph   $G\in\CC$ (over the signature $\Sigma$)
by interpreting the unary predicates $U_1,\ldots,U_k$ as subsets of $V(G)$ of size $1$.
Then $\CC'$ has the same tree rank as $\CC$.

Let $\phi'$ be the following sentence over signature $\Sigma'$: $$\phi':=\exists x_1\ldots \exists x_k. U_1(x_1)\ldots U_k(x_k)\land \phi.$$
Apply the result to $\phi'$, obtaining a sentence $\psi'$ over the signature $\Sigma'$, of alternation rank $3d-1$ and size elementary in the size of $\phi'$.

Define the formula $\psi(x_1,\ldots,x_k)$ 
as the formula $\psi'$, in which each atom 
$U_i(x)$ is replaced by the atom $x=x_i$.
Then $\psi$ satisfies the required conditions.
\end{proof}

Note that our  proof of  Theorem~\ref{thm:alternations}
is an existential statement:
for a sentence $\phi$ it proves the existence of a sentence $\phi'$ that is equivalent on graphs of bounded tree rank.
Our proof is constructive, in the sense that $\phi'$ 
can be computed from $\phi$, by following the construction in the proof.
However, the running time of this algorithm is non-elementary in the size of $\phi$, due to our use 
of Gaifman's locality theorem. 
We do not know whether the formula $\phi'$ can be computed  in elementary time, given $\phi$.
If this were true, this would allow us to slightly simplify the argument in Section~\ref{sec:mc}, where we construct an fpt-elementary model checking algorithm
for classes of bounded elementary tree rank.
However, as we show, the mere existence of a sentence $\phi'$ as in the statement of Theorem~\ref{thm:alternations} is sufficient to deduce the existence of an elementarily-fpt model checking algorithm.

\subsection{Lower bound on alternation rank}\label{sec:lower-alt}
We conclude this section by proving the following result, which is the converse of Theorem~\ref{thm:alternations}.

\begin{lemma}\label{lem:converse_alternations}
Let $\CC$ be a monotone graph class 
such that there is a $d\in\N$ such that 
every formula $\phi$ is equivalent to a formula $\psi$ of alternation rank $d$ on $\CC$.
Then $\CC$ has tree rank at most $d$.
\end{lemma}

Together with Theorem~\ref{thm:alternations},
Lemma~\ref{lem:converse_alternations} proves 
Theorem~\ref{thm:collapse} stated in the introduction.
To prove Lemma~\ref{lem:converse_alternations}, we use the following lemma, proved by Chandra and Harel~\cite[Lemma 3.9]{chandra1982structure}.
While their result is not formulated for trees with the claimed combinatorial properties, these 
additional assumptions follow directly from the construction in the proof.
Before presenting this result, let us also clarify that when referring to rooted trees/forests, we view them as structures over a signature $\Sigma$ consisting of a unary function \emph{parent}, which is interpreted as the parent function of the corresponding rooted tree/forest (assuming that for a root $r$, \emph{parent}$(r)=r$).

\begin{lemma}\label{lem:chandra_harel}
Fix $d\in\N$.
There is a sentence $\varphi$ such that for every sentence $\psi$ of alternation rank $d$,
there is a rooted tree  $T$ of depth $d+1$ such that $T\models \phi\leftrightarrow\neg\psi$. Moreover, as a graph, $T$  has no vertices of degree $2$ and at least two of its leaves are at distance $d$ from the root.
\end{lemma}

Intuitively,
Lemma~\ref{lem:chandra_harel} states that bounded alternation rank formulas form a strict hierarchy on rooted trees of bounded depth. In order to prove Lemma~\ref{lem:converse_alternations}, we aim to ``scale''   Lemma~\ref{lem:chandra_harel} to subdivisions of bounded depth trees.
As an intermediate step, we 
will use the following lemma, which will also be used in Section~\ref{sec:dense}. Before stating and proving it, we define
\emph{levels} of vertices of a rooted tree $T$: for any $v \in V(T)$, the \emph{level} of $v$ is the number of vertices on the (unique) path from $v$ to the root of $T$. In particular, the level of the root is $0$, the level of the children of the root is $1$ etc.

\begin{lemma}
\label{lemma:ramsey-like}
    There is a function $h\from\mathbb{N}^3\to\mathbb{N}$ such that
    for every $d,k,\ell\in\mathbb{N}$ with $d \ge 1$, if $T$ is a rooted tree of depth $d$ and branching at least $h(d,k,\ell)$ at every non-leaf vertex, then for every $c\colon V(T)\to \{1,\ldots,\ell\}$, there is a subtree $T'$ of $T$ that is isomorphic to $T^d_{k}$ and for every $v,u\in V(T')$ of the same level in $T'$, $c(v)=c(u)$.
\end{lemma}

\begin{proof}
We set $h(d,k,\ell)= k \ell^{d-1}$.
Our proof proceeds by induction on the depth $i$ of the considered trees.
The case where $d=1$ is trivially true.
Now inductively assume that all trees of depth at most $i-1$ and branching at least $h(i-1,k,\ell)$ satisfy the claimed property.  
Consider a (rooted) tree $T$ of depth $i> 1$ and branching at least $h(i,k,\ell)$ and a function $c\colon V(T)\to \{1,\ldots,\ell\}$.
Let $r$ be the root of $T$ and let $v_1,\ldots, v_p$, where $p\geq h(i,k,\ell)$, be the children of $r$ in $T$.
The subtrees $T_{v_1},\ldots, T_{v_p}$ of $T$ have depth at most $i-1$ and branching at least $h(i,k,\ell)\geq h(i-1,k,\ell)$.
Therefore, by induction hypothesis,
for every $j\in\{1,\ldots,p\}$, we may replace each $T_{v_j}$ with its subtree  $T_{j}'$ of depth $i-1$ and branching $k$ such that for every $v,u\in V(T_{j}')$ of the same level in $T_{j}'$, $c(v) = c(u)$.
Now, with each $v_j$ we associate a vector $(c_j^0,c_j^1,\ldots,c_j^{i-2})$, where $c_j^h$ is the common color under $c$ of all the vertices on level $h$ in $T_j'$. (In particular, $c_j^0=c(v_j)$.)
Since the number of different such vectors is $\ell^{i-1}$ and $p\geq k\ell^{i-1}$, there  is a collection of $k$ vertices among $v_1,\ldots,v_p$ whose  associated vectors are equal. By deleting the other $p-k$ children of $r$ and the subtrees rooted at them, we obtain a subtree $T'$ of $T$ of depth $i$ and branching $k$ with the claimed property.
\end{proof}

Before proving Lemma~\ref{lem:converse_alternations},
we show how to use Lemma~\ref{lemma:ramsey-like} in order to find a subdivision of some tree, where, in each level, subdivision paths have equal length. To describe this formally, let us give some additional definitions.
Let $T$ be a rooted tree with no nodes of degree $2$ and let $T'$ be a subdivision of~$T$.
Say that a vertex $v$ of $T'$ is \emph{principal}
if it has degree either $1$ or at least $3$. 
Those vertices correspond to vertices of~$T$,
and we assume that $V(T)\subset V(T')$.
A \emph{long edge} in $T'$ is a path connecting principal vertices whose internal vertices have degree $2$.
The long edges of $T'$ correspond to edges of $T$.
Given a $d\in\mathbb{N}$ with $d\ge 1$, a (rooted) tree $T$ of depth $d+1$,
and $r_1,\ldots,r_d\in\mathbb{N}$,
with $r_1,\ldots,r_d\ge 1$,
we define the
\emph{$(r_1,\ldots,r_d)$-subdivision} of $T$, denoted by $T^{(r_1,\ldots,r_d)}$, to be the (rooted) tree that is the subdivision of $T$ where, for each $i\in \{1,\ldots,d\}$, every long edge in $T^{(r_1,\ldots,r_d)}$ between  principal vertices corresponding to vertices of $T$ of levels $i-1$ and $i$ has length $r_i$.
We show the following:

\begin{lemma}
\label{lemma:canonical-subd-tree}
    Fix $d,k,r\in\mathbb{N}$ with $d\ge 1$.
    For every  ${\le}r$-subdivision $T'$ of $T^{d}_{h(d,k,r+1)}$,
    there are $r_1,\ldots,r_{d-1}\in\{1,\ldots,r+1\}$
    such that $T'$ contains a tree isomorphic to the $(r_1,\ldots,r_{d-1})$-subdivision of $T^{d}_k$ as a subgraph.
\end{lemma}

\begin{proof}
    We define a function $c$ that maps every principal vertex $v$ of $T'$ to the length $c(v)\in\{1,\ldots,r+1\}$ of the long edge connecting $v$ to its parent in $T^{d}_{h(d,k,r+1)}$.
    By Lemma~\ref{lemma:ramsey-like},
    there is a subtree $T''$ of $T^{d}_{h(d,k,r+1)}$ of depth $d$ and branching $k$
    such that for every vertex $v$ of $T''$ of level $i$, there is an $r_i\in\{1,\ldots,r+1\}$ such that $c(v) = r_i$. 
    By discarding all vertices of $T'$ that do not correspond to vertices of $T''$ and all long edges with at least one endpoint in them, we obtain a subtree of $T'$ isomorphic to the $(r_1,\ldots,r_{d-1})$-subdivision of $T^{d}_k$.
\end{proof}

\begin{proof}[Proof of Lemma~\ref{lem:converse_alternations}]
Suppose towards a contradiction that $\CC$ has tree rank at least $d+1$.
Then, there is some $r\in\mathbb{N}$ such that every tree of depth at most $d+1$ is an $r$-shallow topological minor of some graph in $\CC$.
Since $\CC$ is monotone, for every (rooted) tree $T$ of depth at most $d+1$, 
some ${\le}r$-subdivision of $T$ is contained in $\C$.
Moreover, by Lemma~\ref{lemma:canonical-subd-tree}, we have that for every tree $T$ of depth $d+1$, there are $r_1,\ldots,r_d\in\{1,\ldots,r+1\}$ such that $T^{(r_1,\ldots,r_d)}$ is in $\mathcal{C}$.
By Lemma~\ref{lem:chandra_harel},
there is a sentence $\varphi$ such that for every sentence $\psi$ of alternation rank $d$,
there is a rooted tree $T$ of depth $d+1$ such that $T\models \varphi\leftrightarrow \neg\psi$ and additionally $T$ has no vertices of degree $2$ and at least two of its leaves are in distance $d$ from the root.

For the rest of the proof, we fix the sentence $\phi$ and we aim to find a sentence $\psi$ of alternation rank $d$ (both $\phi$ and $\psi$ are over the signature of rooted trees) such that every rooted tree of depth $d+1$ with no nodes of degree $2$ and at least two nodes of level $d$ satisfies $\varphi\leftrightarrow \psi$. This will contradict Lemma~\ref{lem:chandra_harel}.

First, we claim that we may rewrite $\phi$ to a sentence $\phi^\circ$ over the signature of graphs
so that for every rooted tree $T$ of depth $d+1$ with no nodes of degree $2$ and at least two leaves of level $d$,
we have:
\begin{align}\label{eq:phiphi}
  T\models \phi\quad\iff \quad T^{(r_1,\ldots,r_d)}\models \phi^\circ;
\end{align}
here, $T^{(r_1,\ldots,r_d)}$ is treated as a graph (it is an unrooted tree).
This can be done using the following trick.
To define the root of (the unrooted) $T^{(r_1,\ldots,r_d)}$, we ask the existence of a sequence of vertices that form a maximum length path in $T^{(r_1,\ldots,r_d)}$; the length of such path should be exactly $1+2\cdot \sum_{i=1}^d r_i$ (since $T$ has depth $d+1$ and at least two leaves of level $d$)
and we recover the root as the middle vertex of this path.
Therefore, we can turn sentences on 
rooted trees $T$ of depth $d+1$ with no nodes of degree $2$ and at least two nodes of level $d$
to equivalent sentences on unrooted trees of the form $T^{(r_1,\ldots,r_d)}$. Indeed, a vertex $v$ is a parent of a vertex $u$ if and only if they are adjacent and $v$ belongs to the unique path (of length at most~$d$) from $u$ to the root.

By assumption, $\phi^\circ$ is equivalent to a sentence $\psi^\circ$ of alternation rank $d$ on $\CC$. This sentence $\psi^\circ$ is over the signature of graphs.
We turn $\psi^\circ$ to an equivalent sentence $\psi'$ over the signature of rooted trees by simulating the edge relation using the parent function without changing the alternation rank.
We next prove the following claim. As the proof is conceptually easy, but technically somewhat tedious, we only provide a sketch.
\begin{claim*}
  Fix $d,r_1,\ldots,r_d\in\N$. 
  For every sentence $\psi'$ over the signature of rooted trees, there is a sentence $\psi$ (on the same signature) of the same alternation rank, 
  such that 
  for every rooted
  tree $T$ of depth $d+1$
  and no nodes of degree $2$, if $T'$ is the $(r_1,\ldots,r_d)$-subdivision of $T$, we have:
\begin{align}\label{eq:psipsi}
T\models \psi\quad\iff \quad T'\models \psi'.
\end{align}
\end{claim*}

\begin{subproof}[Proof sketch of the claim]
We assume that $d\ge 1$, since otherwise the claim holds trivially.
Keep in mind that both $T$ and $T'$ are rooted trees. For every vertex $v\in V(T')$, we denote by $P_{v}$ the (unique) path in $V(T')$ connecting $v$ to the root of $T'$.
For every non-root vertex $v$ of $T'$, let $(v^0,v^1,\ldots,v^{i-1})$ be the sequence of principal vertices of $T'$ on $P_v\setminus \{v\}$, as they appear when following $P_v$ from the root towards $v$. Also, let $v^{i}$ be the principal vertex of $T'$ for which the path connecting it to the root contains $v$ and no other principal vertex except of $v^0,v^1,\ldots,v^{i-1}$ and itself.
Note that in case $v$ is a principal vertex of $T'$, $v^{i}=v$.
Let $h$ denote the distance between $v$ and $v^{i}$; then $h\in \{0,\ldots,r_{i}-1\}$.
The vertex $v$ is uniquely determined by the tuple $F(v)=(v^0,\ldots,v^{i},h)$.
In case $v$ is the root, we set $F(v)=(v,0)$. To handle the latter case, we additionally set $r_0\coloneqq 1$.

We prove the claim by induction on the structure of $\psi'$. 
For the purpose of induction, we allow $\psi'$ to have free variables $x_1,\ldots,x_k$, which we signify by writing $\psi'(x_1,\ldots,x_k)$.
We show that, for all $I=(i_1,\ldots,i_k)\in\{0,\ldots,d\}^k$ and $J=(h_1,\ldots,h_k)\in\{0,\ldots, r_{i_1}-1\}\times\cdots\times \{0,\ldots,r_{i_k}-1\}$, we can construct formulas 
$\psi_{I,J}(\bar{x}_1,\ldots,\bar{x}_k)$,
where
$\bar{x}_j = (x_j^0,\ldots,x_j^{i_j})$ for $j=1,\ldots,k$,
whose alternation rank is equal to the alternation rank of $\psi'$,
such that the following property holds. 
For every $T$ and $T'$ as in the statement of the claim,
all vertices $v_1,\ldots,v_k\in V(T')$ 
and for
every choice of $I$ and $J$,
and
$v_1^0,\ldots,v_1^{i_1},\ldots,v_k^0,\ldots,v_k^{i_k}\in V(T)$ satisfying $(v_j^0,\ldots,v_j^{i_j},h_j)=F(v_j)$
for $j=1,\ldots,k$ it holds that 
\begin{align*}
  T'\models\psi'(v_1,\ldots,v_k)\quad\iff\quad T\models\psi_{I,J}(v_1^0,\ldots,v_1^{i_1},\ldots,v_k^0,\ldots,v_k^{i_k}). 
\end{align*}
Note that if $\psi'$ is a sentence, that is, $k=0$,
we obtain the conclusion of the claim.

In the induction base, we consider atomic formulas. These are formulas of the form\footnote{For $i \in \mathbb{N}$ and $v \in T'$, we set \emph{parent}$^i(v)$ to be the $i$-fold application of the function \emph{parent} of $v$ and we set \emph{parent}$^0(v)=v$.} \emph{parent}$^\alpha(x)=$\emph{parent}$^{\beta}(y)$, for non-negative integers $\alpha$ and $\beta$. We construct $\psi_{I,J}$ by hand by checking whether the choices of $i_x,i_y,h_x,h_y$ given by $I$ and $J$ (and the given equalities between $x^0,\ldots,x^{i_x}$ and $y^0,\ldots,y^{i_y}$ and the numbers $\alpha,\beta$) certify these atomic formulas. 
To encode this formally, we fix some additional notation: we use $i_x^\star$ to denote the index for which \emph{parent}$^\alpha(x)$ is in $V(P_{v^{i_x^\star}})\setminus V(P_{v^{i_x^\star-1}})$ (in case it is the root, we set $i_x^\star\coloneqq 0$) and $h_x^\star$ to denote the distance of \emph{parent}$^\alpha(x)$ from $v^{i_x^\star}$.
For completeness, we give equivalent expressions for $i_x^\star$ and $h_x^\star$. We set $i_x^\star$ be the maximum $i\in\{1,\ldots,i_x\}$ such that $\alpha+h_x<\sum_{j=i}^{i_x}r_{j}$ and if such $i$ does not exist, i.e., $\alpha+h_x\ge \sum_{j=1}^{i_x}r_j$, we set $i_x^\star\coloneqq 0$. For $h_x^\star$, we distinguish the following cases: (i) if $\alpha=0$, we set $h_x^\star\coloneqq h_x$, (ii) if $\alpha+h_x\ge \sum_{j=1}^{i_x}r_j$, then $h_x^\star=0$ (recall that this is the case where \emph{parent}$^\alpha(x)$ is the root and $i_x^\star=0$), and (iii) otherwise $h_x^\star\coloneqq \alpha +h_x - \sum_{j=i_x^\star+1}^{i_x}r_j$ (which is equal to $\alpha+h_x$ when $i_x^\star=i_x$). We define $i_y^\star$ and $h_y^\star$ analogously. Then, we check whether $i_x^\star=i_y^\star$, $x^{i_x^\star}=y^{i_y^\star}$, and $h_y^\star=h_x^\star$.

In the induction step, we need to handle the case when $\psi'$ is a Boolean combination of two formulas,
and when $\psi'$ is of the form $\exists x\,\delta'(x,\tup y)$. The case of Boolean combinations is handled easily, using the formulas obtained by inductive assumption.
So suppose that $\psi'$ is of the form $\exists x\,\delta'(x,\tup y)$.
Apply the induction hypothesis to $\delta'(x,\tup y)$,
obtaining formulas $\delta_{I,J}$ as described above.
To simplify notation, assume that $\tup y$ is empty, thus $\psi'$ is of the form $\exists x\, \delta'(x)$;
the general case proceeds analogously.
By the induction hypothesis, we obtain for each
$i\in\{0,\ldots,d\}$
and each $h\in\{0,\ldots,r_{i}-1\}$ a
formula $\delta_{i,h}(x^0,\ldots,x^{i})$ of the same alternation rank as $\delta'$,
such that for every $v\in V(T')$ and every choice of $i,h$ and $v^0,\ldots,v^{i}$ such that $(v^0,\ldots,v^{i},h)= F(v)$, it holds that $T'\models \delta'(v)\iff T\models \delta_{i,h}(v^0,\ldots,v^{i})$.
The formula $\psi$ is then defined as:
$$\bigvee_{i\in\{0,\ldots,d\}}\bigvee_{h\in\{0,\ldots,r_{i}-1\}}\exists x^0
\ldots \exists x^{i} (\rho_{i}(x^0,\ldots,x^{i})\wedge \delta_{i,h}(x^0,\ldots,x^{i})),$$
where $\rho_{i}(x^0,\ldots,x^i)$ is a quantifier-free first-order formula on the signature of rooted trees,
such that for every $v^0,\ldots,v^{i}\in V(T)$,
$T\models \rho_{i}(v^0,\ldots,v^{i})$ if and only if $v^0$ is the root of $T$ and $v^{j-1}$ is the parent of $v^j$ in $T$ for all $j\in\{1,\ldots,i\}$.
It follows from the inductive assumption that $\psi$ and $\psi'$ have equal  alternation ranks.

This completes the sketch of the construction of $\psi$, given $\psi'$.
\end{subproof}

Let $\psi$ be the formula obtained in the claim from $\psi'$. Then $\psi$ is a sentence on the signature of rooted trees that has alternation rank $d$.
Combining \eqref{eq:phiphi} with \eqref{eq:psipsi},
and since $\phi'$ is equivalent to $\psi'$ 
on $\cal C$, which contains for every tree $T\in\cal T_{d+1}$ the tree $T^{(r_1,\ldots,r_d)}$ for some $r_1,\ldots,r_d\in\{1,\ldots,r+1\}$,
we conclude that $\psi$ is equivalent to $\phi$ on rooted trees of depth $d+1$ with no nodes of degree $2$ and at least two nodes of level $d$.
This is a contradiction with Lemma~\ref{lem:chandra_harel}. This proves Lemma~\ref{lem:converse_alternations}, and completes the proof of Theorem~\ref{thm:collapse}.
\end{proof}

\section{Model checking}
\label{sec:mc}

In this section, we prove the existence of 
an elementarily-fpt model checking algorithm
for classes of bounded elementary tree rank.
We show the following statement, which, in particular, implies Theorem~\ref{thm:main}.

\begin{theorem}\label{thm:mc}
    Let $\CC$ be a class of  elementary tree rank $d$, for some $d\in\N$.
    Then model checking first-order logic is elementarily-fpt on  $\CC$.
    More precisely, there exists an elementary function $f$
    and an algorithm that, given a graph $G\in\CC$ and 
    first-order sentence $\phi$,
    determines whether $G\models \phi$ 
    in time $f(|\phi|)\cdot |G|.$
\end{theorem}

The idea is as follows. We first prove that every class of bounded tree rank has bounded expansion.
We then use the result of \DKT~\cite{dvovrak2013testing},
which states that every sentence can be tested in linear time on graphs from a fixed class of bounded expansion. By inspecting the proof, we observe that the running time of this algorithm is elementarily-fpt for sentences of  alternation rank bounded by a fixed constant. However, Theorem~\ref{thm:alternations} does not allow us to directly convert (in elementary time) a given sentence $\phi$ into a  sentence of bounded alternation rank.
Nevertheless, it implies that in order to determine whether two vertices satisfy the same formulas $\psi(x)$ of a given quantifier rank $q$, it is enough to test 
whether they satisfy the same formulas of  alternation rank bounded by a fixed constant (independent of $q$). It is known that model checking can be reduced to the problem 
of deciding whether two vertices satisfy the same formulas of a given quantifier rank. Combining those observations yields the required result. We now present the details.

\medskip
We start by showing that every class of bounded tree rank has bounded expansion, whose definition we now recall.
For a graph $G$ and number $r\in\N$, 
let $\tilde\nabla_r(G)$ denote the maximum 
average degree, $2|E(H)|/|V(H)|$, among all $r$-shallow topological minors  $H$ of $G$.
A class $\CC$ has \emph{bounded expansion}
if for every $r\in\N$ there is a constant $c_r$ such that  $\tilde\nabla_r(G)\le c_r$ for all $G\in\CC$. Equivalently, one can postulate that $\tilde\nabla_r(\CC)\le \infty$ for every $r\in \N$, where we define $\tilde\nabla_r(\CC)\coloneqq \sup_{G\in \CC} \tilde\nabla_r(G)$.

\begin{proposition}\label{prop:be}Fix $d\in\N$.
    If $\CC$ is a class of tree rank at most $d$,
    then $\CC$ has bounded expansion.
    Moreover, if $\CC$ has elementary tree rank at most $d$, then there is an elementary function $f\from\N\to\N$ such that $\tilde\nabla_r(G)\le f(r)$, for all $G\in\CC$.
\end{proposition}

First, we show that classes of bounded tree rank have bounded minimum degree.
This follows immediately from the next lemma, which is a classic result of Chv\'atal~\cite{Chvatal77}\footnote{Chv\'atal's Lemma is sligthly stronger, it asserts that a graph of minimum degree at least $\delta$ contains every tree on $\delta+1$ vertices. Here we only need the simpler variant proved below.}.

\begin{lemma}\label{lem:degeneracy-forests}
    A graph of minimum degree at least
    $\delta$ contains every tree on $\delta$ vertices as a subgraph.
\end{lemma}
\begin{proof}
    We proceed by induction on $\delta$.
    For $\delta=1$ the statement holds trivially.
    Suppose $\delta>1$, and that the statement holds for $\delta-1$.
    Let $G$ be a graph of minimum degree at least $\delta$, 
    and $T$ be a tree on $\delta$ vertices. We show that $G$ contains $T$ as a subgraph.    
    
    Let $v$ be any leaf of $T$, and $v'$ its unique neighbor.
By inductive assumption, $G$ contains $T-v$ as a subgraph,
so assume that $T-v\subset G$.
Since $v'$ has degree at least $\delta$ in $G$ and $|V(T-v)|=\delta-1$, there is some vertex $w\in V(G)-V(T-v)$ that is adjacent to $v'$ in $G$. We can identify the leaf $v$ of $T$ with the vertex $w$ of $G$, demonstrating that 
$T$ is a subgraph of $G$.
\end{proof}

\begin{proof}[Proof of Proposition~\ref{prop:be}]
    Let $\CC$ be a class of tree rank at most $d$, and let $r\in \N$. We prove that 
    $\tilde\nabla_r(\CC)<\infty$. Without loss of generality, $\CC$ is monotone.

    By assumption, there is a tree $T$ of depth ${d+1}$
    such that $T\not\in \topminors_r(\CC)$.
    By Lemma~\ref{lem:degeneracy-forests}, 
    all graphs in $\topminors_r(\CC)$ have minimum degree less than $|T|$.

    It is well known that every graph $H$ of average degree
    at least $k$ has a subgraph $H'$ of minimum degree at least $k/2$ (see e.g.~\cite[Proposition~1.2.2]{Diestel}). It follows that all graphs in $\topminors_r(\CC)$ have average degree bounded by $2|T|$.
Hence, $\tilde\nabla_r(\CC)\le 2|T|$. 
This proves the first claim of the proposition. The second claim, about elementary tree rank, follows immediately by noting that $|T|$ is assumed to be bounded by an elementary function of~$r$.
\end{proof}

The main result of \cite{dvovrak2013testing} states 
that model checking can be solved in linear time, 
for every fixed sentence $\phi$ and class $\CC$ of bounded expansion.
The proof proceeds by the method of \emph{quantifier elimination}, and actually proves a stronger statement,
for formulas $\phi(x_1,\ldots,x_k)$ with free variables. In particular, the following result is proved.

\begin{theorem}\label{thm:dkt}
    Let $\CC$ be a class of bounded expansion.
    There is a function $h$ and an
    algorithm that, given a graph $G\in\CC$ and formula $\phi(x_1,\ldots,x_k)$ 
    computes in time $h(\phi)\cdot |G|$ a data structure that answers in time $h(\phi)$ queries of the form:
    given $v_1,\ldots,v_k\in V(G)$, does $G\models \phi(v_1,\ldots,v_k)$?
\end{theorem}

We will need an improved bound on the function $h$. This is obtained in the next lemma, which is proved by analysing and refining the proof of Theorem~\ref{thm:dkt}. 
\begin{lemma}\label{lem:dkt-constant}
	Suppose that for some integer $s \in \N$ it holds that $\tilde\nabla_r(\CC) \leq \tower(s, r)$ for all $r \in \N$.
	Then the value $h(\phi)$
    in Theorem~\ref{thm:dkt} can be bounded by $\tower((s + 5)(q + 1), \poly(|\phi|))$,
    where $q$ is the alternation rank of 
    the formula $\phi(x_1,\ldots,x_k)$.
\end{lemma}

For completeness, we give
the proofs of both Theorem~\ref{thm:dkt} and Lemma~\ref{lem:dkt-constant} in Appendix~\ref{app:qe}. We now proceed with the proof of Theorem~\ref{thm:mc}.
\medskip

We recall the following fundamental notion from logic.
Fix a number $q\in\N$.
Two vertices $u,v$ of a structure $G$ 
are \emph{$q$-equivalent}, written $u\equiv_q v$, if $G\models \psi(u)\leftrightarrow \psi(v)$ for every formula $\psi(x)$ of quantifier rank at most $q$.

It is known that for every fixed number $q$,
the number of $\equiv_q$-equivalence classes is bounded by a constant independent of the graph $G$.
However, finding the equivalence classes, given $G$,
is just as hard as the model checking problem
for formulas of quantifier rank at most $q$.
More precisely, in \cite{gajarsky2022differential} Gajarsk\'y, Gorsky, and Kreutzer showed
that in order to efficiently test whether a given sentence of quantifier rank at most $q$ holds in a graph $G$, it is enough to efficiently solve the following problem, for the graph $G$
expanded with $q$ constant symbols, interpreted as vertices of $G$:
For a fixed number $q$, 
compute a set $R_q\subset V(G)$ containing at least one representative of each $\equiv_q$-equivalence class, such that $R_q$ has size bounded by some constant (depending only on $q$, and not on $G$). 
We will use this observation to prove that model checking is elementarily-fpt on classes of bounded tree rank. 

Thus, it suffices to show that we can efficiently 
compute a set $R_q$ as above, given $G$.
By Theorem~\ref{thm:alternations},
for classes of tree rank $d$,
the equivalence relation $\equiv_q$ 
is refined by the equivalence relation 
defined analogously, but where instead 
of formulas $\psi(x)$ of quantifier rank at most $q$,
we consider formulas of alternation rank at most $3d$.
By Lemma~\ref{lem:dkt-constant}, 
such formulas can be evaluated in elementarily-fpt time. This way, we can efficiently compute a refinement of $\equiv_q$ with a bounded number 
of equivalence classes, which, by the observation of \cite{gajarsky2022differential} is sufficient to obtain an efficient model checking algorithm.

We now make this argumentation precise.

\medskip

For a class  $\CC$ of structures and $k\in\N$, let $\CC[k]$ denote the class of 
tuples $(G,v_1,\ldots,v_k)$,
where $G\in\CC$  and $v_1,\ldots,v_k\in V(G)$.
We view such a tuple as a structure over the signature of $\CC$ expanded with $k$ constant symbols $c_1,\ldots,c_k$, interpreted as the elements $v_1,\ldots,v_k$.

Fix a class $\CC$ of structures and numbers $q,p\in\N$.
A \emph{rank $q$ selection of order $p$} in a structure $G$ is a set $R$ of at most $p$ 
vertices of $G$ which contains some representative of every $\equiv_q$-equivalence class of $G$
(note that the same equivalence class may be represented more than once in $R$).
Let $g\from\N\times\N\to \N$ be a function.
A \emph{selector} of complexity $g$
for $\CC$ 
is an algorithm that, given $q,k\in\N$ and $(G,v_1\ldots,v_k)\in\CC[k]$
runs in time $g(q,k)\cdot |G|$ and 
returns a rank $q$ selection of order $g(q,k)$ for~$(G,v_1,\ldots,v_k)$.

\begin{lemma}\label{lem:selectors-suffice}
    Fix a class $\CC$ of structures
    and a function $g\from \N\times\N\to\N$.
    Suppose there is  selector of complexity $g$ for the class $\CC$.
    Then, given a structure $G\in\CC$ and a sentence $\phi$
    of quantifier rank $q$,
    it can be decided whether $\phi$ holds in $G$
     in time $\Oh(g(q,q)^q\cdot |\phi|)\cdot |G|$.
\end{lemma}

\begin{proof}Fix $\CC$ and a selector of complexity $g$ for $\CC$. We may assume that the function $g$ is monotone in each coordinate.

We prove the statement by induction on $\phi$.
In the base case, $\phi$ is an atomic formula
(equality check or relation symbol applied to constant symbols), and can be decided in time $\Oh(1)$\footnote{Here we assume the model of computation on relational structures where given a tuple of elements, we can in constant time decide membership of this tuple to any relation present in the structure. In case of graph classes of bounded expansion (which encompass classes of bounded tree rank, by Theorem~\ref{prop:be}), this can be easily emulated in the standard RAM model. Namely, it is well-known that in classes of bounded expansion one can compute, in linear time, edge orientations with bounded outdegree. Then to check whether two vertices $u,v$ are adjacent it suffices to check whether $v$ belongs to the list of outneighbors of $u$ or $u$ belongs to the list of outneighbors of $v$; and this takes constant time. A similar trick can be also applied to classes of structures of bounded expansion with arity larger than $2$; see e.g.~\cite{aggregate-queries} for details.}.
In the inductive step, 
if $\phi$ is a Boolean combination of two formulas $\phi_1,\phi_2$, then the conclusion follows 
from the inductive assumption applied to $\phi_1$ and $\phi_2$. Same reasoning applies when $\phi$ is a negation of a formula $\phi_1$.
So assume $\phi$ is of the form $\exists x\,\psi(x)$,
for some formula $\psi(x)$ of quantifier rank at most $q-1$.

Let $G\in\CC$. 
Run the given selector on $G$, for $k=0$ and $q$ equal to the quantifier rank of $\phi$.
This results in a set $R\subset V(G)$ with $|R|\le g(q,k)$.
Observe that the following equivalence holds:
\begin{align}\label{eq:equiv}
    G\models \phi\quad\Longleftrightarrow\quad 
    G\models \psi(v)\text{\quad for some $v\in R$}.        
\end{align}
The leftwards implication is obvious.
For the rightwards implication,
suppose that $G\models \phi$.
Then there is some $u\in V(G)$ such that $G\models \psi(u)$. By definition of $R$, there is some $v\in R$ 
such that $v\equiv_q u$.
Since $\psi(x)$ has quantifier rank at most $q$
and $v\equiv_q u$,
it follows that $G\models \psi(v)$.
This proves the rightwards implication.

For $v\in V(G)$, 
let $G_v$ denote the structure obtained 
by expanding $G$ with a constant symbol $c$, which is interpreted as the element $v$.
Let $\psi'$ be the sentence obtained from $\psi(x)$ by replacing the variable $x$ by the constant symbol $c$. Then 
\begin{align}\label{eq:psi'}
    G\models \psi(v)\quad\iff\quad G_v\models \psi'.     
\end{align}
The sentence $\psi'$ has quantifier rank $q-1$ and size $|\psi'|=|\psi|$.

Let $\CC'$ denote the class $\CC[1]$.
Observe that $\CC'[k]=\CC[k+1]$ for all $k\in\N$, up to a renaming of constant symbols. Therefore, the selector for $\CC$ of complexity $g$ can be cast to obtain a selector for $\CC'$ of complexity $g'\from\N\times\N\to \N$,
where $g'(q,k)=g(q,k+1)$.

By inductive assumption applied to $\CC'$ and $\psi'$,
for each fixed $v\in V(G)$ we can test whether 
$G_v\models \psi'$ in time $\Oh(g'(q-1,q-1)^{q-1}\cdot |\psi'|)\cdot |G|$.

By \eqref{eq:equiv} and \eqref{eq:psi'}, to test whether $G\models\phi$ holds, it is enough to 
test whether $G_v\models\psi'$ holds, for each $v\in R$. As $|R|\le g(q,q)$, this can be done jointly in time 
$\Oh(g(q,q)^q\cdot |\phi|)\cdot |G|$.
\end{proof}

\begin{lemma}\label{lem:selectors-exist}
    Let $\CC$ be a class of graphs of elementary tree rank at most $d$.
    There is an elementary function $g\from \N\times \N\to\N$ such that $\CC$ has a selector of complexity $g$.
\end{lemma}
\begin{proof}
    By the ``moreover'' part of  Theorem~\ref{thm:alternations}
    there is an elementary function $f\from\N\to\N$
    such that every formula $\phi(x)$ 
    is equivalent to a formula $\phi'(x)$ 
    of alternation rank $3d$, with $|\phi'|\le f(|\phi|)$.

    Consider a parameter $k$ and let $\Delta$ denote the set of all formulas $\psi(x)$
over the signature of $\CC[k]$
of alternation rank at most $3d$ and size at most $f(|\phi|)$, up to syntactic equivalence.
Recalling that $d$ is fixed, the same standard argument as the one used in the proof of Lemma~\ref{lem:batched-to-alt} allows us to compute the set $\Delta$ in elementary time, given $k$ and $f(|\phi|)$. Thus, since $f$ is elementary, $\Delta$ can be computed in
elementary time, given~$k$ and $\phi$.

    Fix $(G,v_1,\ldots,v_k)\in\CC[k]$.
    For two vertices $u,v$ of $G$,
 write $u\equiv_\Delta v$ if $G$ satisfies the formula $\psi(u) \leftrightarrow\psi(v)$ for every formula $\psi(x)\in \Delta$ (where the constant symbols $c_1,\ldots,c_k$ occurring in $\psi$ are interpreted as $v_1,\ldots,v_k$).
By Theorem~\ref{thm:alternations}, we have that 
$u\equiv_\Delta v$ implies $u\equiv_q v$.
The equivalence relation $\equiv_\Delta$ has at most $2^{|\Delta|}$ equivalence classes, and $|\Delta|$ is bounded by an elementary function of $|\phi|$.

Let $h$ be as in Theorem~\ref{thm:dkt}.
By the `moreover' part of Proposition~\ref{prop:be}, there is an $s \in \N$ such that $\tilde\nabla_r(G) \leq \tower(s,r)$ for all $r \in \N$. 
Hence Lemma~\ref{lem:dkt-constant} applies and so $h(\psi)$ is bounded by an elementary function of $|\phi|$ and $k$.
Let $c\coloneqq \max_{\psi\in\Delta}h(\psi)$.
Then $c$ is also bounded by an elementary function of $|\phi|$ and $k$.

It follows from Theorem~\ref{thm:dkt} applied to each of the formulas $\psi(x)$, that given $(G,v_1,\ldots,v_k)$ with $G\in\CC$ in time $\textit{poly}(c\cdot |\Delta|)\cdot |G|$ we can compute a data structure that answers in time $\Oh(c)$ queries of the form: given $v\in V(G)$ and $\psi(x)\in\Delta$, does $G\models \psi(v)$?
(Again, the constants $c_1,\ldots,c_k$ occurring in $\psi$ are interpreted as $v_1,\ldots,v_k$).

Therefore, by iterating through all vertices $v\in V(G)$ and querying the above data structure 
for each $\psi(x)\in\Delta$,
we can compute in time $\textit{poly}(c\cdot 2^{|{\Delta}|})\cdot |G|$ 
a (minimal) set $R$ of representatives of the equivalence 
relation $\equiv_\Delta$.

 To summarize, we have that:
 \begin{itemize}
    \item The set $\Delta$ can be computed in elementary time, given $q$ and $k$. 
    \item $|R|\le 2^{|\Delta|}$,
    \item $R$ can be computed in time $\textit{poly}(c\cdot 2^{|{\Delta}|})\cdot |G|$, given $(G,v_1,\ldots,v_k)\in\CC[k]$ and $\Delta$,
    \item $R$ contains some representative of every $\equiv_q$-equivalence class.
 \end{itemize}

 Therefore, $R$ is a rank $q$ selection 
 of order at most $2^{|\Delta|}$.
 Moreover, it follows from the above that the algorithm computing $R$ described above is a selector for $\CC$ 
 of complexity bounded by some elementary function $g\from\N\times\N\to \N$.
\end{proof}

 Theorem~\ref{thm:mc} 
follows by combining Lemma~\ref{lem:selectors-exist} with Lemma~\ref{lem:selectors-suffice}.

\begin{proof}[Proof of Theorem~\ref{thm:mc}]
Let $\CC$ be a class of graphs of elementary tree rank at most $d$.
By Lemma~\ref{lem:selectors-exist}, 
 the assumptions of Lemma~\ref{lem:selectors-suffice} are satisfied for an elementary function $g$.
 This yields an algorithm that tests whether a given structure $G\in\CC$ satisfies a given sentence $\phi$ of quantifier rank $q$
 in time $\Oh(g(q,q)^q\cdot |\phi|)\cdot |G|$.
 Since $g(q,q)^q\cdot |\phi|$ is elementary in $\phi$,
 this proves that model checking is elementarily-fpt on $\CC$.
\end{proof}

\section{Lower bounds}\label{sec:lowerbound}

In this section we provide our complexity lower bound for classes of large tree rank. To state the result in a detailed form, we need the following definition that provides a finer perspective on the concept of elementary running time.
For $d\in\N$ and a graph class $\CC$, we say that $\CC$ admits \emph{$d$-fold elementary model checking}
if there is an algorithm that for a given $G\in \mathcal{C}$ and sentence $\varphi$,
decides whether $G\models \varphi$ in time $\tower(d,c\cdot |\varphi|)\cdot {|V(G)|}^{c'}$, for some constants $c$ and~$c'$.
We prove the following result, which immediately implies Theorem~\ref{thm:lowerbound-intro}.
\begin{theorem}\label{thm:lower-bound}
Assuming ${\sf FPT}\neq {\sf AW}[*]$, every monotone graph class with $d$-fold elementary model checking has tree rank at most~$d+5$.
\end{theorem}

Let $m\in\mathbb{N}$ with $m\geq 1$.
A \emph{tree of depth $1$ over $[m]$} is an element of $[m]$.
For $i > 1$, a \emph{tree of depth $i$ over $[m]$} is a rooted tree in which the subtrees rooted at the children of the root are non-isomorphic trees of depth $i-1$ over $[m]$.
For $i \ge 1$, a \emph{forest of depth $i$ over $[m]$} is a collection of pairwise non-isomorphic trees of depth $i$ over $[m]$.
Recall that we view rooted forests and trees as structures equipped with the parent function. Forests and trees over $[m]$ are viewed as structures on the signature consisting of the unary function \emph{parent} and $m$ unary relation symbols (that correspond to colors on the leaves).
We refer to such signature as a signature of \emph{rooted forests with $m$ colors}.
To shift from the signature of rooted forests with colors to the signature of (uncolored and unrooted) forests, we will use the following easy lemma.
The \emph{depth} of an unrooted tree $T$ is the minimum $k$ 
for which there is a way to root $T$ such that the obtained rooted tree has depth $k$.
The \emph{depth of an unrooted forest} is the maximum depth of its components.

\begin{lemma}
\label{lem:uncoloring-forests}
    Fix $d,m\in\mathbb{N}$ with $d\ge 1$ and a signature $\Sigma$ of rooted forests with $m$ colors.
    For every $\Sigma$-sentence $\varphi$ there is a sentence $\varphi'$ over the signature of graphs, such that for every forest $F$ of depth at most $d$ over $[m]$, there is a forest $F'$ (over the signature of graphs) such that $F\models \phi$ if and only if $F'\models\phi'$.
    Moreover, $|\phi'|=\Oh(m\cdot d\cdot  |\phi|)$, $|F'|=\Oh(m\cdot |F|)$, $F'$ has depth at most $d+1$ and no nodes of degree $2$, and $\phi'$ and $F'$ can be computed from $\phi$ and $F$ in time $\Oh(m\cdot d\cdot  |\phi|)$ and $\Oh(m\cdot |F|)$ respectively.
\end{lemma}

\begin{proof}
Let $C_1,\ldots,C_m$ be the interpretations in $F$ of the unary predicate symbols of $\Sigma$ and keep in mind that, by construction, every leaf of $F$ belongs to exactly one $C_i$.

We first define the unrooted forest $\hat{F}$ with $m+2$ colors by introducing an additional unary predicate $C_{m+1}$ that is interpreted as the root
and a unary predicate $C_{m+2}$ that is interpreted as the set of all nodes of $F$ that are neither the root nor leaves.
Observe that we can define a sentence $\hat{\phi}$ obtained from $\phi$ by encoding the fact that a node $v$ is a parent of a node $u$ by asking the existence of an edge between $u$ and $v$ and a path from $u$ to a node in $C_{m+1}$ that contains $v$ (keep in mind that in $\hat{F}$, $C_{m+1}$ is the singleton consisting of the root of $F$). It is easy to see that $F\models \phi$ if and only if $\hat{F}\models \hat{\phi}$ and that the size of $\hat{\phi}$ is $\Oh(d\cdot |\phi|)$.

To obtain an (uncolored) forest $F'$ from $\hat{F}$, we add $i+2$ pendant nodes to each node $v$ of $\hat{F}$ that belongs to $C_i$, for each $i\in\{1,\ldots,m+2\}$.
Observe that $|F'|=\Oh(m\cdot |F|)$ and $F'$ has depth at most $d+1$ and no nodes of degree $2$.
The sentence $\phi'$ is obtained from $\phi$ by restricting the quantification to nodes of degree at least $3$ and by encoding whether a node belongs to $C_i$ by asking whether it is adjacent to $i+2$ nodes of degree $1$.
Observe that $\hat{F}\models \hat{\phi}$ if and only $F'\models \phi'$ and that $|\phi'|=\Oh(m\cdot|\hat{\phi}|)=\Oh(m\cdot d\cdot |\phi|)$. Also, note that the construction of $\phi'$ and $F'$ can be done within the claimed running time.
\end{proof}

Observe that for every fixed $d,m\in\mathbb{N}$ with $d,m \ge 1$, there are precisely $\tower(d-1,m)$ non-isomorphic trees of depth $d$ over $[m]$ and precisely $\tower(d,m)$ non-isomorphic forests of depth $d$ over~$[m]$.
Given a forest $F$ and a node $a\in V(F)$, we denote by $T_a$ the subtree of $F$ rooted at $a$.

 The following proposition holds. We omit the proof, as it is analogous to \cite[Lemma 30]{frick2004complexity}, which proves the special case where $m=2$.

\begin{proposition}\label{prop:encoding}
Let $d,m\in\mathbb{N}$ with $d,m\ge 1$.
There is a formula $\xi_{d,m}(x,y)$ of size $\Oh(d m)$ such that for every forest $F$ of depth $d$ over $[m]$, and every $a,b\in V(F)$, the subtrees $T_a$ and $T_b$ of $F$ rooted at $a$ and $b$ are isomorphic if and only if $F\models \xi_{d,m}(a,b)$.
\end{proposition}

Using Proposition~\ref{prop:encoding}, we prove the following:

\begin{lemma}\label{lem:lowerbound:constr}
Let $d,n,m\in\mathbb{N}$ with $d,m\ge 1$ and $\tower(d,m-1)  <  n \le \tower(d,m)$.
For every graph $G$ on $n$ vertices and every sentence $\varphi$,
there is a forest $F$ and a sentence $\psi$ such that the following hold:
\begin{itemize}
\item $F$ is a forest of depth $d+2$ over $[m]$,
\item $F$ has $\Oh(n^{d+2})$ vertices,
\item $\psi$ has size $\Oh(|\varphi|\cdot d m)$, and
\item $G\models \varphi$ if and only if $F\models \psi$.
\end{itemize}
Moreover, $F$ and $\psi$ can be constructed from $G$ and $\phi$ in time $\Oh(n^{d+2})$ and $\Oh(|\varphi|\cdot d m)$, respectively.
\end{lemma}

\begin{proof}
We map each vertex $v\in V(G)$ to a tree $T_v$ of depth $d+1$ over $[m]$.
Since  $n \le \tower(d,m)$, we can choose all these trees to be non-isomorphic.
For every $v\in V(G)$,
we consider the tree $H_v$ obtained by
creating a new root $r_v$ and making it the parent of the root of $T_v$.
Also, for every edge $e=uv\in E(G)$,
we consider the tree $H_e$ obtained by taking a copy $T_u^e$ of $T_u$ and a copy $T_v^e$ of $T_v$, creating a new root $r_e$, and making it the parent of the roots of $T_u^e$ and $T_v^e$.
Then, we set $F$ to be the disjoint union of all the trees~$H_x$, for $x\in V(G)\cup E(G)$. Thus, $F$ is a rooted forest over $[m]$.

We now analyze the size of each tree $H_v$.  By our assumption we have $\tower(d,m-1) <  n$, and we claim that this implies that no tree $T_v$ can have a vertex with branching $n$. Suppose towards a contradiction that there is a vertex $z$ of $T_v$ of branching at least $n$.
Let $i\ge 2$ be the depth of the subtree $T_z$ of $T_v$ rooted at~$z$. Since the branching of $z$ is at most the number of non-isomorphic trees of depth $i-1$ over $[m]$, we have that $n \le \tower(i-2,m) \le \tower(d-2,m)$.
But for large enough $n$ and $m\geq 2,$ we have that $\tower(d-2, m) < \tower(d, m-1)$.
This would imply that $n < n$, a contradiction. Consequently, each $T_v$ can have at most $n^d$ nodes.
Therefore, $F$ has size $\Oh(n^{d+2})$. Also, note that $F$ can be constructed from $G$ in the claimed running time.

To obtain $\psi$ from $\varphi$,
we first restrict
the quantification to the roots of the trees $H_v, v\in V(G)$ in $F$; the roots of the trees $H_v, v\in V(G)$ are distinguished from the roots of the trees $H_e, e\in E(G)$ by the number of their children (the root of each $H_v$ has exactly one child and the root of each $H_e$ has exactly two children).
We then use the formula $\xi_{d,m}(x,y)$ from Proposition~\ref{prop:encoding} to replace every atomic formula $E(x,y)$ in $\varphi$ by the formula that checks 
the existence of an edge $e=uv$ by asking the following:
the existence of a vertex $r_e$ (the root of the candidate $H_e$ encoding $e$)
that has two children 
$z_v$ and $z_u$ whose subtrees (in $H_e$) are pairwise isomorphic to $T_v$ and $T_u$.
It is easy to see that since $\xi_{d,m}$ has size $\Oh(dm)$,
the obtained formula $\psi$ has size $\Oh(|\varphi|\cdot d m)$, $G\models \varphi$ if and only if $F\models \psi$, and that $\psi$ can be constructed in the claimed running time.
\end{proof}

\begin{lemma}\label{lemma:lower-bound:reduction}
Let $d\in\mathbb{N}$.
Let $\mathcal{C}$ be a monotone class of tree rank at least $d+6$.
If $\mathcal{C}$ admits $d$-fold elementary model checking, then the class of all graphs admits $(d+1)$-fold elementary model checking.
\end{lemma}

\begin{proof}
Let $G$ be a graph on $n$ vertices and let $\varphi$ be a sentence.
We set $m\in\mathbb{N}$ such that $\tower(d+2, m-1) < n \le \tower(d+2,m)$.

We consider the forest $F$ of depth $d+4$ over $[m]$ and the sentence $\psi$ obtained from Lemma~\ref{lem:lowerbound:constr}.
Keep in mind that $F$ has 
$\Oh(n^{d+4})$ vertices,
$\psi$ is a sentence in the signature of rooted trees over $[m]$ and has size $\Oh(|\varphi|\cdot (d+2)m)$ and $G\models \varphi$ if and only if $F\models \psi$.
By applying Lemma~\ref{lem:uncoloring-forests},
we obtain a sentence $\psi'$ over the signature of graphs and an (uncolored and unrooted) forest $F'$ of depth at most $d+5$ with no nodes of degree $2$,
such that $|\psi'|=\Oh(m\cdot d\cdot |\psi|) = \Oh(|\varphi|\cdot d(d+2)m^2)$, $|F'|\leq \Oh(m\cdot |F|)$, and $F\models \psi$ if and only if $F'\models\psi'$. Note that $F'$ and $\psi'$ can be obtained from $G$ and $\phi$ in time polynomial in $|\phi|$ and $n$ for fixed $d$.

Observe that, since $\mathcal{C}$ has tree rank at least $d+6$,
there is some $r\in\mathbb{N}$ such that every
tree of depth $d+6$, and therefore every forest of depth at most $d+5$,
is contained in a graph in $\mathcal{C}$ as an $r$-shallow topological minor.
Hence, an ${\le}{r}$-subdivision $F''$ of $F'$  is contained as a subgraph in some graph in $\mathcal{C}$,
and the fact that $\mathcal{C}$ is monotone implies that $F''$ is itself in $\mathcal{C}$. Note that $|F''|\leq \Oh(r\cdot |F'|)\leq \Oh(rm\cdot |F|)=\Oh(rm\cdot n^{d+4})$.

We now define a sentence $\psi''$
over the signature of graphs such that $F'\models\psi'$ if and only if $F''\models\psi''$.
To get $\psi''$, we restrict the quantification of $\psi'$ only to principal vertices (these are the vertices of degree different than $2$) and by interpreting edges as paths of length at most $r$ between principal vertices whose internal vertices are all of degree $2$.
Observe that $|\psi''| = \Oh(r\cdot |\psi'|)$.
All in all, the following properties hold for $F''$ and $\phi''$:
\begin{itemize}
\item $F''$ is a forest (over the signature of graphs) in $\mathcal{C}$ and of size $\Oh(rm\cdot n^{d+4})= \Oh(n^{d+5})$,
\item $\varphi''$ has size $\Oh(|\varphi|\cdot r d(d+2)m^2)$,
and
\item $G\models \varphi$ if and only if $F''\models \varphi''$.
\end{itemize}
Therefore, to decide whether $G\models \varphi$, we equivalently decide whether $F''\models \varphi''$,
which by assumption can be done in time $\tower(d,\tilde{c}\cdot|\varphi''|)\cdot |V(F'')|^{\tilde{c}'}$, for some constants $\tilde{c}$ and $\tilde{c}'$, which is upper-bounded by $\tower(d,c\cdot |\varphi|\cdot r d(d+2)m^2)\cdot n^{c'd}$, for some constants $c$ and $c'$.

We now study the term $\tower(d,c\cdot |\varphi|\cdot r d(d+2)m^2)$.
It is easy to show that for every $h,x,y\in\mathbb{N}$, \begin{align*}
    \tower(h,x\cdot y) & \le
\tower(h,x^2 +y^2)\le\tower(h,x^2) \cdot  \tower(h,y^2)\\
& = \Oh(\tower(h+1,x) \cdot \tower(h+1,y)).
\end{align*}
Therefore, we have that
\[
\tower(d,c\cdot |\varphi|  \cdot r d(d+2)m^2)  = \Oh(\tower(d+1,c\cdot d(d+2)r\cdot |\varphi|) \cdot \tower(d+1,m^2)).\]
This, combined with the fact that $m$ is chosen with the property that $\tower(d+2,m-1)< n$ and that $\tower(d+1,m^2)\le \tower(d+2,m-1)$, implies that deciding whether $G\models \varphi$ can be done in time
$\tower(d+1,c\cdot d(d+2)r\cdot |\varphi|) \cdot n^{c''}$ for some constant~$c''$ (depending on $d$). This means that the class of all graphs admits $(d+1)$-fold elementary model checking.
\end{proof}

Theorem~\ref{thm:lower-bound} follows by combining Lemma~\ref{lemma:lower-bound:reduction} and the following result from~\cite[Theorem~3]{frick2004complexity}.

\begin{proposition}
Assume that ${\sf FPT}\neq {\sf AW}[*]$.
There is no $d\in \mathbb{N}$ such that the class of all graphs admits $d$-fold elementary model-checking.
\end{proposition}

\section{Rank vs tree rank}
\label{sec:dense}

In this section we compare the notion of tree rank of a graph class with a more general notion of \emph{rank}. The notion of rank is defined using transductions, which we present here in the setting of colored graphs. A {transduction} is an operation which inputs a graph, adds some unary predicates in an arbitrary way, and applies a fixed interpretation,
thus obtaining an output graph. This is formalized below. 

Let $\Sigma$, $\Gamma$ be two signatures of colored graphs and let $p$ be the number of unary predicates in $\Gamma$. A \emph{transduction} $P$ from $\Sigma$ to $\Gamma$ consists of unary predicate symbols $L_1,\ldots, L_m$ and an interpretation $I$ from $\Sigma \cup \{L_1,\ldots, L_m\}$ to $\Gamma$. We say that $H$ is an output of $P$ on $G$ if one can mark some vertices of $G$ by unary predicates $L_1,\ldots, L_m$ to obtain  a colored graph $G'$ such that $H = I(G')$. Note that in general there are multiple outputs of $P$ on any $G$, and so $P$ generates a \emph{set} of outputs on $G$ and not just a single graph. We denote the set of outputs of $P$ on $G$ by $P(G)$ and denote the fact that $H$ is an output of $P$ on $G$ by $H \in P(G)$.

Let $P$ be a transduction from $\Sigma$ to $\Gamma$ and let  $I = \{\psi_E,\vartheta,\theta_1,\ldots,\theta_p\}$ be the interpretation used in $P$.
Since in a transduction we can mark vertices in $G$ by unary predicates in an arbitrary way, we can assume that the formulas $\vartheta$ and $\theta_1,\ldots, \theta_p$ from $I$ are of the form $\vartheta(x):=L'(x)$ and $\theta_i(x):=L_i'(x)$ for some unary predicates $L',L_1',\ldots,L_p'$. Thus, at the cost of introducing $p+1$ extra unary predicate symbols, we can assume that all formulas in $P$ except $\psi_E$ are  quantifier-free. We then define the \emph{quantifier rank of a transduction} to be the quantifier rank of $\psi_E$. We also define the \emph{number of unary predicates used by $P$} to be  $m$ plus the number of unary predicates in~$\Sigma$.

Transductions can be applied to graph classes component-wise: $P(\C): = \bigcup_{G \in \C} P(G)$. We say that a graph class $\C$ \emph{transduces} a graph class $\D$ if there is a transduction such that $\D \subseteq P(\C)$.

\begin{definition}Fix $d\in\N$.
    The \emph{rank} of a class $\CC$ 
    is the largest number $d$ such that $\CC$ transduces the class 
    $\cal T_d$ of all trees of depth $d$.    
\end{definition}

\subsection{Rank \texorpdfstring{$\le$}{≤} tree rank + 1}

We first show that rank is a less restrictive concept than tree rank through the following result.

\begin{proposition}
\label{prop:sparse_rank_rank}
Every graph class of tree rank at most $d$ has rank at most $d + 1$.
\end{proposition}

We start with a lemma which tells us that if we can transduce a tree of depth $d+1$ and large branching in a graph $G$, then we can locally transduce a tree of depth $d$ and smaller but still large branching in $G$.

\begin{lemma}
\label{lem:transduction_radius}
Let $d \in \N$ with $d\ge 1$ and let $P$ be a transduction of quantifier rank $q$. 
Then there exists a transduction $P'$
and a function $g\from\N \to \N$ such that for any graph the following holds: If $T^{d+1}_{g(k)} \in P(G)$, then there is $w \in V(G)$ such that for any $r' \ge (6d + 1)7^q$ we have that $T^{d}_k \in P'(B^G_{r'}(w))$.

Moreover, the quantifier rank of $P'$ is $f(q)$, where $f$ is the function from Gaifman's theorem (Theorem~\ref{thm:Gaifman}).
\end{lemma}

In the proof below, for a formula $\rho(x,y)$ and a graph $G$, we define $\rho(G)$ to be the graph with the same vertex set as $G$ and edge set $E(\rho(G)):=\{uv~|~G \models \rho(u,v)\}$.

\begin{proof}
Let $P$ be a transduction of quantifier rank $q$ and let 
 $\psi(x,y)$ be the interpretation formula of quantifier rank $q$ used in $P$.
We apply Gaifman's theorem to $\psi(x,y)$ to obtain an equivalent formula $\hat{\psi}(x,y)$ of quantifier rank $f(q)$ which is a Boolean combination of basic local sentences (with locality radius $r$) and $r$-local formulas, where $r=7^q$.
Without loss of generality, we can assume that $\hat{\psi}(x,y)$ is of the form $\lor_{i=1}^\ell \big(\tau_i \land \rho_i(x,y)\big)$, where $\ell\in\mathbb{N}$, such that the following properties hold:
\begin{itemize}[nosep]
    \item each $\tau_i$ is a basic local formula (with locality radius $r$), and
    \item the formulas $\tau_i$ are mutually exclusive, i.e., for every graph $H$ there is exactly one index  $i\in\{1,\ldots,\ell\}$ such that $H\models\tau_i$, and 
    \item each $\rho_i(x,y)$ is an $r$-local formula.
\end{itemize}

Note that if $G$ is any graph (here we think of $G$ as being already marked with the unary predicates which $P$ uses), then after evaluating every basic local sentence $\tau_i$ on $G$, the fact that $\tau_i$ are mutually exclusive implies that we can determine for which $i\in\{1,\ldots,\ell\}$ the formula $\psi(x,y)$
reduces to $\rho_i(x,y)$ (in fact, it could happen that $\psi$ is evaluated to \textit{true} or \textit{false} by the aforementioned procedure, but then $\psi(G)$ is the complete graph or the edgeless graph on $|V(G)|$ vertices, respectively, and the lemma holds trivially).
The transduction $P'$ uses $\ell+1$ extra unary predicates which we denote by $C_1,\ldots,C_{\ell+1}$.
The predicate $C_{\ell+1}$ is used in order to (arbitrarily) mark the vertices of the interpreted graph in order to allow considering (arbitrary) induced subgraphs of the graphs in $P'(G)$.
The predicates $C_1,\ldots,C_\ell$ are used as `flags' in order to choose which formula $\rho_i(x,y)$ is used as the interpretation formula. This is because $P'$ may be evaluated only on an induced subgraph of $G$ and therefore cannot know which sentence $\tau_i$ evaluates to true in $G$ but can retrieve this information from $C_i$.

Our plan in the rest of the proof is as follows. After defining the function $g$ suitably, we consider a graph $G$ such that  $T^{d+1}_{g(k)} \in P(G)$. 
We then know that for some $i\in\{1,\ldots,\ell\}$, $T^{d+1}_{g(k)}$ is an induced subgraph of $\rho_i(G)$.
We will find a vertex $w$ of $G$ such that $T^d_k$ is an induced subgraph of 
$\rho_i(B^G_{r'}(w))$. 
Then $P'$ will recover $T^d_k$ from $B^G_{r'}(w)$ by using the extra unary predicate to mark the vertices of $T^d_k$ in $\rho_i(B^G_{r'}(w))$ and therefore getting $T^d_k$ as an output of $P'$.
In the rest of the proof we will denote by $\rho(x,y)$ the formula $\rho_i(x,y)$.

Before defining $g$ and proceeding with the proof, we will discuss some additional properties of $\rho$. 
We will use the following fact which follows from locality of $\rho$: 

\textit{($\ast$) Let $G$ be a graph, $S \subset V(G)$ and $u,v \in S$ vertices of $G$ such that $B^G_r(v) \subset S$ and $B^G_r(u) \subset S$. Then $G \models \rho(u,v) \Longleftrightarrow G[S] \models \rho(u,v)$.}

 It is well known that for any fixed value of quantifier rank there is a bounded number of non-equivalent formulas of this quantifier rank. 
Consequently, there is a fixed number of non-equivalent $r$-local formulas of quantifier rank at most $f(q)$, where  $f$ is the function from Gaifman's theorem (Theorem~\ref{thm:Gaifman}).
For any $G$ and $v \in V(G)$ we define the type of $v$ by
$$tp(v):=\{\alpha~|~\text{ $\alpha$ is $r$-local, $qr(\alpha) \le f(q)$ and $G \models \alpha(v)$}\}$$
Note that by the discussion above, the set ${\sf Types}$ of possible types has bounded size, say $z$. 
One can deduce from Gaifman's theorem that if $G$ is any graph and $u,v$ are two vertices of $G$ such that $\dist(u,v) > 2r$, then whether  $G \models \rho(u,v)$ holds depends only on $tp(u)$ and $tp(v)$ (see for example~\cite{bd_deg_interp} for details). We then define, for any two types $a,b \in {\sf Types}$, the $\rho$-type of $\{a,b\}$ to be $1$ if we have  $G \models \rho(u,v)$ for any $G$ and $u,v \in V(G)$ such that $\dist(u,v) > 2r$, $tp(u) = a$ and $tp(v) = b$; and $0$ otherwise. 

Finally, we set $g(k)=h(d+1,\tilde{k},z)$, where $h$ is the function of Lemma~\ref{lemma:ramsey-like} and $\tilde{k}\coloneqq\max\{k,3\}$ (the reason why we consider $\tilde{k}$ instead of $k$ is a technical detail of the proof which is clarified later).
For the rest of the proof fix a graph $G$ such that $\psi(G)$ contains $T^{d+1}_{g(k)}$ as an induced subgraph and let us denote this copy by $T$.   As discussed above, it is enough to consider the case when $T$ is an induced subgraph of $\rho(G)$.
Our goal will be to find a subtree $\hat{T}$ of $T$ of depth $d$ and branching $k$ such that $V(\hat{T})$ is contained $B_{6dr}^{G}(w)$ for some vertex $w \in V(G)$. 
Then, after taking any $r' \ge (6d+1)r$, all vertices of $\hat{T}$ and their $r$-neighborhoods in $G$  are in 
$B^G_{r'}(w)$. 
Consequently, by $(\ast)$ we will have for any $u,v \in V(\hat{T})$ that $G \models \rho(u,v) \Longleftrightarrow B^G_{r'}(w) \models \rho(u,v)$, and so $\hat{T}$ will be an induced subgraph of $\rho(B^G_{r'}(w))$, which will finish the proof.

As an intermediate step before finding $\hat{T}$, we first find a  subtree $T'$ of $T$ of branching $\tilde{k}$ and depth $d+1$ in which all vertices at the same level have the same type in $G$.
The existence of $T'$ follows directly from Lemma~\ref{lemma:ramsey-like}.

We claim that for any two vertices $u,v$ of $T'$ such that $u$ is a child of $v$ and $v$ is not the root of $T'$ we have that  $\dist_G(u,v) \le 6r$ in $G$.
We consider two cases. If $\dist_G(u,v) \le 2r$, then we are done. Otherwise, if $\dist_G(u,v) > 2r$ then we proceed as follows.
Since $uv \in E(T')$, the $\rho$-type of $\{tp(u),tp(v)\}$ has to be~$1$. Let $v_1,v_2$ be two siblings of $v$ (these exist since $T'$ has branching $\tilde{k}\ge 3$). Keep in mind that, since $v_1,v_2,$ and $v$ are at the same level of $T'$, it holds that $tp(v)=tp(v_1)=tp(v_2)$. In $T'$ we have that $uv_1 \not\in E(T')$ and  $uv_2 \not\in E(T')$, and so we cannot have $\dist_G(u,v_1) > 2r$ or $\dist_G(u,v_2) > 2r$, as this would contradict the $\rho$-type of $\{tp(u),tp(v)\}$. This means that we have $\dist_G(u,v_1) \le 2r$ and
$\dist_G(u,v_2) \le 2r$.
Let $u_1$ and $u_2$ be children of $v_1$ and $v_2$, respectively. With the same reasoning as above we can argue that $\dist_G(u_1,v) \le 2r$ and $\dist_G(u_1,v_2) \le 2r$. Then the path from $u$ going through $v_2$ and $u_1$ and finishing in $v$ certifies that the distance between $u$ and $v$ is at most $6r$, as desired.

Let $w$ be any child of the root of $T'$ and let $\hat{T}$ be the subtree of $T'$ rooted at $w$. Since the distance in $G$ between any two parent-child non-root nodes in $T'$ (and therefore in $\hat{T}$ as well) is at most $6r$, the whole vertex set of $\hat{T}$ is contained in $B^G_{6dr}(w)$. Then we have that the $r$-neighborhood of any $v \in V(\hat{T})$ is contained in $B^G_{(6d+1)r}(w)$, and consequently we have that $\hat{T}$ is an induced subgraph of $\rho(B^G_{r'}(w))$, as desired. Note that from $\hat{T}$, which has branching $\tilde{k}$, we can also get a subtree of branching $k$ (and depth $d+1$) as an induced subgraph of $\rho(B^G_{r'}(w))$.
\end{proof}

Proposition~\ref{prop:sparse_rank_rank} follows easily from the next lemma. In its proof we will use some tools from Section~\ref{sec:aux}, namely the notion of $(r,m)$-game rank (Definition~\ref{def:rm_rank}), operation $G \ast S$ (Definition~\ref{def:isolate}) and Lemma~\ref{lem:game_rank}.

\begin{lemma}
\label{lem:sparse_rank_implies_rank}
Let $d \in \N$ with $d\ge 1$.
Let $P$ be a transduction of quantifier rank $q$.
There exists $r=r(d,q)$
and, for every $m\in \N$,
$k=k(d,q,m)$
so that for any $G$ of  $(r,m)$-game rank $d$ we have that $T^{d+2}_k \not\in P(G)$.
\end{lemma}

\begin{proof}

We proceed by induction on $d$. Let $d=1$. Let $P$ be any transduction of quantifier rank $q$. Set $r\coloneqq (6d+1)\cdot 7^q$.
Let $m$ be arbitrary. 
We set $k\coloneqq g(m)$, where $g$ is the function from Lemma~\ref{lem:transduction_radius}, and claim that for no graph $G$ of  $(r,m)$-game rank one  we have that $T^3_k \in P(G)$.
Assume for contradiction that there is a graph $G$ of $(r,m)$-game rank one such that $T^3_k \in P(G)$. Then by Lemma~\ref{lem:transduction_radius} there is a transduction $P'$ and a vertex $w\in V(G)$ such that $T^2_m \in P'(B^G_r(w))$. But since $G$ is of $(r,m)$-game rank one, we know that $B^G_r(w)$ contains at most $m$ vertices, and since $T^2_m$ has $m+1$ vertices, this yields a contradiction.

For $d> 1$ assume that $P$ is any transduction of quantifier rank $q$. 
We apply Lemma~\ref{lem:transduction_radius} to $P$ to obtain transduction $P'$ of quantifier rank $f(q)$
and function $g$.
We apply the induction hypothesis to obtain value 
$r_0\coloneqq r_{P'}(d-1, f(q))$ where the subscript in $r_{P'}$ illustrates that the function $r(\cdot,\cdot)$ obtained by induction depends on the transduction $P'$.

Set $r\coloneqq \max\{(6d+1)\cdot 7^q,r_0\}$.
Let $m$ be arbitrary.
We set 
$k = k(d,q,m) \coloneqq g(k')$, where $k'\coloneqq k(d-1,f(q),m)$
is defined through the inductive hypothesis,
and claim that for every $G$ of  $(r,m)$-game rank at most  $d$ we have that $T^{d+2}_k \not \in P(G)$. 
Suppose for contradiction that there is such $G$ with $T^{d+2}_k \in P(G)$. 
We will show that there is a graph $H$ of  $(r,m)$-game rank $d-1$ and a transduction $P''$ of quantifier rank $f(q)$
such that  $T^{d+1}_{k'} \in P''(H)$.
Since $r \ge r_0$, we have that  $H$ is also of $(r_0,m)$-game rank at most $d-1$, and so this will yield a contradiction with the induction hypothesis.

We first construct $H$. By Lemma~\ref{lem:transduction_radius}, as $r \ge (6d+1)7^q$, there is a vertex $w$ of $G$ such that $T^{d+1}_{k'} \in P'(B^G_r(w))$.  For convenience we set $G' \coloneqq  B^G_r(w)$.  Since $G$ is of $(r,m)$-game rank at most $d$, there exist $m$ vertices $u_1,\ldots, u_m$ in $V(G')$ such that $G' \setminus \{u_1,\ldots, u_m\}$ is of $(r,m)$-game rank at most $d-1$.
By Lemma~\ref{lem:game_rank} the graph $G' \ast \{u_1,\ldots, u_m\}$ is also of $(r,m)$-game rank at most $d-1$. We then take as $H$ to be the graph $H:= G' \ast \{u_1,\ldots, u_m\}$.

We now proceed with describing $P''$. It is easily seen that $G'$ can be transduced from $H$ by a quantifier-free transduction $Q$ which uses $2m$ new predicates -- this transduction marks each $u_i$ with a unary predicate $L_i$ and each neighbor of $u_i$ in $G'$ with a unary predicate $N_i$ and then using a quantifier-free formula recovers the deleted edges. We then have $G' \in Q(H)$. 
Then, since $T^{d+1}_{k'} \in P'(G')$ and $G' \in Q(H)$, by composing $P'$ with $Q$ we obtain a transduction $P''$ such that $T^{d+1}_{k'} \in P''(H)$. Since $P'$ has quantifier rank $f(q)$
and $Q$ is quantifier-free,
the transduction $P''$ has quantifier rank $f(q)$,
as desired.
\end{proof}

\begin{proof}[Proof of Proposition~\ref{prop:sparse_rank_rank}]
Let $\C$ be a graph class of  tree rank at most $d$. Let $P$ be any transduction. Let $r$ be the integer obtained by applying Lemma~\ref{lem:sparse_rank_implies_rank} to $P$. Since $\C$ has tree rank at most $d$, there is $m$ such that every $G \in \C$ has $(r,m)$-game rank at most $d$. Then by Lemma~\ref{lem:sparse_rank_implies_rank} there exists $k$ such that we have that $T^{d+2}_k \not\in P(G)$ for every $G \in \C$, which means that $\C$ cannot have  rank $d+2$ and so it has rank at most $d+1$.
\end{proof}

\subsection{Tree rank \texorpdfstring{$\le$}{≤} rank, for weakly sparse classes}
\label{sec:wsparse}

In this section, we prove the following.

\begin{proposition}\label{prop:wsparse}
    For every weakly sparse class, its tree rank is not greater than its rank.
\end{proposition}

Proposition~\ref{prop:wsparse} follows by combining several known results, which we now recall.
The following result is due to Kierstead and Penrice~\cite{kierstead-penrice}.
\begin{theorem}[Kierstead and Penrice,~\cite{kierstead-penrice}]\label{thm:kp}
    Let $\CC$ be a weakly sparse class of unbounded minimum degree. Then $\CC$ contains all trees as induced subgraphs.
\end{theorem}

The following result is due to 
Dvo\v r\'ak \cite{DVORAK2018143}.

\begin{theorem}[Dvo\v r\'ak, \cite{DVORAK2018143}]\label{thm:dvorak}
    Let $\CC$ be a hereditary, weakly sparse class of unbounded expansion.
    Then there is a number $r\ge 0$ such that for every $\delta\in\N$,
    $\CC$ contains the $r$-subdivision of some graph of minimum degree larger than $\delta$.
\end{theorem}

\begin{corollary}\label{cor:dvorak}
    Let $\CC$ be a weakly sparse class of unbounded expansion.
    Then $\CC$ transduces a weakly sparse class of unbounded minimum degree.
\end{corollary}
\begin{proof}
    Since every class transduces its own hereditary closure, we can assume without loss of generality that $\CC$ is hereditary.
    Let $r\ge 0$ be as in Theorem~\ref{thm:dvorak}.
    If $r=0$ then there is nothing to prove, so suppose $r\ge 1$.
    
    Let $\DD$ be the class of graphs such that the $r$-subdivision of  every graph in $\DD$ belongs to $\CC$.
    Then $\DD$ has unbounded minimum degree, and is monotone; monotonicity follows from the fact that the removal of an edge of a graph in $\mathcal{D}$ corresponds to the removal of at most $r$ (subdivision) vertices of a graph in $\mathcal{C}$ and the fact that $\mathcal{C}$ is hereditary.
    Moreover, is easy to see that $\CC$ transduces $\DD$.
    If $\DD$ is weakly sparse, then we are done.
    Otherwise, $\DD$ contains all bipartite graphs, by monotonicity.
    There exists a class $\DD'$ of $K_{2,2}$-free bipartite graphs of unbounded minimum degree,
    so $\DD\supseteq \DD'$; such class $\DD'$ can be obtained e.g. by considering $2$-dimensional spaces over finite fields and taking the incidence graph of the points and lines of such space.
    In particular, $\CC$ transduces $\DD'$, and we are done.
\end{proof}

Combining Corollary~\ref{cor:dvorak} with Theorem~\ref{thm:kp} yields:
\begin{corollary}\label{cor:forests}
    Let $\CC$ be a weakly sparse class of unbounded expansion. Then $\CC$ transduces the class of all trees.
\end{corollary}

In particular, if $\CC$ has bounded rank and is weakly sparse, then $\CC$ has bounded expansion. 
Classes of bounded expansion in particular have bounded star-chromatic number~\cite[Theorem 7.1]{nevsetvril2008grad}. The definition will not matter here,
as we will only rely on the following, known lemma
(see e.g. \cite[Lemma 7.2]{gajarsky2022stable}).

\begin{lemma}\label{lem:star-chrom}
    Let $\CC$ be a class of bounded star-chromatic number,
    and let $\CC'$ be its monotone closure. Then $\CC$ transduces $\CC'$.
\end{lemma}

Proposition~\ref{prop:wsparse} easily follows 
from the above results.

\begin{proof}[Proof of Proposition~\ref{prop:wsparse}]
    Let $\CC$ be a weakly sparse class,
and suppose that  $\CC$  has tree rank at least $d$.
    We argue that $\CC$ transduces $\cal T_{d}$,
    thus proving that $\CC$ has rank at least $d$.

    If $\CC$ has unbounded expansion then $\CC$ transduces the class of all trees, 
    and therefore has unbounded rank, and we are done.

    So suppose $\CC$ has bounded expansion.
    Let $\CC'$ denote the monotone closure of $\CC$.
Then $\CC$  transduces $\CC'$ by Lemma~\ref{lem:star-chrom}.
By assumption, there is $r\ge 0$ such that $\CC'$ contains some ${\le}r$-subdivision of every tree in $\cal T_{d}$. It is easy to see that $\CC'$ transduces $\cal T_{d}$.
By transitivity, $\CC$ also transduces $\cal T_{d}$.
    In any case,  $\CC$ has rank at least $d$.
\end{proof}

To conclude, note that Proposition~\ref{prop:wsparse}, combined with Proposition~\ref{prop:sparse_rank_rank}, proves Theorem~\ref{thm:dense}.

\section{Open problems}

To conclude this article, let us restate the main open questions that we have already mentioned in the introduction. The first has to do with the conjecture of Pilipczuk and Toruńczyk  \cite[Conjecture XI.24]{flipwidth} (see Conjecture~\ref{conj:dense}) claiming that the hereditary graph classes that have elementarily-fpt model checking are precisely the classes of bounded rank. Towards this, showing that even classes that can be transduced from classes of bounded elementary tree rank are elementarily-fpt is an interesting step, and it is left for future work.
Also, answering Conjecture~\ref{conj:dense-collapse} would be interesting as well. Recall that this conjecture states that among hereditary graph classes, classes of bounded rank are precisely those classes for which the alternation hierarchy collapses. Finally, another possible direction is to relate our notion of tree rank to the upper and lower bounds of Dawar, Grohe, Kreutzer, and Schweikardt~\cite{model-theory-makes-formulas-large} on the sizes of various normal forms of first-order logic formulas.

\printbibliography

@inproceedings{DBLP:conf/lics/GajarskyP0ST24,
  author       = {Jakub Gajarsk{\'{y}} and
                  Michal Pilipczuk and
                  Marek Sokolowski and
                  Giannos Stamoulis and
                  Szymon Torunczyk},
  editor       = {Pawel Sobocinski and
                  Ugo Dal Lago and
                  Javier Esparza},
  title        = {Elementary first-order model checking for sparse graphs},
  booktitle    = {Proceedings of the 39th Annual {ACM/IEEE} Symposium on Logic in Computer
                  Science (LICS)},
  pages        = {36:1--36:14},
  publisher    = {{ACM}},
  year         = {2024},
  doi          = {10.1145/3661814.3662094},
  bibsource    = {dblp computer science bibliography, https://dblp.org}
}

@article{BODLAENDER19981,
	author = {Hans L. Bodlaender},
	doi = {https://doi.org/10.1016/S0304-3975(97)00228-4},
	issn = {0304-3975},
	journal = {Theoretical Computer Science},
	number = {1},
	pages = {1-45},
	title = {A partial k-arboretum of graphs with bounded treewidth},
	url = {https://www.sciencedirect.com/science/article/pii/S0304397597002284},
	volume = {209},
	year = {1998}}

@inproceedings{blcw,
author = {Bonnet, \'{E}douard and Dreier, Jan and Gajarsk\'{y}, Jakub and Kreutzer, Stephan and M\"{a}hlmann, Nikolas and Simon, Pierre and Toru\'{n}czyk, Szymon},
title = {Model Checking on Interpretations of Classes of Bounded Local Cliquewidth},
year = {2022},
isbn = {9781450393515},
doi = {https://doi.org/10.1145/3531130.3533367},
publisher = {ACM},
booktitle = {Proceedings of the 37th Annual ACM/IEEE Symposium on Logic in Computer Science, LICS 2022},
articleno = {54},
numpages = {13}
}

@article{tww4,
  author       = {{\'{E}}douard Bonnet and
                  Ugo Giocanti and
                  Patrice Ossona de Mendez and
                  Pierre Simon and
                  St{\'{e}}phan Thomass{\'{e}} and
                  Szymon Torunczyk},
  title        = {Twin-Width {IV:} Ordered Graphs and Matrices},
  journal      = {J. {ACM}},
  volume       = {71},
  number       = {3},
  pages        = {21},
  year         = {2024},
  doi          = {10.1145/3651151},
  bibsource    = {dblp computer science bibliography, https://dblp.org}
}

@article{tww1,
  author       = {{\'{E}}douard Bonnet and
                  Eun Jung Kim and
                  St{\'{e}}phan Thomass{\'{e}} and
                  R{\'{e}}mi Watrigant},
  title        = {Twin-width {I:} Tractable {FO} Model Checking},
  journal      = {J. {ACM}},
  volume       = {69},
  number       = {1},
  pages        = {3:1--3:46},
  year         = {2022},
  doi          = {10.1145/3486655},
  bibsource    = {dblp computer science bibliography, https://dblp.org}
}

@article{BraunfeldNMS25,
  author       = {Samuel Braunfeld and
                  Jaroslav Ne\v{s}et\v{r}il and
                  Patrice {Ossona de Mendez} and
                  Sebastian Siebertz},
  title        = {On first-order transductions of classes of graphs},
  journal      = {Log. Methods Comput. Sci.},
  volume       = {21},
  number       = {2},
  year         = {2025},
  doi          = {10.46298/LMCS-21(2:26)2025}
}

@article{chandra1982structure,
	title = {Structure and complexity of relational queries},
	journal = {Journal of Computer and System Sciences},
	volume = {25},
	number = {1},
	pages = {99-128},
	year = {1982},
	issn = {0022-0000},
    doi = {10.1016/0022-0000(82)90012-5},
	author = {Ashok Chandra and David Harel}
}

@article{Chvatal77,
  author       = {Va\v{s}ek Chv{\'{a}}tal},
  title        = {Tree-complete graph {Ramsey} numbers},
  journal      = {J. Graph Theory},
  volume       = {1},
  number       = {1},
  pages        = {93},
  year         = {1977},
  url          = {https://doi.org/10.1002/jgt.3190010118},
  doi          = {10.1002/JGT.3190010118},
  bibsource    = {dblp computer science bibliography, https://dblp.org}
}

@article{Courcelle90,
  author       = {Bruno Courcelle},
  title        = {The {M}onadic {S}econd-{O}rder {L}ogic of Graphs. {I}. {R}ecognizable Sets of
                  Finite Graphs},
  journal      = {Inf. Comput.},
  volume       = {85},
  number       = {1},
  pages        = {12--75},
  year         = {1990},
  url          = {https://doi.org/10.1016/0890-5401(90)90043-H},
  doi          = {10.1016/0890-5401(90)90043-H},
  bibsource    = {dblp computer science bibliography, https://dblp.org}
}

@inproceedings{model-theory-makes-formulas-large,
author = {Dawar, Anuj and Grohe, Martin and Kreutzer, Stephan and Schweikardt, Nicole},
title = {Model Theory Makes Formulas Large},
year = {2007},
isbn = {3540734198},
doi = {https://doi.org/10.1007/978-3-540-73420-8_78},
publisher = {Springer-Verlag},
booktitle = {Proceedings of the 34th International Conference on Automata, Languages and Programming, ICALP 2007},
pages = {913--924},
numpages = {12},
}

@book{Diestel,
	Author = {Reinhard Diestel},
	Publisher = {Springer},
	Series = {Graduate texts in mathematics},
	Title = {Graph Theory, 4th Edition},
	Volume = {173},
	doi = {https://doi.org/10.1007/978-3-662-70107-2},
	Year = {2012}}

@INPROCEEDINGS{DreierEMMPT24,
  author={Jan Dreier and
                  Ioannis Eleftheriadis and
                  Nikolas M{\"{a}}hlmann and
                  Rose McCarty and
                  Micha{\l{}} Pilipczuk and
                  Szymon Toru{\'n{}}czyk},
  booktitle={Proceedings of the 2024 IEEE 65th Annual Symposium on Foundations of Computer Science (FOCS)}, 
  title={First-Order Model Checking on Monadically Stable Graph Classes}, 
  year={2024},
  volume={},
  number={},
  pages={21-30},
  doi={10.1109/FOCS61266.2024.00012}}

@inproceedings{dms,
author = {Dreier, Jan and M\"{a}hlmann, Nikolas and Siebertz, Sebastian},
title = {First-Order Model Checking on Structurally Sparse Graph Classes},
year = {2023},
isbn = {9781450399135},
publisher = {ACM},
url = {https://doi.org/10.1145/3564246.3585186},
doi = {10.1145/3564246.3585186},
booktitle = {Proceedings of the 55th Annual ACM Symposium on Theory of Computing, STOC 2023},
pages = {567--580},
numpages = {14}
}

@article{DVORAK2018143,
	author = {Zden{\v e}k Dvo{\v r}{\'a}k},
	doi = {https://doi.org/10.1016/j.ejc.2017.10.004},
	issn = {0195-6698},
	journal = {European Journal of Combinatorics},
	pages = {143-148},
	title = {Induced subdivisions and bounded expansion},
	url = {https://www.sciencedirect.com/science/article/pii/S0195669817301634},
	volume = {69},
	year = {2018}}

@article{dvovrak2013testing,
	Author = {Dvo{\v{r}}{\'a}k, Zden{\v{e}}k and Kr{\'a}l, Daniel and Thomas, Robin},
	Journal = {J. ACM},
	Number = {5},
	Pages = {36},
	Publisher = {ACM},
	Title = {Testing first-order properties for subclasses of sparse graphs},
	doi ={https://doi.org/10.1145/2499483},
	Volume = {60},
	Year = {2013}}

@phdthesis{dvorak2007asymptotical,
	Author = {Dvo\v{r}{\'a}k, Zden\v{e}k},
	School = {Charles University, Faculty of Mathematics and Physics},
	Title = {Asymptotical structure of combinatorial objects},
	Year = {2007}}

@article{flum-grohe,
	author = {Flum, J\"{o}rg and Grohe, Martin},
	doi = {10.1137/S0097539799360768},
	eprint = {https://doi.org/10.1137/S0097539799360768},
	journal = {SIAM Journal on Computing},
	number = {1},
	pages = {113-145},
	title = {Fixed-Parameter Tractability, Definability, and Model-Checking},
	url = {https://doi.org/10.1137/S0097539799360768},
	volume = {31},
	year = {2001}}

@article{frick2004complexity,
  author = {Markus Frick and Martin Grohe},
  journal = {Annals of Pure and Applied Logic},
  doi = {10.1016/j.apal.2004.01.007},
  number = {1-3},
  pages = {3--31},
  title = {The complexity of first-order and monadic second-order logic revisited},
  volume = {130},
  year = {2004}}

@article{gaifman1982local,
	Author = {Gaifman, Haim},
	Journal = {Studies in Logic and the Foundations of Mathematics},
	Pages = {105--135},
	Publisher = {Elsevier},
	Title = {On local and non-local properties},
	Volume = {107},
	doi ={https://doi.org/10.1016/S0049-237X(08)71879-2},
	Year = {1982}}

@inproceedings{gajarsky2022differential,
  author       = {Jakub Gajarsk{\'{y}} and
                  Maximilian Gorsky and
                  Stephan Kreutzer},
  title        = {Differential Games, Locality, and Model Checking for {FO} Logic of
                  Graphs},
  booktitle    = {Proceedings of the 30th {EACSL} Annual Conference on Computer Science Logic, CSL 2022},
  series       = {LIPIcs},
  volume       = {216},
  pages        = {22:1--22:18},
  publisher    = {Schloss Dagstuhl - Leibniz-Zentrum f{\"{u}}r Informatik},
  year         = {2022},
  doi          = {10.4230/LIPIcs.CSL.2022.22}
 }

@article{gajarsky-hlineny,
  author       = {Jakub Gajarsk{\'{y}} and
                  Petr Hlinen{\'{y}}},
  title        = {Kernelizing {MSO} Properties of Trees of Fixed Height, and Some Consequences},
  journal      = {Log. Methods Comput. Sci.},
  volume       = {11},
  number       = {1},
  year         = {2015},
  url          = {https://doi.org/10.2168/LMCS-11(1:19)2015},
  doi          = {10.2168/LMCS-11(1:19)2015},
  bibsource    = {dblp computer science bibliography, https://dblp.org}
}

@article{bd_deg_interp,
  author       = {Jakub Gajarsk{\'{y}} and
                  Petr Hlin\v{e}n{\'{y}} and
                  Jan Obdr\v{z}{\'{a}}lek and
                  Daniel Lokshtanov and
                  M. S. Ramanujan},
  title        = {A New Perspective on {FO} Model Checking of Dense Graph Classes},
  journal      = {{ACM} Trans. Comput. Log.},
  volume       = {21},
  number       = {4},
  pages        = {28:1--28:23},
  year         = {2020},
  url          = {https://doi.org/10.1145/3383206},
  doi          = {10.1145/3383206},
  bibsource    = {dblp computer science bibliography, https://dblp.org}
}

@article{lsd-journal,
	author = {Jakub Gajarsk{\'{y}} and Stephan Kreutzer and Jaroslav Ne\v{s}et\v{r}il and Patrice {Ossona de Mendez} and Micha\l{} Pilipczuk and Sebastian Siebertz and Szymon Toru\'nczyk},
	bibsource = {dblp computer science bibliography, https://dblp.org},
	doi = {10.1145/3382093},
	journal = {{ACM} Trans. Comput. Log.},
	number = {4},
	pages = {29:1--29:41},
	title = {First-Order Interpretations of Bounded Expansion Classes},
	url = {https://doi.org/10.1145/3382093},
	volume = {21},
	year = {2020}}

@inproceedings{gajarsky2022stable,
  author       = {Jakub Gajarsk{\'{y}} and
                  Micha\l{} Pilipczuk and
                  Szymon Toru\'nczyk},
  title        = {Stable graphs of bounded twin-width},
  booktitle    = {37th Annual {ACM/IEEE} Symposium on Logic in Computer
                  Science, {LICS} 2022},
  pages        = {39:1--39:12},
  publisher    = {{ACM}},
  year         = {2022},
  url          = {https://doi.org/10.1145/3531130.3533356},
  doi          = {10.1145/3531130.3533356},
  bibsource    = {dblp computer science bibliography, https://dblp.org}
}

@article{shrubdepth-journal,
	author = {Robert Ganian and Petr Hlinen{\'{y}} and Jaroslav Nesetril and Jan Obdrz{\'{a}}lek and Patrice Ossona de Mendez},
	bibsource = {dblp computer science bibliography, https://dblp.org},
	doi = {10.23638/LMCS-15(1:7)2019},
	journal = {Log. Methods Comput. Sci.},
	number = {1},
	title = {Shrub-depth: Capturing Height of Dense Graphs},
	url = {https://doi.org/10.23638/LMCS-15(1:7)2019},
	volume = {15},
	year = {2019}}

@incollection{grokre11,
	Author = {Martin Grohe and Stephan Kreutzer},
	Booktitle = {Model Theoretic Methods in Finite Combinatorics},
	Editor = {M. Grohe and J.A. Makowsky},
	Pages = {181-206},
	Publisher = {American Mathematical Society},
	Series = {Contemporary Mathematics},
	Title = {Methods for Algorithmic Meta Theorems},
	Volume = {558},
	doi ={https://doi.org/10.1090/conm/558/11051},
	Year = {2011}}

@article{gks,
	author = {Martin Grohe and Stephan Kreutzer and Sebastian Siebertz},
	bibsource = {dblp computer science bibliography, https://dblp.org},
	doi = {10.1145/3051095},
	journal = {J. {ACM}},
	number = {3},
	pages = {17:1--17:32},
	title = {Deciding First-Order Properties of Nowhere Dense Graphs},
	url = {https://doi.org/10.1145/3051095},
	volume = {64},
	year = {2017}}

@article{DBLP:journals/lmcs/KazanaS19,
  author       = {Wojciech Kazana and
                  Luc Segoufin},
  title        = {First-order queries on classes of structures with bounded expansion},
  journal      = {Log. Methods Comput. Sci.},
  volume       = {16},
  number       = {1},
  year         = {2020},
  url          = {https://doi.org/10.23638/LMCS-16(1:25)2020},
  doi          = {10.23638/LMCS-16(1:25)2020},
  bibsource    = {dblp computer science bibliography, https://dblp.org}
}

@article{kierstead-penrice,
author = {Kierstead, H. A. and Penrice, S. G.},
title = {Radius Two Trees Specify $\chi$-Bounded Classes},
year = {1994},
doi ={https://doi.org/10.1002/jgt.3190180203},
issue_date = {March 1994},
publisher = {John Wiley \& Sons, Inc.},
address = {USA},
volume = {18},
number = {2},
issn = {0364-9024},
journal = {J. Graph Theory},
month = mar,
pages = {119-129},
numpages = {11}
}

@article{KIROUSIS1986205,
	author = {Lefteris M. Kirousis and Christos H. Papadimitriou},
	doi = {https://doi.org/10.1016/0304-3975(86)90146-5},
	issn = {0304-3975},
	journal = {Theoretical Computer Science},
	pages = {205-218},
	title = {Searching and pebbling},
	url = {https://www.sciencedirect.com/science/article/pii/0304397586901465},
	volume = {47},
	year = {1986}}

@article{lampis,
  author       = {Michael Lampis},
  title        = {First Order Logic on Pathwidth Revisited Again},
  journal      = {Log. Methods Comput. Sci.},
  volume       = {22},
  number       = {2},
  year         = {2026},
  url          = {https://doi.org/10.46298/lmcs-22(2:26)2026},
  doi          = {10.46298/LMCS-22(2:26)2026},
  bibsource    = {dblp computer science bibliography, https://dblp.org}
}

@article{nevsetvril2008grad,
	Author = {Ne{\v{s}}et{\v{r}}il, Jaroslav and Ossona de Mendez, Patrice},
	Journal = {European Journal of Combinatorics},
	Number = {3},
	Pages = {760--776},
	Publisher = {Elsevier},
	Title = {Grad and classes with bounded expansion {I}. {D}ecompositions},
	doi ={https://doi.org/10.1016/j.ejc.2006.07.013},
	Volume = {29},
	Year = {2008}}

@book{sparsity,
	Author = {Jaroslav Ne\v{s}et\v{r}il and Patrice {Ossona de Mendez}},
	Publisher = {Springer},
	Series = {Algorithms and combinatorics},
	Title = {Sparsity --- {G}raphs, {S}tructures, and {A}lgorithms},
	Volume = {28},
	doi = {https://doi.org/10.1007/978-3-642-27875-4},
	Year = {2012}}

@inproceedings{NesetrilRMS20,
  author       = {Jaroslav Ne\v{s}et\v{r}il and
                  Roman Rabinovich and
                  Patrice {Ossona de Mendez} and
                  Sebastian Siebertz},
  title        = {Linear rankwidth meets stability},
  booktitle    = {Proceedings of the 31st {ACM-SIAM} Symposium on Discrete Algorithms,
                  {SODA} 2020},
  pages        = {1180--1199},
  publisher    = {{SIAM}},
  year         = {2020},
  url          = {https://doi.org/10.1137/1.9781611975994.72},
  doi          = {10.1137/1.9781611975994.72},
  bibsource    = {dblp computer science bibliography, https://dblp.org}
}

@article{see96,
	Author = {Detlef Seese},
	Bibsource = {dblp computer science bibliography, http://dblp.org},
	Journal = {Mathematical Structures in Computer Science},
	Number = {6},
	Pages = {505--526},
	Title = {Linear Time Computable Problems and First-Order Descriptions},
	doi ={10.1017/s0960129500070079},
	Volume = {6},
	Year = {1996}}

@inproceedings{aggregate-queries,
author = {Toru\'{n}czyk, Szymon},
title = {Aggregate Queries on Sparse Databases},
year = {2020},
isbn = {9781450371087},
publisher = {ACM},
url = {https://doi.org/10.1145/3375395.3387660},
doi = {10.1145/3375395.3387660},
booktitle = {Proceedings of the 39th ACM SIGMOD-SIGACT-SIGAI Symposium on Principles of Database Systems, PODS 2020},
pages = {427–443},
numpages = {17}
}

@inproceedings{flipwidth,
      title={Flip-width: Cops and Robber on dense graphs}, 
      author={Szymon Toruńczyk},
  eprinttype ={extra-open},
  eprint={https://arxiv.org/abs/2302.00352},
      booktitle={Proceedings of the  64th IEEE Annual Symposium on Foundations of Computer Science, FOCS 2023},
  pages={663--700},
  publisher={IEEE},
  year={2023},
  doi={10.1109/FOCS57990.2023.00045}}

\newpage

\appendix

\section{Quantifier elimination in classes of bounded expansion}\label{app:qe}
In this section, in particular we prove Lemma~\ref{lem:dkt-constant}. We prove a more general statement, Theorem~\ref{thm:be-qe} below.

We first define the \emph{Gaifman graph} of structures over signatures $\Sigma$ consisting of a finite set of relation symbols, each with prescribed arity that is a non-negative integer, and unary function symbols.
Given such a signature $\Sigma$ and a $\Sigma$-structure $\mathbb{A}$, the \emph{Gaifman graph} of $\mathbb{A}$ is the graph whose vertex set is the universe of $\mathbb{A}$ and two elements $u,v\in \mathbb{A}$ are adjacent if and only if either $u$ and $v$ appear together in some tuple in some relation in $\mathbb{A}$, or $f(u)=v$ or $f(v)=u$ for some partial function $f$ in $\mathbb{A}$.

For a signature $\Sigma$ and a class $\CC$ of $\Sigma$-structures, we say that $\CC$
has \emph{bounded expansion} if the class of Gaifman graphs of structures in $\CC$ has bounded expansion.
By a \emph{flag} we mean a relation symbol $F$ of arity~$0$.
Therefore, it is interpreted in a given structure 
as a Boolean value, \emph{true} or \emph{false},
and we can use it in a formula, e.g. writing $F\land \alpha$ for some formula $\alpha$.
Next, whenever $\tup x$ is a~set of variables (of first-order logic) and $\str A$ is a~$\Sigma$-structure, we define $\str A^{\tup x}$ as the set of all valuations of $\tup x$ comprising arbitrary elements of $\str A$ (i.e., the set of functions from $\tup x$ to $\str A$), and for a~first-order formula $\varphi(\tup x)$, define $\varphi_{\str A} \coloneqq \{\tup a \in \str A^{\tup x}\,\colon\, \str A \models \varphi(\tup a)\}$.

The goal of this section is to show the following result, essentially due to \DKT~\cite{dvovrak2013testing}. Our addition here is the ``moreover'' part, which is obtained by analysing  the proof presented in \cite{lsd-journal}, carefully keeping track of the dependence of the time complexity on $\varphi$.

\begin{theorem}\label{thm:be-qe}
Let $\CCC$ be a class of structures with bounded expansion and let $\phi(\tup x)$ be a first-order formula.
There exists a signature $\wh\Sigma$ consisting of unary relation and function symbols and flags, a class $\wh\CCC$ of~~$\wh\Sigma$-structures with bounded expansion, a quantifier-free $\wh\Sigma$-formula $\widehat\phi(\tup x)$ 
and an algorithm which given $\str A\in\CCC$
outputs in time $\Oof_{\CCC,\phi}(|\str A|)$ a structure $\widehat{\str A}\in\widehat{\CCC}$ whose Gaifman graph is contained in the Gaifman graph of $\str A$, such that $\phi_{\str A}=\widehat\phi_{\widehat{\str A}}.$

Moreover, if, for some $s \in \N$, we have that $\tilde\nabla_r(\CC) \le \tower(s, \poly(r))$ for all $r \in \N$, then:
\begin{itemize}
	\item the formula $\wh\phi$ can be computed, given $\phi$, in time $\tower((s + 5)(q + 1),\poly(|\phi|))$,	 where $q$ is the alternation rank of $\phi$,
	\item  	
	the structure $\wh {\str A}$ can be computed, given $\phi$ and $\str A$, in time $\tower((s + 5)(q + 1),\poly(|\phi|))\cdot |\str A|$.
\end{itemize}
\end{theorem}

The key points of the statement are that (1) the formula $\widehat\phi(\tup x)$ is quantifier-free,  (2) the algorithm computes $\wh{\str A}$ from $\str A$ 
in linear time, and (3) $\phi_{\str A}=\widehat\phi_{\widehat{\str A}}.$
Thus, the theorem allows us to replace a formula $\phi(\tup x)$ with a quantifier-free formula $\wh\phi(\tup x)$, 
after processing the structure $\str A$ 
in linear time. 
For our purpose in this paper, the key observation is encompassed in the last bullet, concerning elementarity.

In the case when $\phi$ is a sentence (that is, $\tup x$ is empty),
Theorem~\ref{thm:be-qe} in particular allows us to solve the model-checking problem (\emph{does $\str A$ satisfy $\phi$?}) in time $\Oof_{\CCC,\phi}(|\str A|)$, for $\str A\in\CC$.
 Theorem~\ref{thm:be-qe} also solves a~more general problem: it allows computing, in time $\Oof_{\CC,\phi}(|\str A|)$,
a data structure that answers the following queries in 
time $\Oof_{\CC,\phi}(1)$: \emph{does a given tuple $\tup a\in \str A^{\tup x}$ satisfy $\phi(\tup a)$?}
We remark that even though the statement of Theorem~\ref{thm:be-qe} is non-uniform in $\phi$,
the proof is in fact uniform in $\phi$, 
as is required in the statements of Theorem~\ref{thm:dkt}
and  Lemma~\ref{lem:dkt-constant}.

This proves  Theorem~\ref{thm:dkt}.
Furthermore, Lemma~\ref{lem:dkt-constant} follows immediately from the ``moreover'' part of Theorem~\ref{thm:be-qe}. 

\medskip
To prove Theorem~\ref{thm:be-qe} it is enough 
to consider the case when $\phi(\tup x)$ 
is an~existential formula -- that is, of the form $\exists \bar y.\psi(\tup x,\tup y)$ where $\psi$ is quantifier-free -- as the general case proceeds by an repeated application of this case (see below for more details).
The proof of Theorem~\ref{thm:be-qe} proceeds by considering such existential formulas, and 
classes $\CC$ of increasing generality:
\begin{enumerate}
	\item First we prove the statement in the case where $\CC$ is a class of forests (possibly vertex-colored). 
	\item Next, we prove the statement in the case where $\CC$ is a class of relational structures of bounded treedepth (meaning that the class of underlying Gaifman graphs has bounded treedepth). This case is solved by reducing to the case of forests of bounded depth, by considering the elimination forests of the Gaifman graphs of the structures in $\CC$, expanded with colors used to encode the relations. 
	\item Finally, to prove the general case, we reduce it to the previous one, using low treedepth colorings.
\end{enumerate}

We start with considering the special case when $\CCC$ is a class of forests of bounded depth.

\subsection*{Quantifier elimination in forests of bounded depth}

\newcommand{\parent}{\mathit{parent}}
Fix a signature $\Sigma$ consisting of unary relation symbols, flags,
and one unary function $\parent$.
In this section, 
by a \emph{forest} of depth $d$ we mean a structure $F$ over the signature $\Sigma$,
	where $\parent$ is the parent function of some rooted forest  (see Fig.~\ref{fig:forest}). We assume by convention that if $r$ is a root, then $\parent(r)=r$.
	Whenever $A \subseteq F$, we define $\parent(A) = \{\parent(v)\,\colon\,v \in A\}$.
	For $i \in \N$ and $v \in F$, we set $\parent^i(v)$ to be the $i$-fold application of $\parent$ of $v$; we define $\parent^i(A)$ for $A \subseteq F$ accordingly.
	Since $F$ has depth at most $d$, we can assume that $0 \leq i < d$ in all considered formulas.

	\begin{figure}[ht!]\centering
		\includegraphics[page=3,scale=1.4]{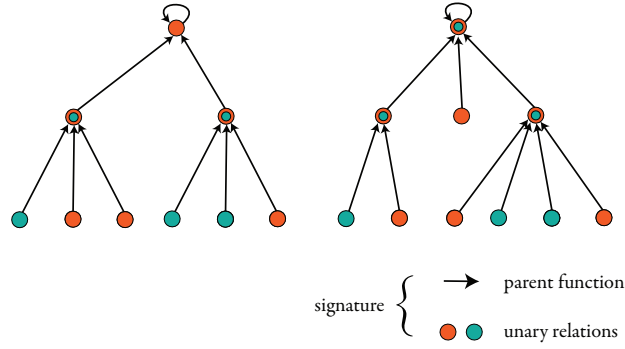}
		\caption{A forest of depth $3$ with two unary relations.}
		\label{fig:forest}
	   \end{figure}

	We prove a slight strengthening of Theorem~\ref{thm:be-qe} in the case when $\CC$ is a class of forests of depth at most $d$. In this strengthening, the constructed graph $\wh F$ is a forest 
	whose vertices and parent function are exactly the same as in $F$,
	and $\wh F$ may have some additional unary predicates and flags.
	In this case, we call $\wh F$ a \emph{coloring} of $F$.

\begin{lemma}\label{lem:forest-qe}	
	For every existential
 $\Sigma$-formula $\phi(\tup x)=\exists \tup y. \psi(\tup x,\tup y)$, where $\psi$ is quantifier-free,
	and every colored forest $F$ (over the signature $\Sigma$),
	there exist
	\begin{itemize}
		\item a 
		 signature $\wh\Sigma$ extending $\Sigma$ by unary predicates and flags, 
		\item a quantifier-free $\wh\Sigma$-formula $\wh\phi(\tup x)$, 
		\item a $\wh\Sigma$-coloring $\wh F$ of $F$, 
	\end{itemize}	
	such that:
	\begin{enumerate}
		\item  $\phi_F=\wh\phi_{\wh F},$
		\item $\wh\Sigma$ and $\wh\phi(\tup x)$ are computable from $\phi(\tup x)$ in time bounded by $\tower(4, \poly(|\phi|))$,
		\item $\wh F$ is computable from $\phi(\tup x)$ and $F$ in time bounded by
		$\tower(4, \poly(|\phi|)) \cdot |F|$.
	\end{enumerate}
\end{lemma}
The proof of  Lemma~\ref{lem:forest-qe} is given below. It follows the proof from~\cite{lsd-journal},
and uses automata on trees, which are recalled below.

\newcommand{\singleton}{\emph{singleton}}
\paragraph{Monadic Second Order logic}
In the course of the proof of Lemma~\ref{lem:forest-qe} it will be much more convenient to work with a~more powerful Monadic Second Order logic (MSO), allowing for quantification over sets of vertices: that is, each variable $X$ in an~MSO formula may be interpreted as a~set of vertices of a~forest $F$.
More precisely, a~\emph{valuation} of a~collection $\tup X$ of set variables is a~function $\tup A \colon V(T) \to 2^{\tup X}$.

For convenience, we will actually assume that all variables in MSO formulas are set variables: a~vertex $v$ will be modeled by a~singleton $\{v\}$.
Given a~set variable $X$, we define the predicate $\singleton(X)$ verifying that $X$ contains exactly one element.
Then, MSO formulas on forests can use the function symbol $\parent$ mapping sets to sets, the predicate $\singleton$, test for inclusion (and therefore equality) of sets and quantify over sets of vertices.

It can be easily verified that a~formula (resp.\ quantifier-free formula, existential formula) $\varphi$ of first-order logic can be rewritten in linear time into an~equivalent formula (resp.\ quantifier-free formula, existential formula) $\varphi'$ of monadic second-order logic with the same number of variables, and that $|\varphi'| \leq \Oh(|\varphi|)$.

	\newcommand{\multiset}[1]{\{\!\!\{#1\}\!\!\}}%
	\newcommand{\aut}{\cal}%
\paragraph{Tree valuation automata}

For a set $Q$, let $M(Q)$ denote 
the set of 
multisets of elements of $Q$, which are formally represented by functions $A\from Q\to \N$ counting the number of occurrences of each element of $Q$ in the multiset.
For $q_1,\ldots,q_n\in Q$,
let $\multiset{q_1,\ldots,q_n}$ denote 
the element of $M(Q)$ represented by the function $A\from Q\to\set{0,1,\ldots,n}$
such that $A(q)=|\setof{i\in\set{1,\ldots,n}}{q_i=q}|$.

Fix $k\in\N$. We denote by $M_k(Q)$ the set of functions $A\from Q\to \set{0,\ldots,k}$.
Also, we define the mapping  $\textit{trim}_k\from\N\to \set{0,1,\ldots,k}$
such that $\textit{trim}_k(n)=\min(n,k)$.
This naturally induces a mapping 
$\textit{trim}_k\from M(Q)\to M_k(Q)$,
such that $\textit{trim}_k(A)(q)=\textit{trim}_k(A(q))$.
We write $\multiset{q_1,\ldots,q_n}_k$
for $\textit{trim}_k(\multiset{q_1,\ldots,q_n})$.

\medskip
We now define an automata model that inputs pairs $(T,\tup A)$ consisting of a tree $T$ and valuation $\tup A$ of a~collection of set variables $\tup X$.

A \emph{tree valuation automaton} $\cal A(\tup X)$ consists of:
\begin{itemize}
	\item a finite set of \emph{states} $Q$,
	\item a set $Q_F\subset Q$ of \emph{accepting states},
	\item a set $\delta\subset M(Q)\times 2^{\tup X}\times Q$ of \emph{transitions}.
\end{itemize}
We require that the relation $\delta$ 
is finitely presentable, in the following sense.
There is a number $k\in\N$
such that $(A,\tup X,q)\in\delta$ 
if and only if $(\textit{trim}_k(A),\tup X,q)\in\delta$. The minimum such $k$ is called the \emph{threshold} of $\aut A$.
In this case, $\delta$ is naturally represented by the set of triples 
$(A,\bar X,q)\in \delta$ such that $A\in M_k(Q)$.
Also, the \emph{size} of $\aut A$ is defined as $|\aut A|\coloneqq (k+1)^{|Q|} \cdot 2^{|\tup X|} \cdot |Q|$.

Fix a tree $T$
with a valuation $\tup A\from V(T)\to 2^\tup X$. 
A \emph{run} of $\cal A$ on $(T,\tup A)$ 
is a mapping $\rho\from V(T)\to Q$ such that for every node $v$ 
and 
	children $v_1,\ldots,v_n$, we have that $$(\multiset{\rho(v_1),\ldots,\rho(v_n)},\tup A[v],\rho(v))\quad\in\quad\delta.$$
	The run is \emph{accepting} if 
		it labels the root of $T$ by an accepting state.
The automaton $\cal A(\tup X)$ \emph{accepts}
 $(T,\tup A)$ 
if there is some accepting run of $\cal A$ on $(T,\tup A)$.

\paragraph{Deterministic automata}
Let $\mathcal{A}$ be an automaton as above. We say that $\cal A$ is \emph{deterministic}
if $\delta$ is a function 
$\delta\from M(Q)\times 2^{\tup X}\to Q$.
Such an automaton has a unique run on every pair $(T,\tup A)$, and this run can be computed in time $\Oof(|\cal A||T|)$, by processing the tree $T$ from the leaves towards the root.
Any automaton can be converted to an equivalent deterministic automaton, using the classic powerset construction recalled below.

Given an automaton $\cal A$ as above,
 define the deterministic \emph{powerset automaton} $P(\cal A)$,
with the following components:
\begin{itemize}
	\item set of states $2^Q$,
	\item set of accepting states consisting of sets $\subset Q$ that contain an accepting state of $\cal A$,
	\item transition function $\hat\delta\from 
	M(2^Q)\times 2^{\tup X}\to 2^Q$ defined by:
	\[\begin{split}
	\hat\delta(\multiset{Q_1&,\ldots,Q_n},\tup A)\\&=\setof{q\in Q}{\exists q_1\in Q_1\ldots\exists q_n\in Q_n: (\multiset{q_1,\ldots,q_n},\tup A,q)\in \delta}.\end{split}\]
\end{itemize}
It is easy to see that $\cal A$ and $P(\cal A)$ accept the same pairs $(T,\tup A)$.
Moreover, $P(\aut A)$ has threshold bounded from above by $|Q|\cdot t_{\cal A}$, where $t_{\cal A}$ is the threshold of $\cal A$, and can be computed from $\aut A$ in exponential time.

\paragraph{Product construction}
By a~standard construction, given two automata $\aut A_1$ and $\aut A_2$ over the same set of variables $\tup X$ with the sets of states $Q_1, Q_2$, respectively,  we can produce a~product automaton $\aut A_1 \cup \aut A_2$ (resp.\ $\aut A_1 \cap \aut A_2$) over $\tup X$ with the set of states $Q_1 \times Q_2$ such that, for any input pair $(T, \tup A)$, the product automaton accepts $(T, \tup A)$ if and only if either $\aut A_1$ or $\aut A_2$ accept (resp.\ both $\aut A_1$ and $\aut A_2$ accept).
The product automaton can be constructed in time linear in the size of the resulting automaton.
It can also be easily seen that if $\aut A_1$ and $\aut A_2$ have threshold at most $k$, then so do $\aut A_1 \cup \aut A_2$ and $\aut A_1 \cap \aut A_2$.

\paragraph{From MSO formulas to automata}
The following result is folklore. We sketch its proof below.
\begin{lemma}\label{lem:MSO-to-aut}
	For every existential formula $\phi(\bar X)=\exists \bar Y.\psi(\bar X,\bar Y)$  there is a tree valuation automaton $\aut A(\bar X)$ 
	such that for every tree $T$ and every valuation $\bar A$,
		 $$T\models \phi(\bar A)\text{ \quad if and only if \quad}\aut A(\bar X)\text{ accepts }(T,\bar A).$$ 
		 $\aut A(\bar X)$ has threshold at most $2$, at most $\tower(2,\poly(|\varphi|))$ states and can be constructed 
		 from $\varphi$ in time $\tower(3,\poly(|\varphi|))$.
\end{lemma}

\begin{proof}[Proof sketch]
First consider the case when $\varphi$ is quantifier-free: $\bar Y=\emptyset$
and $\phi(\bar X)=\psi(\bar X)$.
Moreover, we can assume that $\psi$ is in negation normal form. Then $\psi$ is a positive Boolean combination of \emph{literals}; a literal is an atomic formula or its negation.
When $\psi$ is a literal, an automaton of threshold at most $2$ and at most $2^{\mathcal{O}(|\varphi|)}$ states can be constructed by hand. Such an automaton handles atomic formulas of the form $\parent^i(X)\subseteq \parent^j(Y)$ by keeping track for each node the set of levels in the tree that there is a descendant in $X$ (resp. $Y$) in distance $i$ (resp. $j$). Since $i$ and $j$ are upper-bounded by $|\varphi|$, we get the aforementioned bound. An automaton constructed this case is deterministic and therefore, by negating, we get the same number of states for negations of atomic formulas.
For conjunctions and disjunctions, the product automaton construction is employed.
The resulting automaton has threshold at most $2$ and at most $\tower(2,\poly(|\varphi|))$ states, so its size is at most $\tower(3,\poly(|\varphi|))$.
This handles the case when $\psi$ is quantifier-free.

In the general case, first construct the automaton $\aut B(\tup X\cup\tup Y)$ for the quantifier-free formula $\psi(\tup X,\tup Y)$ as just described.
The required automaton $\aut A(\tup X)$ for $\phi(\tup X)$ has the same states and accepting states as $\aut B$, and has transitions 
$(A,B\cap \tup X,q)$ for every transition $(A,B,q)$ of $\aut B$.
\end{proof}

\begin{proof}[Proof of Lemma~\ref{lem:forest-qe}]
We now move on to the proof of Lemma~\ref{lem:forest-qe}.
We prove a version for trees, rather than forests, as this is more consistent with the model of automata introduced above. Thus, we assume that the input forests are trees. The general case reduces easily to this case, by converting a forest to a tree by introducing a dummy root (the opposite conversion, from trees back to forests is also applied in the end, and this requires using flags to mark the unary predicates of the removed root).
Furthermore, 
for simplicity, we assume that the signature $\Sigma$
consists only of the parent function, and does not include unary predicates. The unary predicates do not cause any essential trouble, but only complicate the notation slightly.

We now describe the construction of the formula $\wh\phi(\tup x)$ and tree $\wh T$,
given $\phi(\tup x)=\exists \bar y.\psi(\tup x,\tup y)$ and tree~$T$.

We may convert $\phi(\tup x)$ to a monadic second-order formula $\phi'(\tup X)$, as described above.
So for a tree~$T$, a tuple $\tup a\in V(T)^{\tup x}$
may be seen as a valuation of $\tup X$, which assigns to each set variable $X\in\tup X$ -- corresponding to the variable $x$ of $\phi$ -- the set $\set{\tup a(x)}\subset V(T)$.
Then, 
\begin{align}\label{eq:phi'}
	T\models\phi(\tup a)\quad\iff\quad T\models \phi'(\tup a).	
\end{align}

Compute the automaton $\aut A'(\tup X)$ from $\phi'(\tup X)$, as given by Lemma~\ref{lem:MSO-to-aut}.
Next, we convert 
$\aut A'$ to an equivalent deterministic automaton $\aut A$, using the powerset construction.
This incurs an exponential blowup in the size of the automaton; so the number of states of this automaton is upper-bounded by $\tower(3, \poly(|\varphi|))$, and its size is upper-bounded by $\tower(4, \poly(|\varphi|))$.
Let $Q$ be the set of states of $\aut A$.

\newcommand{\cnt}{\emph{count}}
Given the tree $T$,
compute the (unique) run $\rho_\emptyset$ of $\aut A(\tup X)$ on $(T,\tup \emptyset)$, where $\bar\emptyset$ is the \emph{empty valuation} of~$\bar X$, namely $\bar\emptyset(X)=\emptyset$ for all $X\in \bar X$.
For a~node $v \in T$ let $\cnt_v\from Q \to \N$ be the function sending a~state $q \in Q$ to the number of children of $v$ that are labeled $q$ according to $\rho_\emptyset$, and define $\cnt^\star_v : Q \to \{0, 1, \dots, |\tup x| + 2\}$ by letting $\cnt^\star_v(q) = \mathit{trim}_{|\tup x| + 2}(\cnt_v(q))$.
Finally, we define $\wh T$ as the tree $T$, where each node $v \in T$ is assigned the label $(q, \cnt^\star_v(q))$.
Observe that the number of possible labels is bounded by $|Q| \cdot (|\tup x| + 3) \leq \tower(3, \poly(|\varphi|))$.
We now show that:

\begin{claim*}
	For any rooted tree $T$ and a~tuple $\tup a \in V(T)^{\tup x}$, the satisfaction of $T \models \varphi(\tup a)$ depends only on the atomic type of $\tup a$ in $\wh T$ -- i.e., the set of all inclusions of the form $\parent^i(\{u\})\subseteq \parent^j(\{v\})$ for $0 \leq i, j < d$ and $u, v \in \tup a$ that hold in $\wh T$, and the labels of vertices $\parent^i(\{u\})$ in $\wh T$ for $0 \leq i < d$ and $u \in \tup a$.
	Moreover, given an~atomic type of $\tup a$ in $\wh T$, we can verify if $T \models \varphi(\tup a)$ in time $\tower(4, \poly(|\varphi|))\cdot|T|$.
\end{claim*}
\begin{subproof}[Proof sketch of the claim]
	Given a tuple $\tup a\in V(T)^{\tup x}$, we have that $T\models\phi(\tup a)$ if and only if $\aut A$ accepts $(T,\tup a)$.
	Since $\aut A$ is deterministic, it has exactly one accepting run on $(T,\tup a)$; denote this run by $\rho_{\tup a}$.
	The run $\rho_{\tup a}$ differs from the run $\rho_{\emptyset}$ only on the vertices $v$ that are ancestors of the vertices in $\tup a$.

	We will now sketch how to reconstruct $\rho_{\tup a}(v)$ for these vertices $v$.
	We iterate all ancestors $v$ of the vertices in $\tup a$ in the bottom-up order.
	Let $C_v$ be the set of children of $v$ in $T$, and let $C_{v, \tup a}$ be the restriction of $C_v$ to the ancestors of the vertices in $\tup a$; note that $|C_{v, \tup a}| \leq |\tup a| = |\tup x|$.
	Since $\aut A$ has threshold $2$ and is deterministic, the state $\rho_{\tup a}(v)$ can be uniquely deduced from $\mathit{trim}_2(\multiset{\rho_{\tup a}(c) :\, c \in C_v})$ and, for every entry $a \in \tup a$, the information whether $v = a$.
	Letting $A\,\colon\,Q \to \N$ be the function representing $\multiset{\rho_{\tup a}(c) :\, c \in C_v}$, we have that
	\[ A(q) = \cnt_v(q) - |\{c \in C_{v, \tup a} : \rho_\emptyset(c) = q\}| + |\{c \in C_{v, \tup a} : \rho_{\tup a}(c) = q\}|. \]
	The values $c \in C_{v, \tup a}$ for which $\rho_\emptyset(c) = q$ can be determined from the atomic type of $\tup a$, and the states $\rho_{\tup a}(c)$ for $c \in C_{v, \tup a}$ have been already computed; hence we can uniquely determine $\Delta := - |\{c \in C_{v, \tup a} : \rho_\emptyset(c) = q\}| + |\{c \in C_{v, \tup a} : \rho_{\tup a}(c) = q\}|$.
	Since $\Delta \geq -|\tup x|$, we can then compute $\mathit{trim}_2(A(q))$ from $\cnt^\star_v(q)$ and $\Delta$.
	Given that, we can evaluate $\mathit{trim}_2(\multiset{\rho_{\tup a}(c) :\, c \in C_v})$.
	All in all, we can determine $\rho_{\tup a}(v)$, looking only at the atomic type of $\tup a$ in $\wh T$ and the previously computed values $\rho_{\tup a}(c)$ for those children $c$ of $v$ that are ancestors of the vertices in $\tup a$.
	
	Thus the satisfaction of $T \models \varphi(\tup a)$ only depends on the atomic type of $\tup a$ in $\wh T$ (and not on the tree $\wh T$ itself).
	It is also easy to verify that this procedure can be implemented in time $\tower(4, \poly(|\varphi|))\cdot |T|$.
\end{subproof}

By the claim above, we determine the set $\Lambda$ of atomic types of $\tup x$-tuples in $\wh T$, for which $T \models \varphi(\tup a)$ if and only if the atomic type of $\tup a$ belongs to $\Lambda$.
Note that $|\Lambda| \leq \tower(4, \poly(|\varphi|))$ since there are at most $\tower(4, \poly(|\varphi|))$ distinct atomic types of $\tup x$-tuples in $\wh T$ in total.
Also, the set $\Lambda$ can be found in time $\tower(4, \poly(|\varphi|))\cdot|T|$.

Given that, we rewrite $\varphi$ to a~quantifier-free formula $\wh \varphi(\tup x)$ as follows:
\[ \wh \varphi(\tup x) = \bigvee_{\tau \in \Lambda} [\tup x\text{ has atomic type }\tau\text{ in }\wh T]. \]
Such $\tup x$ (together with the tree $\wh T$) satisfies all the requirements of the lemma.

This completes the sketch of the proof of  Lemma~\ref{lem:forest-qe}.
See \cite{lsd-journal} for more details.
\end{proof}

\subsection*{Bounded treedepth}
We now prove Theorem~\ref{thm:be-qe} in the case of existential formulas when $\CCC$ is a class 
of structures of bounded treedepth (we refer to \cite{sparsity} for the definition and properties of treedepth). We repeat the statement in this case.

\begin{lemma}\label{lem:td-qe}
	Let $\phi(\tup x)$  be an existential first-order formula of the form $\exists \tup y.\psi(\tup x,\tup y)$, and let $d\in\N$. There exists a signature $\wh\Sigma$ consisting of unary relation and function symbols and flags, and  a quantifier-free $\wh\Sigma$-formula $\widehat\phi(\tup x)$ 
	such that for every $\Sigma$-structure  $\str A$ of treedepth $d$
	there is a structure $\widehat{\str A}\in\widehat{\CCC}$ whose Gaifman graph is contained in the Gaifman graph of $\str A$, such that $\phi_{\str A}=\widehat\phi_{\widehat{\str A}}.$
	Moreover:
	\begin{itemize}
		\item the formula $\wh\phi$ can be computed, given $\phi$ and $d$, in time bounded by $\tower(5, \poly(|\varphi|, d))$,
		\item the structure $\wh {\str A}$ can be computed, given $\phi, d$ and $\str A$, in time $\tower(5, \poly(|\varphi|, d)) \cdot |\str A|$.
	\end{itemize}

\end{lemma}
\begin{proof}
	We show how to reduce this case to the case of forests of bounded depth.
	Without loss of generality, we assume that the signature $\Sigma$ of $\CCC$ is relational by converting the functions to relations in the standard way (this possibly increases the number of existentially quantified variables $\tup y$, but this increase is at most linear in the size of $\varphi$).
	
	Suppose the Gaifman graph of each $\str A\in\CCC$ has treedepth at most $d$.
	Fix $\str A\in\CCC$ and any depth-first search tree $F$ of the Gaifman graph (say, $G$) of $\str A$.
	Then $F$ is a~subgraph of $G$ that is also a~(rooted) elimination forest $F$ of $G$ (since $F$ is a~DFS tree).
	Also since $G$ has treedepth at most $d$, it does not contain the path with $2^d$ vertices as a~subgraph; so neither does $F$, implying that $F$ has depth less than $2^d$.
	
    Let $\parent\from\str A\to\str A$ denote the parent function of $F$; note that whenever $\parent(v) \neq v$, we have $\{v, \parent(v)\} \in E(G)$.
    For every relation $R\in\Sigma$ of arity $k$ and every tuple
    $\tup h=(h_1,\ldots,h_k)\in\set{0,\ldots,2^d - 1}^k$ containing at least one $0$ define a unary predicate $R_{\tup h}$ on $\str A$ so that for $a\in\str A$,
    \[R_{\tup h}(a)\quad\Longleftrightarrow\quad R(\parent^{h_1}(a),\ldots,\parent^{h_k}(a)).\]
    Then the following equivalence holds  in $\str A$:
	\begin{align}\label{eq:forest-formula}
    	R(x_1,\ldots,x_k)\quad\Longleftrightarrow\quad\bigvee_{i=1}^k\bigvee_{\tup h}\left(R_{\tup h}(x_i)\land \bigwedge_{j=1}^k (\parent^{h_j}(x_i)=x_j)\right).    
	\end{align} 
	The implication $\Leftarrow$ holds by definition of $R_{\tup h}$. For the  implication $\Rightarrow$, if $R(a_1,\ldots,a_k)$ holds, then $a_1,\ldots,a_k$ form a clique in the Gaifman graph of $\str A$, and hence lie on some root-to-leaf path in $F$ as $F$ is an elimination forest. Then there is some $i$ such that $a_i$ is the descendant of all $a_1,\ldots,a_k$, and  a tuple $\tup h\in\set{0,\ldots,2^d-1}^k$ such that $a_j=\parent^{h_j}(a_i)$ for $j=1,\ldots,k$ and $h_i=0$. Then $R_{\tup h}(a_i)$ holds.

	Let $\Sigma'$ be the signature consisting of the unary function symbol $\parent$  and the unary predicates $R_{\tup h}$, for each $R\in\Sigma$ of arity $k$ and $\tup h\in\set{0,\ldots,2^d - 1}^k$; so $\Sigma'$ has at most $2^{\poly(|\varphi|, d)}$ unary predicates.
	For $\str A\in\CCC$ let 
 	$\str A'$ be the $\Sigma'$-structure with the same domain as $\str A$ in which the function $\parent$ and relations $R_{\tup h}$ are interpreted as defined above.
	Let $\phi'$ be obtained from $\phi$ by  replacing each atom $R(x_1,\ldots,x_k)$ in $\phi(\tup x)$ according to the right-hand side of~\eqref{eq:forest-formula}. 
	Then $|\phi'| \leq 2^{\poly(|\varphi|, d)}$ and we have that $\phi_{\str A}=\phi'_F$ for every $\str A\in\CCC$ and $F\in \CCC'$ obtained as described above, and moreover, $F$ can be computed from $\str A$ in time $2^{\poly(|\varphi'|, d)} \cdot |F|$. Applying Lemma~\ref{lem:forest-qe} to $\phi'$ and $\str A'$ yields the conclusion.
\end{proof}

\subsection*{Bounded expansion}

\newcommand{\A}{\str A}
\newcommand{\whA}{\wh{\A}}
We now proceed to proving  Theorem~\ref{thm:be-qe} in the most general case of bounded expansion.
We again begin with presenting the quantifier elimination for existential formulas.

\begin{lemma}\label{lem:be-existential}
	Let $\CCC$ be a class of structures with bounded expansion and let $\phi(\tup x) = \exists{\tup y}. \psi(\tup x, \tup y)$ be an existential first-order formula.
	Let $s \in \N$ be such that $\tilde\nabla_r(\str A) \le \tower(s, \poly(r))$ for $r \in \N$ and $\str A \in \CCC$.

	There exists a signature $\wh\Sigma$ consisting of unary relation and function symbols and flags; a class $\wh\CCC$ of~~$\wh\Sigma$-structures with bounded expansion with the property that for each $\wh{\str A} \in \wh\CCC$ the Gaifman graph of $\wh{\str A}$ is a~subgraph of the Gaifman graph of some $\str A \in \CCC$; a quantifier-free $\wh\Sigma$-formula $\widehat\phi(\tup x)$ such that for every $\Sigma$-structure $\A \in \CCC$ there is a~structure $\wh \A \in \wh\CCC$ whose Gaifman graph is contained in the Gaifman graph of $\str A$, such that $\phi_{\str A}=\widehat\phi_{\widehat{\str A}}$.
	Moreover:
	\begin{itemize}
		\item the formula $\wh\phi$ can be computed, given $\phi$, in time $\tower(s + 5,\poly(|\phi|))$,
		\item the structure $\wh {\str A}$ can be computed, given $\phi$ and $\str A$, in time $\tower(s + 5,\poly(|\phi|))\cdot |\str A|$.
	\end{itemize}
\end{lemma}
\begin{proof}
	As in the proof of Lemma~\ref{lem:td-qe}, we assume that $\Sigma$ is relational at the expense of the linear increase in the size of $\varphi$.
	Set $p := |\tup x| + |\tup y|$ (so $p \leq |\varphi|$).
	The results of Sparsity theory \cite{sparsity} imply the following.
	
	\begin{claim*}[{\cite{sparsity}}] 
		\label{cl:low-td-col}
		There is a~set of colors $C$, with $|C| \leq (\tilde\nabla_{2^{\raisebox{1.5pt}{\tiny $p$}}}(\CC))^{2^{\raisebox{1.5pt}{\tiny $\poly(p)$}}}$, and an~algorithm which given $\A \in \CCC$ computes, in time $|C| \cdot |\A|\leq (\tilde\nabla_{2^{\raisebox{1.5pt}{\tiny $p$}}}(\CC))^{2^{\raisebox{1.5pt}{\tiny $\poly(p)$}}} \cdot |\A|$, a~coloring $\chi : \A \to C$, such that for every nonempty $D \subseteq C$ with $|D| \leq p$, the subgraph of the Gaifman graph of $\A$ induced by $\chi^{-1}(D)$ has treedepth at most $|D|$.
	\end{claim*}
	Note that the algorithmic part is not shown in \cite{sparsity}, but it follows from the arguments presented there.
	
	Fix $\str A\in \CCC$.
    For each $c\in C$ let $U_c$ be the unary predicate marking the set $\chi^{-1}(c)$ in $\str A$.
    Then the following equivalence holds in $\str A$:
    \begin{align}\label{eq:disjunction}
        \exists \tup y.\psi(\tup x,\tup y)\quad\Longleftrightarrow\quad \bigvee_{f\from {\tup x\tup y}\to C}\exists \tup y.\left(\psi(\tup x;\tup y)\land\bigwedge_{z\in \tup x\tup y}U_{f(z)}(z)\right).
    \end{align}
    The implication $\Leftarrow$ is immediate.
    For the  implication $\Rightarrow$, suppose 
    $\psi(\tup a,\tup b)$ holds in $\str A$.
    Also, let $f(z)=\chi(\tup a\tup b(z))$ for $z\in\tup x\tup y$.
    By definition,  $\tup a\tup b$ satisfies $\bigwedge_{z\in\tup x\tup y}U_{f(z)}(z)$.
    
    Fix a function $f\from {\tup x\tup y}\to C$ assigning colors to variables in $\tup x\tup y$, and let $D\subset C$ be its range.
    For each $\str A\in\CCC$ let $\str A_f$ denote the substructure of $\str A$ induced by $\chi^{-1}(D)$ together with the unary predicates $U_c=\chi^{-1}(c)$ for $c\in D$.
    Let $\CCC_f=\setof{\str A_f}{\str A\in\CCC}$. Then $\CCC_f$ is a class of relational structures of treedepth at most $|D| \leq p$.
    Denote $$\phi_f(\tup x):=\exists \tup y.\left(\psi(\tup x,\tup y)\land \bigwedge_{z\in \tup x\tup y}U_{f(z)}(z)\right).$$ By Lemma~\ref{lem:td-qe}, 
	there is a quantifier-free formula $\wh\phi_f(\tup x)$, such that for $\str A_f\in\CCC_f$
	we can compute a structure $\wh {\str A_f}$ such that $(\phi_f)_{\str A_f}=(\wh \phi_f)_{\wh {\str A}_f}$.
	The formula $\wh\phi_f$ can be computed from 
	$\phi_f$ in time $\tower(5, \poly(|\varphi_f|, p)) \leq \tower(5, \poly(|\varphi|))$, which also implies that $|\wh\phi_f| \leq \tower(5, \poly(|\varphi|))$. Moreover, the structure $\wh{\str A_f}$ can be computed from $\phi_f$ and $\str A$ in time 
	$\tower(5, \poly(|\varphi|)) \cdot |\str A|$.

	Without loss of generality we assume that the signatures of the classes 
	$\CCC_f$ are pairwise disjoint. 
	Let $\wh {\str A}$ be the structure obtained by superimposing the structures $\wh{\str A}_f$, for all $f\from\tup x\tup y\to C$; in particular, this means that we extend the universe of each $\wh{\str A}_f$ so that $V(\wh{\str A})=V(\wh{\str A}_f)$ for all functions $f$ and we add all relations of each $\wh{\str A}_f$ to the structure $\wh{\str A}$. 
	Let $\wh\phi$ be the disjunction of the formulas $\wh\phi_f$,
	for all $f\from {\tup x\tup y}\to C$.
	By \eqref{eq:disjunction}, we  conclude that $\phi_\str A=\wh \phi_{\wh {\str A}}$, as  required.
	The construction time and the size of $\wh\phi$ are each bounded by
	\begin{align*}
	|C|^p \cdot \tower(5, \poly(|\phi|))
    & \leq \tower(s + 5, \poly(|\phi|)) \cdot \tower(5, \poly(\phi|)) \\& \leq \tower(s + 5, \poly(|\phi|)),
	\end{align*}
	and similarly, $\wh{\str A}$ can be constructed in time $\tower(s + 5, \poly(|\phi|)) \cdot |\str A|$.
\end{proof}

In the general case suppose we are given a~formula $\phi'$ of alternation rank $q'$.
By Lemma~\ref{lem:alternation-rank}, we convert it in time $\poly(|\phi'|)$ to a~batched formula $\phi$ of batched quantifier rank at most $q' + 1$ with $|\phi| \leq \poly(|\phi'|)$.
Then it remains to prove the following statement:
\begin{lemma}
Let $\CCC$ be a class of structures with bounded expansion.
Let $\phi(\tup x)$ be a first-order formula of batched quantifier rank $q$.
There exists a signature $\wh{\Sigma(\phi)}$ consisting of unary relation and function symbols and flags, a class $\wh{\CCC(\phi)}$ of~~$\wh{\Sigma(\phi)}$-structures with bounded expansion, a quantifier-free $\wh{\Sigma(\phi)}$-formula $\widehat\phi(\tup x)$ 
and an algorithm which given $\str A\in\CCC$
outputs in time $\Oof_{\CCC,\phi}(|\str A|)$ a structure $\widehat{\str A(\phi)}\in\widehat{\CCC(\phi)}$ whose Gaifman graph is contained in the Gaifman graph of $\str A$, such that $\phi_{\str A}=\widehat\phi_{\widehat{\str A(\phi)}}.$
Moreover, if $\tilde\nabla_r(\CC) \le \tower(s, \poly(r))$ for all $r \in \N$, then:
\begin{itemize}
	\item the formula $\wh\phi$ can be computed, given $\phi$, in time $\tower((s + 5)q,\poly(|\phi|))$,
	\item  	
	the structure $\wh {\str A(\phi)}$ can be computed, given $\phi$ and $\str A$, in time $\tower((s + 5)q,\poly(|\phi|))\cdot |\str A|$.
\end{itemize}
\end{lemma}
\begin{proof}
	Induction on the structure of $\phi$.
	For $q = 0$, the statement is trivial as we can take $\wh{\phi} = \phi'$ and $\wh{\A(\phi)} = \A$ for all $\A \in \CC$.
	From now on suppose that $q \geq 1$.
	\begin{itemize}
		\item If $\phi = \neg \phi'$, then define inductively $\wh{\phi} = \neg\wh{\phi'}$ and $\wh{\A(\phi)} = \wh{\A(\phi')}$ for all $\A \in \CC$.
		\item If $\phi = \phi_1 \vee \phi_2$, then consider that the inductive signatures $\wh{\Sigma(\phi_1)}$, $\wh{\Sigma(\phi_2)}$ are disjoint. Define $\wh{\Sigma(\phi)} = \wh{\Sigma(\phi_1)} \cup \wh{\Sigma(\phi_2)}$ and $\wh{\phi} = \wh{\phi_1} \vee \wh{\phi_2}$.
		For $\A \in \CC$, define $\wh{\A(\phi)}$ by superimposing the structures $\wh{\A(\phi_1)}$, $\wh{\A(\phi_2)}$.
		\item If $\phi(\tup x) = \exists{\tup y}. \psi(\tup x, \tup y)$, where $\psi$ has batched quantifier rank at most $q - 1$, first apply induction to $\psi$; this yields (in particular) a~formula $\wh{\psi}$ of size at most $\tower((s + 4)(q - 1), \poly(|\psi|))$.
		Then, apply Lemma~\ref{lem:be-existential} to the existential formula $\phi'(\tup x) = \exists{\tup y}.\wh{\psi}(\tup x, \tup y)$; this way we obtain a~formula $\wh{\phi}$ of size at most
		$\tower(4, \poly(|\wh{\psi}|)) \leq \tower((s+4)q, \poly(|\psi|))$.
		Also, for $\A \in \CC$, we define $\wh{\A(\phi)}$ by applying Lemma~\ref{lem:be-existential} to $\wh{\A(\psi)}$.
	\end{itemize}
	The time bounds can be verified straightforwardly with Lemma~\ref{lem:be-existential}.
\end{proof}
This completes the proof of Theorem~\ref{thm:be-qe}.

\end{document}